\RequirePackage{fix-cm}
\documentclass{svjour3}  

\smartqed  
\usepackage{amsmath}

\usepackage{times}
\usepackage{graphicx}
\usepackage{color}

\usepackage{multirow}
\usepackage[toc,page]{appendix}
\usepackage{subcaption}

\usepackage{rotating}
\usepackage{bbm}
\usepackage{latexsym}
\usepackage[dvipsnames]{xcolor}
\usepackage{pgfplots}
\pgfplotsset{compat=1.9}

\usepackage{amssymb,amsmath,amsfonts}
\usepackage{epsf}
\usepackage{wrapfig}
\usepackage[normalem]{ulem}
\usepackage{float}
\usepackage{fixmath}
\usepackage{hyperref}
\usepackage{tikz}
\usepackage{tabularx}
\usepackage[flushleft]{threeparttable}

\usepackage{rotating}
\usepackage{bbm}
\usepackage{enumitem}
\usepackage{latexsym}
\usepackage{pgfplots}
\pgfplotsset{compat=1.9}

\renewcommand{\oe}{\mathcal{O}\left(\epsilon\right)}
\newcommand{\oee}{\mathcal{O}\left(\epsilon^2\right)}
\newcommand{\E}{\mathbb{E}}

\newcommand{\subsubsubsection}[1]{{\flushleft \bf #1}}

\newtheorem{innercustomgeneric}{\customgenericname}
\providecommand{\customgenericname}{}
\newcommand{\newcustomtheorem}[2]{%
  \newenvironment{#1}[1]
  {%
   \renewcommand\customgenericname{#2}%
   \renewcommand\theinnercustomgeneric{##1}%
   \innercustomgeneric
  }
  {\endinnercustomgeneric}
}

\newcustomtheorem{customthm}{Theorem}
\newcustomtheorem{customlemma}{Lemma}

\newcommand{\defn}{\stackrel{\Delta}{=}} 

\newcommand{\R}{\mathbb{R}}

\renewcommand{\th}{^{\mbox{\footnotesize th}}}

\newcommand{\Na}{Na$^+$} 
\newcommand{\K}{K$^+$}

\newcommand{\mref}{M_\text{ref}}
\newcommand{\nref}{N_\text{ref}}
\newcommand{\mtot}{M_\text{tot}}
\newcommand{\ntot}{N_\text{tot}} 

\renewcommand{\Pr}{\mathbb{P}}

\newcommand{\mbe}{\mathbf{e}}

\newcommand{\mbQ}{\mathbf{Q}}
\newcommand{\mbr}{\mathbf{r}}
\newcommand{\mbx}{\mathbf{x}}
\newcommand{\mby}{\mathbf{y}}
\newcommand{\mbX}{\mathbf{X}}

\newcommand{\LL}{\mathcal{L}}
\def\given{\:|\:}

\newcommand{\mbN}{\mathbf{N}}
\newcommand{\mbn}{\mathbf{n}}

\newcommand{\mbdW}{\mathbf{dW}}
\newcommand{\mbW}{\mathbf{W}}
\newcommand{\mbF}{\mathbf{F}}

\newcommand{\mbZ}{\mathbf{Z}}
\newcommand{\mbM}{\mathbf{M}}
\newcommand{\mbm}{\mathbf{m}}

\definecolor{blue}{rgb}{0,0,1}
\definecolor{darkgreen}{rgb}{0,.5,0}
\definecolor{darkred}{rgb}{.75,0,0}
\definecolor{red}{rgb}{1,0,0}

\newcommand{\Ld}{\mathcal{L}^\dagger}
\newcommand{\bigzero}{\mbox{\normalfont\Large\bfseries 0}}
\newcommand{\tbar}{\overline{T}}

\begin{document}

\title{Resolving Molecular Contributions of Ion Channel Noise to  Interspike Interval Variability through Stochastic Shielding
}

\titlerunning{Stochastic Shielding and ISI Variability}        

\author{Shusen Pu$^{\displaystyle 1}$  \and Peter J.~Thomas$^{\displaystyle 1- \displaystyle 4}$}


\institute{Shusen Pu\at
            $^{\displaystyle 1}$ Case Western Reserve University Department of \\
 Mathematics, Applied Mathematics, and Statistics\\
              \email{sxp600@case.edu}           
           \and
           Peter J.~Thomas$^{\displaystyle 1- \displaystyle 5}$  \at
           Case Western Reserve University Departments of \\
$^{\displaystyle 1}$ Mathematics, Applied Mathematics, and Statistics\\
$^{\displaystyle 2}$ Biology, $^{\displaystyle 3}$ Cognitive Science, $^{\displaystyle 4}$ Data and Computer Science, and $^{\displaystyle 5}$ Electrical, Computer, and Systems Engineering\\
\email{pjthomas@case.edu}
}

\date{Received: Oct, 2020}

\maketitle

\begin{abstract}
Molecular fluctuations can lead to macroscopically observable effects.  
The random gating of ion channels in the membrane of a nerve cell provides an important example.
The contributions of independent noise sources to the variability of action potential timing has not previously been studied at the level of molecular transitions within a conductance-based model ion-state graph.  
Here we study a stochastic Langevin model for the Hodgkin-Huxley (HH) system based on a detailed representation of the underlying channel-state Markov process, the ``$14\times 28$D model'' introduced in (Pu and Thomas 2020, Neural Computation).
We show how to resolve the individual contributions that each transition in the ion channel graph makes to the variance of the interspike interval (ISI).
We extend the mean--return-time (MRT) phase reduction developed in (Cao et al.~2020, SIAM J.~Appl.~Math) to the second moment of the return time from an MRT isochron to itself.  
Because fixed-voltage spike-detection triggers do not correspond to MRT isochrons, the \emph{inter-phase interval} (IPI) variance only approximates the ISI variance.  
We find the IPI variance and ISI variance agree to within a few percent when both can be computed.
Moreover, we prove rigorously, and show numerically, that our expression for the IPI variance is accurate in the small noise (large system size) regime;
our theory is exact in the limit of small noise.   
%
By selectively including the noise associated with only those few transitions responsible for most of the ISI variance, our analysis extends the stochastic shielding (SS) paradigm (Schmandt and Gal\'{a}n 2012, Phys.~Rev.~Lett.) from the stationary voltage-clamp case to the current-clamp case. We show numerically that the SS approximation has a high degree of accuracy even for larger, physiologically relevant noise levels.
Finally, we demonstrate that the ISI variance is not an unambiguously defined quantity, but depends on the choice of voltage level set as the spike-detection threshold.  
We find a small but significant increase in ISI variance, the higher the spike detection voltage, both for simulated stochastic HH data and for voltage traces recorded in \textit{in vitro} experiments.
In contrast, the IPI variance is invariant with respect to the choice of isochron used as a trigger for counting ``spikes".

\end{abstract}




\section{Introduction}
Nerve cells communicate with one another, process sensory information, and control motor systems through transient voltage pulses, or spikes.
At the single-cell level, neurons exhibit a combination of deterministic and stochastic behaviors.  
In the supra-threshold regime, the regular firing of action potentials under steady current drive suggests limit cycle dynamics, with the precise timing of voltage spikes perturbed by noise. 
Variability of action potential timing persists even under blockade of synaptic connections, consistent with an intrinsically noisy neural dynamics arising from the random gating of ion channel populations, or ``channel noise" \cite{White2000Elsevier}. 

Understanding the molecular origins of spike time variability may shed light on several phenomena in which channel noise plays a role.
For example, microscopic noise can give rise to a stochastic resonance behavior \cite{SchmidGoychukHanggi2001EPL}, and can contribute to cellular- and systems-level timing changes in the integrative properties of neurons \cite{Dorval2005JNeuro}.
Jitter in spike times under steady drive may be observed in neurons from the auditory system \cite{GerstnerKempterVanHemmenWagner1996Nature,Goldwyn2010JCompNeu,MoiseffKonishi1981JCompPhys} as well as in the cerebral cortex \cite{Mainen1995Science} 
and may play a role in both fidelity of sensory
information processing and in precision of motor control \cite{SchneidmanFreedmanSegev1998NECO}.
\begin{figure}
  \centering
  \includegraphics[width=1\textwidth]{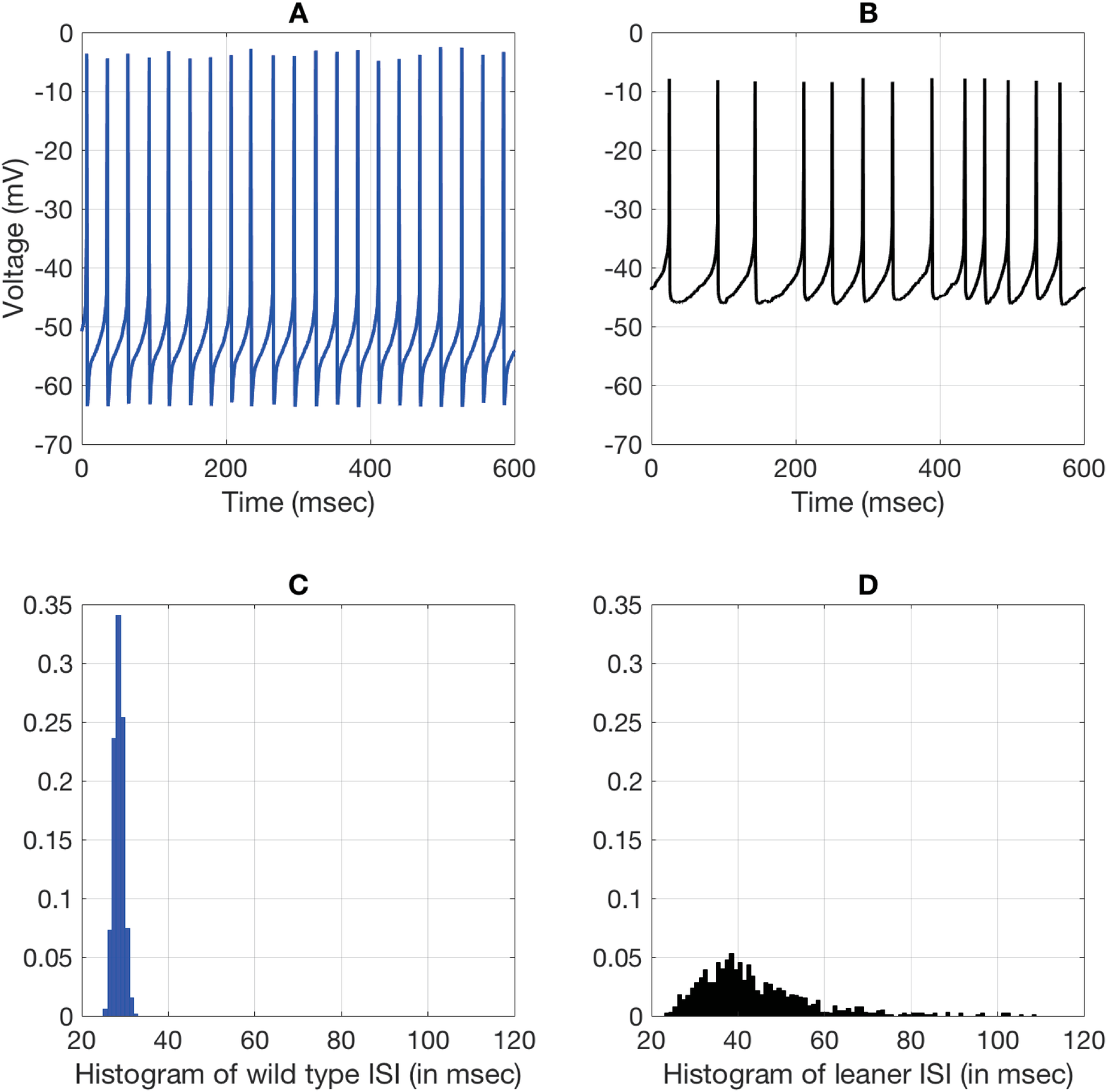} \\
  \caption{Somatic voltage recorded in vitro from intact Purkinje cells (cerebellar slice preparation) during spontaneous tonic firing, with synaptic input blocked. 
  Recordings courtesy of Dr.~D.~Friel. See \S\ref{sec:experimental-methods} for details.
 \textbf{A}: Sample voltage recordings from a wild type Purkinje cell showing precise spontaneous firing with interspike interval (ISI) coefficient of variation (CV=standard deviation / mean ISI) of approximately 3.9\%. 
 \textbf{B}: Sample recordings from Purkinje cells with leaner mutation in P/Q-type calcium channels showing twice the variability in ISI (CV c.~30\%). \textbf{C, D}: histogram of ISI for wild type and leaner mutation, respectively. Bin width = 1 msec for each. }\label{fig:WTLE_stats} 
\end{figure}
As a motivating example for this work, channel noise is thought to underlie jitter in spike timing observed in cerebellar Purkinje cells recorded in vitro from the ``leaner mouse", a P/Q-type calcium channel mutant with profound ataxia \cite{Walter2006NatNeuro}.
Purkinje cells fire \Na action potentials spontaneously \cite{Llinas1980JPH2,Llinas1980JPH1}, and may do so at a very regular rate \cite{Walter2006NatNeuro}, even in the absence of  synaptic input (cf.~Fig.~\ref{fig:WTLE_stats}~A and C).
Mutations in an homologous human calcium channel gene are associated with episodic ataxia type II, a debilitating form of dyskinesia \cite{Pietrobon2010Springer,Rajakulendran2012Nature}.
Previous work has shown that the {\it{leaner}} mutation increases the variability of spontaneous action potential firing in Purkinje cells \cite{OvsepianFriel2008EJN,Walter2006NatNeuro}  (Fig.~\ref{fig:WTLE_stats}~B and D).
It has been proposed that increased channel noise akin to that observed in the leaner mutant plays a mechanistic role in this human disease \cite{Walter2006NatNeuro}. 

Despite its practical importance, a quantitative understanding of distinct molecular sources of  macroscopic timing variability remains elusive.  
Significant theoretical attention has been paid to the variance of phase response curves and interspike interval (ISI) variability. 
Most analytical studies are based on the integrate-and-fire model \cite{Brunel2003NCom,Lindner2004PRE_interspike,VilelaLindner2009}, except \cite{ErmentroutBeverlinTroyerNetoff2011JCNS}, which perturbs the voltage of a conductance-based model with a white noise current rather than through biophysically-based channel noise.
Standard models of stochastic ion channel kinetics  comprise hybrid stochastic systems.  As illustrated in Fig.~\ref{plot:HHNaKgates}, the membrane
potential evolves deterministically, given the transmembrane currents; the
currents are determined by the ion channel state; the ion channel states
fluctuate stochastically with opening and closing rates that depend on the
voltage \cite{AndersonErmentroutThomas2015JCNS,Bressloff2014APS,Buckwar2011JMB,Pakdaman2010CUP}.
This closed-loop nonlinear dynamical stochastic system is difficult to
study analytically, because of recurrent statistical dependencies of the
variables one with another.  
An important and well studied special case is
fixed-voltage clamp, which reduces the ion channel subsystem to a time invariant Markov process \cite{SchmidtThomas2014JMN}.  
Under the more natural current clamp, the ion channel dynamics
taken alone are no longer Markovian, as they intertwine with current and voltage. 
A priori, it is challenging to draw a direct connection between the variability of spike timing and molecular-level stochastic events, such as the opening and closing of specific ion channels, as spike timing is a pathwise property reflecting the effects of fluctuations accumulated around each orbit or cycle.

\begin{figure}
    \centering
    \begin{tikzpicture}
 \foreach \x/\w/\y/\z/\num/\a/\p in {1/ /3/4/3/1/00,2/2/7/8/2/2/10,3/3/11/12/ /3/20}
 {
 \draw (2.5*\x,2.5) circle (13pt); 
 \draw [color=black,->] (2.5*\x+0.8,0.18+2.5) -- (2.5*\x+1.6,0.18+2.5) node[right, color=blue]{\tiny \y};
 \node at (2.5*\x+1.15,12*0.036+2.5) {\num$\alpha_m$};
 \draw [color=black,<-] (2.5*\x+0.8,-0.18+2.5) node[left, color=blue]{\tiny \z} -- (2.5*\x+1.6,-0.18+2.5);
 \node at (2.5*\x+1.15,-12*0.036+2.5) {\w$\beta_m$};
 \node at (2.5*\x,0.72+2.5)  {$M_{\p}$};
 }
 \draw  (2.5*4,2.5) circle (13pt);
 \node  at (2.5*4,0.72+2.5) {$M_{30}$}; 
  \foreach \x/\w/\y/\z/\num/\p/\q in {1/ /15/16/3/01/5,2/2/17/18/2/11/6,3/3/19/20/ /21/7}
 {
 \draw (2.5*\x,0) circle (13pt); 
 \draw [color=black,->] (2.5*\x+0.8,0.18) -- (2.5*\x+1.6,0.18) node[right, color=blue]{\tiny \y};
 \node at (2.5*\x+1.15,12*0.036) {\num$\alpha_m$};
 \draw [color=black,<-] (2.5*\x+0.8,-0.18) node[left, color=blue]{\tiny \z} -- (2.5*\x+1.6,-0.18);
 \node at (2.5*\x+1.15,-12*0.036) {\w$\beta_m$};
 \node at (2.5*\x,-0.72) {$M_{\p}$} node  at (2.5*2.5,2.5+1.25) {A: $\text{Na}^+\  \text{Channel}$};
 }
 \draw [fill=green!30] (2.5*4,0) circle (13pt);
 \node at (2.5*4,-0.72) {$M_{31}$} ;
 \foreach \x/\a/\b in {1/2/1,2/6/5,3/10/9,4/14/13}
 {
  \draw [->] (2.5*\x-0.18,0.73) -- (2.5*\x-0.18,0.7+1) node[above, color=blue]{\tiny \a};
  \node at (2.5*\x-12*0.036,1.2) {$\beta_h$};
  \draw [->] (2.5*\x+0.18,0.7+1) -- (2.5*\x+0.18,0.73) node[below, color=blue]{\tiny \b};
  \node at (2.5*\x+12*0.036,1.2) {$\alpha_h$};
 }
 \end{tikzpicture}
 
  \begin{tikzpicture}
 \foreach \x/\w/\y/\z/\num/\id in {1/ /1/2/4/0,2/2/3/4/3/1,3/3/5/6/2/2,4/4/7/8/ /3}
 {
 \draw (1.9*\x,0) circle (12pt);
 \draw [color=black,->] (1.9*\x+0.7,5pt) -- (1.9*\x+1.2,5pt) node[right, color=blue]{\tiny \y};
 \node at (1.9*\x+0.85,12pt) {\num$\alpha_n$};
 \draw [color=black,<-] (1.9*\x+0.7,-5pt) node[left, color=blue]{\tiny \z} -- (1.9*\x+1.2,-5pt);
 \node at (1.9*\x+0.85,-12pt) {\w$\beta_n$};
 \node at (1.9*\x,-20pt) {$N_\id$};
 }
 \draw [fill=green!30] (1.9*5,0) circle (12pt);
  \node  at (1.9*5,-20pt) {$N_4$} node  at (1.9*3,1) {B: $\text{K}^+\  \text{Channel}$};
 \end{tikzpicture}
 \begin{tikzpicture}
  \node at (1.9*3,1) {C: Interdependent Variables Under Current Clamp};
 \end{tikzpicture}
     \includegraphics[width=4.5in]{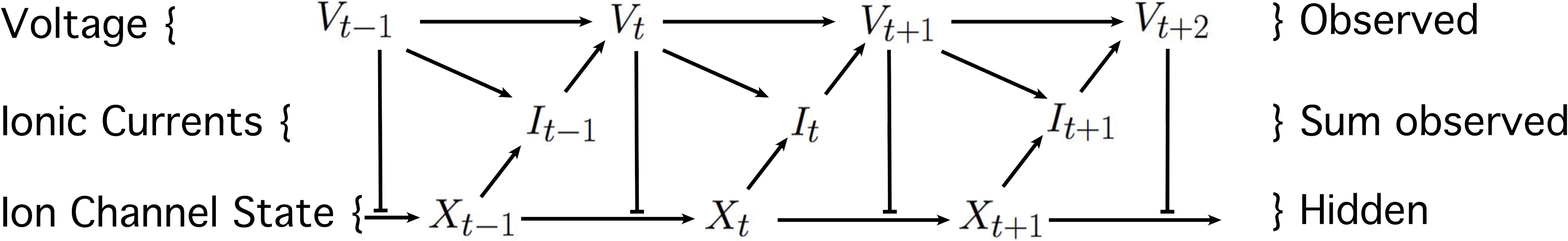}
\vspace{0.5cm}
\caption{Statistical dependencies among voltage ($V_t$), ionic currents ($I_t$), and ion channel state ($X_t$) form a hybrid, or piecewise deterministic, stochastic model. 
\textbf{A, B:} Molecular sodium (\Na) and potassium (\K)  channel states for the Hodgkin-Huxley model.  Filled circles mark conducting states $N_4$ and $M_{31}$. Per capita transition rates for each directed edge ($\alpha_n$, $\beta_n$, $\alpha_m$, $\beta_m$, $\alpha_h$ and $\beta_h$) are voltage dependent (cf.~App.~\ref{app:notations}). Directed edges are numbered 1-8 (\K~channel) and 1-20 (\Na-channel), marked in small red numerals.
\textbf{C:} 
Voltage $V_t$ at time $t$ influences current $I_t$ as well as the transition from channel state $X=[M_{00},M_{10},\ldots,M_{31},N_{0},N_{1},\ldots,N_{4}]$,
at time $t-1$ to time $t$.  
(For illustration we assume discrete sampling at integer multiples of a nominal time step $\Delta t$ in arbitrary units.) 
Channel state dictates the subsequent ionic current, which dictates the voltage increment.  
Arrowheads ($\to$) denote  deterministic dependencies.  
T-connectives ($\perp$) denote statistical dependencies. }
\label{plot:HHNaKgates}
\end{figure}

In \cite{SchmandtGalan2012PRL} Schmandt and Gal\'{a}n introduced \emph{stochastic shielding} as a fast, accurate approximation scheme for stochastic ion channel models.  
Rather than simplifying the Markov process by aggregating ion channel states, stochastic shielding reduces the complexity of the underlying sample space by eliminating independent noise sources (corresponding to individual directed edges in the channel state transition graph) that make minimal contributions to ion channel state fluctuations.
In addition to providing an efficient numerical procedure, stochastic shielding leads to an \emph{edge importance measure} \cite{SchmidtThomas2014JMN} that quantifies the contribution of the fluctuations arising along each directed edge to the variance of channel state occupancy (and hence the variance of the transmembrane current).
The stochastic shielding method then amounts to simulating a stochastic conductance-based model using only the noise terms from the most important transitions.
While the original, heuristic implementation of stochastic shielding considered both current and voltage clamp scenarios \cite{SchmandtGalan2012PRL}, subsequent mathematical analysis of stochastic shielding considered only the constant voltage-clamp case \cite{SchmidtGalanThomas2018PLoSCB,SchmidtThomas2014JMN}.

In our previous work \cite{PuThomas2020NECO}, we numerically estimated the contribution of each directed edge in the transition graph (Fig.~\ref{plot:HHNaKgates}A,B) to the variance of ISIs.
In this paper we provide, to our knowledge, the first analytical treatment of
the variability of spike timing under current clamp arising from the random
gating of ion channels with realistic (Hodgkin-Huxley) kinetics.  
Building on prior work  \cite{CaoLindnerThomas2020SIAP,PuThomas2020NECO,SchmandtGalan2012PRL,SchmidtGalanThomas2018PLoSCB,SchmidtThomas2014JMN}, we study the variance of the
transition times among a family of Poincar\'{e} sections, the mean--return-time (MRT)
isochrons investigated by \cite{SchwabedalPikovsky2010PRE,CaoLindnerThomas2020SIAP} that extend the notion of phase
reduction to stochastic limit cycle oscillators. 
We prove a theorem that gives the form of the variance, $\sigma^2_\phi$, of inter-phase-intervals (IPI)\footnote{Equivalently, ``iso-phase-intervals'': the time taken to complete one full oscillation, from a given isochron back to the same isochron.}  in the limit of small noise (equivalently, large channel number or system size),
as a sum of contributions $\sigma^2_{\phi,k}$
from each directed edge $k$  in the ion channel
state transition graph (Fig.~\ref{plot:HHNaKgates}A,B).
The IPI variability involves several quantities: the per capita transition rates $\alpha_k$ along each transition, the mean-field ion channel population $M_{i(k)}$ at the source node for each transition, the stoichiometry (state-change) vector $\zeta_k$ for the $k$th transition, and the phase response curve $\mbZ$ of the underlying limit cycle:
$$\sigma^2_\phi=\sum_{k\in\text{all edges}}\sigma^2_{\phi,k}=\epsilon \tbar_0 \sum_k\E\left(   \alpha_k(v(t))M_{i(k)}(t)\left(\zeta_k^\intercal \mbZ(t)  \right)^2 \,dt\right)+ O\left(\epsilon^2\right),
$$
in the limit as $\epsilon\to 0^+$.  
Here $\overline{T}_0$, $v(t)$ and $\mbM(t)$ are the period, voltage, and ion channel population vector of the deterministic limit cycle for $\epsilon=0$.  By $\E$ we denote expectation with respect to the stationary probability density
for the Langiven model (cf.~eqn.~\eqref{eq:langevin-rescaled}).
As detailed below, we scale $\epsilon\propto1/\sqrt{\Omega}$ where the system size $\Omega$ reflects the size of the underlying ion channel populations.

Thus we are able to pull apart the distinct contribution of each independent
source of noise (each directed edge in the ion channel state transition graphs) to the variability of timing. 
Figs.~\ref{fig:ISI_Na_K_all}-\ref{fig:all_SS}  illustrate the additivity of contributions from separate edges for small noise levels.
As a consequence of this linear decomposition, we can extend the stochastic shielding
approximation, introduced in \cite{SchmandtGalan2012PRL} and rigorously analyzed under
voltage clamp in \cite{SchmidtThomas2014JMN,SchmidtGalanThomas2018PLoSCB}, to the current clamp
case.  
Our theoretical result guarantees that, for small noise, we can replace a full stochastic simulation with a more efficient simulation driven by noise from only the most ``important" transitions with negligible loss of accuracy.
We find numerically that the range of validity of the stochastic shielding approximation under current clamp extends beyond the ``small noise limit" to include physiologically relevant population sizes, cf.~Fig.~\ref{fig:all_SS}.
 
The inter-phase-interval (IPI) is a mathematical construct closely related to, but distinct from, the inter-spike-interval (ISI).
The ISI, determined by the times at which the cell voltage moves upward (say) through a prescribed voltage threshold $v_\text{thresh}$, is directly observable from experimental recordings -- unlike the IPI.
However, we note that both in experimental data and in  stochastic numerical simulations, the variance of the ISI is not insensitive to the choice of voltage threshold, but increases monotonically as a function of $v_\text{thresh}$ (cf.~Fig.~\ref{fig:ISI_change}).  
In contrast, the variance of inter-\emph{phase}-interval times is the same, regardless of which MRT isochron is used to define the intervals.  
This invariance property gives additional motivation for investigating the variance of the IPI.

The structure of the paper is as follows: 
In \S\ref{sec:defs}, we review the $14\times 28$D Langevin Hodgkin-Huxley model proposed in \cite{PuThomas2020NECO}, and provide mathematical definitions of   first passage times, interspike intervals, asymptotic phase functions, and iso-phase intervals for the class of model we consider.
In \S\ref{sec:noise_decomp} we state the necessary assumptions and prove the small-noise decomposition theorem. 
In \S\ref{sec:Numerical} we compare the contributions of individual transitions to both interspike interval variability and interphase interval variability, predicted from the decomposition theorem, against the results of numerical simulations.
Section \S\ref{sec:discussion} discusses the theoretical and practical limitations of our results.

\section{Definitions, Notation and Terminology}\label{sec:defs}
In this section, we recall the notation and assumptions of the $14\times 28$D Langevin model for the stochastic Hodgkin-Huxley system introduced in \cite{PuThomas2020NECO}. In addition, we will present definitions, notations and terminology that are necessary for the main result.
We adopt the standard convention that uppercase symbols (e.g.~$V,\mbM,\mbN$) represent random variables, while corresponding lowercase symbols (e.g.~$v,\mbm,\mbn$) represent possible realizations of the random variables. Thus $\Pr(V\le v)$ is the probability that the random voltage $V$ does not exceed the value $v$.
We set vectors in \textbf{bold font} and scalars in standard font.

\subsection{The Langevin HH Model} \label{subsec:14Dmodel}
For the HH kinetic scheme given in Fig.~\ref{plot:HHNaKgates}A-B  (p.~\pageref{plot:HHNaKgates}), we define the eight-component state vector $\mbM$ for the \Na~gates, and the five-component state vector  $\mbN$ for the \K~gates, respectively, as
\begin{align}\label{eq:define_M}
\mbM&=[M_{00},M_{10},M_{20},M_{30},M_{01},M_{11},M_{21},M_{31}]^\intercal \in [0,1]^{8} \\
\label{eq:define_N}
\mbN&=[N_0,N_1,N_2,N_3,N_4]^\intercal\in [0,1]^5,
\end{align}
where $\sum_{i=0}^3\sum_{j=0}^1 M_{ij}=1$ and $\sum_{i=0}^4 N_i=1$.
The open probability for the \Na~channel is  $M_{31}$, and is $N_4$ for the \K~channel. 
Our previous paper \cite{PuThomas2020NECO} proposed a $14\times 28$D Langevin HH model. 
Here, we make the dependence of the channel noise on system size (number of channels) explicit, by introducing a small parameter $\epsilon\propto N_\text{ion}^{-1}$.  
\label{page:epsilon_propto_N}
We therefore consider a one-parameter family of Langevin equations
\begin{equation}\label{eq:langevin-rescaled}
  d\mbX=\mbF(\mbX)\,dt+{\sqrt{\epsilon}}\mathcal{G}(\mbX)\,d\mbW(t)  
\end{equation}
where we define the 14-component vector $\mbX=(V;\mbM;\mbN)$ and $\mbdW(t)$ represents a Wiener (Brownian motion) process.
In the governing Langevin equation \eqref{eq:langevin-rescaled}, the stochastic forcing components in $\mathcal{G}\,d\mbW$ are implicitly scaled by factors proportional to $\sqrt{\epsilon}$, with effective numbers of $\mtot=\mref/\epsilon$ sodium and $\ntot=\nref/\epsilon$ potassium channels.
For comparison, in their study of different Langevin models, Goldwyn and Shea-Brown considered a patch of excitable membrane containing $\mref=6000$ sodium channels and $\nref=1800$ potassium channels \cite{GoldwynSheaBrown2011PLoSComputBiol}.   

The deterministic part of the evolution equation $\mbF(\mbX)\,=\left[\frac{dV}{dt};\frac{d\mbM}{dt};\frac{d\mbN}{dt}\right]$ is the same as the mean-field dynamics, given by
\begin{eqnarray}
C\frac{dV}{dt}&=&-\bar{g}_{\text{Na}}M_8(V-V_{\text{Na}})-\bar{g}_{\text{K}}N_5(V-V_\text{K})-g_\text{L}(V-V_\text{L})+I_\text{app}, \label{eq:RTR_dM}\\
\frac{d\mbM}{dt}&=&A_\text{Na}(V)\mbM, \label{14dhh2}\\
\frac{d\mbN}{dt}&=&A_\text{K}(V)\mbN.\label{14dhh3}
\end{eqnarray}
Here, $C$ ($\mu F/cm^2$) is the capacitance, $I_\text{app}$ ($nA/cm^2$) is the applied current, the maximal conductance is $\bar{g}_\text{ion}$ ($mS/cm^2$),  $V_\text{ion}$ ($mV$) is the associated reversal potential, $\text{ for ion }\in\{\text{Na},\text{K}\}$, and the ohmic leak current is $g_\text{leak}(V-V_\text{leak})$. 
The voltage-dependent drift matrices $A_\text{Na}$ and $A_\text{K}$  given by
 {\footnotesize{\begin{equation}
 A_\text{Na}(V) =\begin{bmatrix}
A_\text{Na}(1) & \beta_m&0 &0 &\beta_h&0&0&0\\
3\alpha_m&A_\text{Na}(2)&2\beta_m&0&0&\beta_h&0&0\\
0&2\alpha_m&A_\text{Na}(3) &3\beta_m&0&0&\beta_h&0 \\
0&0&\alpha_m&A_\text{Na}(4)&0&0&0&\beta_h \\
\alpha_h&0&0&0&A_\text{Na}(5)&\beta_m&0&0\\
0&\alpha_h&0&0&3\alpha_m&A_\text{Na}(6)&2\beta_m&0\\
0&0&\alpha_h&0&0&2\alpha_m&A_\text{Na}(7)&3\beta_m\\
0&0&0&\alpha_h&0&0&\alpha_m&A_\text{Na}(8)\\
\end{bmatrix},
\label{matrix:ANa}
\end{equation}

\begin{equation}
   A_\text{K}(V) =\begin{bmatrix}
   A_\text{K}(1)& \beta_n(V)             & 0                & 0                  & 0\\
   4\alpha_n(V)& A_\text{K}(2)&   2\beta_n(V)              & 0&                   0\\
    0&        3\alpha_n(V)&        A_\text{K}(3)& 3\beta_n(V)&          0\\
    0&        0&               2\alpha_n(V)&          A_\text{K}(4)& 4\beta_n(V)\\
    0&        0&               0&                 \alpha_n(V)&          A_\text{K}(5)
\end{bmatrix},
\label{Matrix:AK}
\end{equation}}}
with diagonal elements $$A_\text{ion}(i)=-\sum_{j\::\:j\neq i}A_\text{ion}(j,i),\text{ for ion }\in\{\text{Na},\text{K}\}.$$
The state-dependent noise coefficient matrix $\mathcal{G}$ is $14\times28$ and can be written as 
\[
\mathcal{G} =\left(\begin{array}{@{}c|c@{}}
  \bigzero_{1\times 20} & 
  \bigzero_{1\times 8}  \\
\hline
  \begin{matrix}
  S_\text{Na}
  \end{matrix} &\bigzero_{8\times 8} \\
  \hline
  \bigzero_{5\times 20} & \begin{matrix}
  S_\text{K}
  \end{matrix}
\end{array}\right).
\]
When simulating \eqref{eq:langevin-rescaled} we use free boundary conditions for the gating variables $M_{ij}$ and $N_i$ \cite{OrioSoudry2012PLoS1,Pezo2014Frontiers,PuThomas2020NECO}.  
With free boundaries, some gating variables may make small, rare  excursions into negative values. 
To avoid inconsistencies we therefore use the absolute values $|M_{ij}|$ and $|N_i|$ when calculating the edge fluxes needed to construct the matrix $\mathcal{G}$.  
The resulting boundary effects are insignificant for all system sizes considered \cite{OrioSoudry2012PLoS1}.   

All parameters, transition rates, and the coefficient matrices $S_\text{K}$  and $S_\text{Na}$ are given in Appendix \ref{app:notations}.

\subsection{Stochastic Shielding}
The stochastic shielding (SS) approximation was first introduced by
Schmandt and Gal\'{a}n  as an efficient numerical procedure to simulate  Markov processes using only those transitions associated with observable states  \cite{SchmandtGalan2012PRL}.
Analysis of the SS approximation leads to an \emph{edge importance measure} \cite{SchmidtThomas2014JMN} that quantifies the contribution of the fluctuations arising along each directed edge to the variance of channel state occupancy (and hence the variance of the transmembrane current) under voltage clamp.
The stochastic shielding method then amounts to simulating a stochastic conductance-based model using only the noise terms from the most important transitions.
While the original, heuristic implementation of stochastic shielding considered both current and voltage clamp scenarios \cite{SchmandtGalan2012PRL}, subsequent mathematical analysis of stochastic shielding considered only the constant voltage-clamp case \cite{SchmidtGalanThomas2018PLoSCB,SchmidtThomas2014JMN}.

In our previous work \cite{PuThomas2020NECO}, we extended the SS approximation to the current clamp case, where we numerically calculated the edge importance for all transitions in Fig.~\ref{plot:HHNaKgates}. 
Given the matrix $\mathcal{G}$ and a list of the ``most important" noise sources (columns of $\mathcal{G}$) the stochastic shielding approximation amounts to setting the columns excluded from the list equal to zero \cite{PuThomas2020NECO,SchmidtThomas2014JMN}.
Within the framework of stochastic shielding, we may ask how each column of $S_\text{K}$ and $S_\text{Na}$ contribute to the variability of stochastic trajectories generated by eq.~\eqref{eq:langevin-rescaled}.

In this paper, our main theorem gives a semi-analytical foundation for the edge-importance measure under current clamp in terms of contributions to ISI variance. 
In \S\ref{sec:Numerical}, we will apply the SS method to numerically test our theorem of the contribution from each edge to the ISI variability under current clamp. 

The next section defines first passage times and interspike intervals for general conductance-based models, which are fundamental to our subsequent analysis.

\subsection{First passage times and interspike intervals}

Reversal potentials $V_\text{ion}$ for physiological ions are typically confined to the  range $\pm 150 $mV. 
For the 4-D and the 14-D HH models, the reversal potentials for \K~and \Na~are -77mv and +50mv respectively \cite{ErmentroutTerman2010book}.  In Lemma \ref{Lemma:vmin_vmax}, we establish that the voltage for conductance-based model in eqn.~\eqref{eq:langevin-rescaled} is bounded.  
Therefore we can find a voltage range $[v_\text{min},v_\text{max}]$ that is forward invariant with probability 1, meaning that the probability of any sample path leaving the range $v_\text{min}\le V(t)\le v_\text{max}$ is zero.  
At the same time, the channel state distribution for any channel with $k$ states  is restricted to a $(k-1)$-dimensional simplex $\Delta^{k-1}\subset \R^k$, given by $y_1+\ldots+y_{k}=1, y_i\ge 0.$  
Therefore, the phase space of any conductance-based model of the form \eqref{eq:langevin-rescaled} may be restricted to a compact domain in finite dimensions.  
\begin{definition}\label{dfn:HHdomain}
We define \emph{the HH domain} $\mathcal{D}$ to be
\begin{equation}
 \mathcal{D}\defn [v_\text{min},v_\text{max}]\times\Delta^7\times \Delta^4
 \end{equation}
where $\Delta^7$ is the simplex occupied by the \Na channel states, and $\Delta^4$ is occupied by the \K channel states.
\end{definition}
  We thus represent the ``14-D" HH model in a reduced phase space of dimension 1+7+4=12.  

\begin{lemma} \label{Lemma:vmin_vmax}
For a conductance-based model of the form \eqref{eq:langevin-rescaled}, and for any fixed applied current $I_\text{app}$, there exist  upper and lower bounds $v_\text{max}$ and $v_\text{min}$ such that trajectories with initial voltage condition $v\in[v_\text{min},v_\text{max}]$ remain within this interval for all times $t>0$, with probability 1, regardless of the initial channel state, provided the gating variables satisfy $0\le M_{ij}\le 1$ and $0\le N_i\le 1$.
\end{lemma}
\begin{proof}
See App.~\ref{append_Lemma_vmin_vmax}. 
\end{proof}

\begin{remark}
\label{rem:alpha-is-bounded}
Lemma \ref{Lemma:vmin_vmax} implies that the per capita transition rates along a finite collection of edges, $\{\alpha_k(v)\}_{k=1}^K$ are bounded above by a constant $\alpha_\text{max}$, as $v$ ranges over $v_\text{min}\le v \le v_\text{max}$.  This fact will help establish Theorem \ref{theorem:noise}.
\end{remark}

\subsubsubsection{Interspike Intervals and First Passage Times}
\begin{figure}
    \centering
\includegraphics[width=0.8\textwidth]{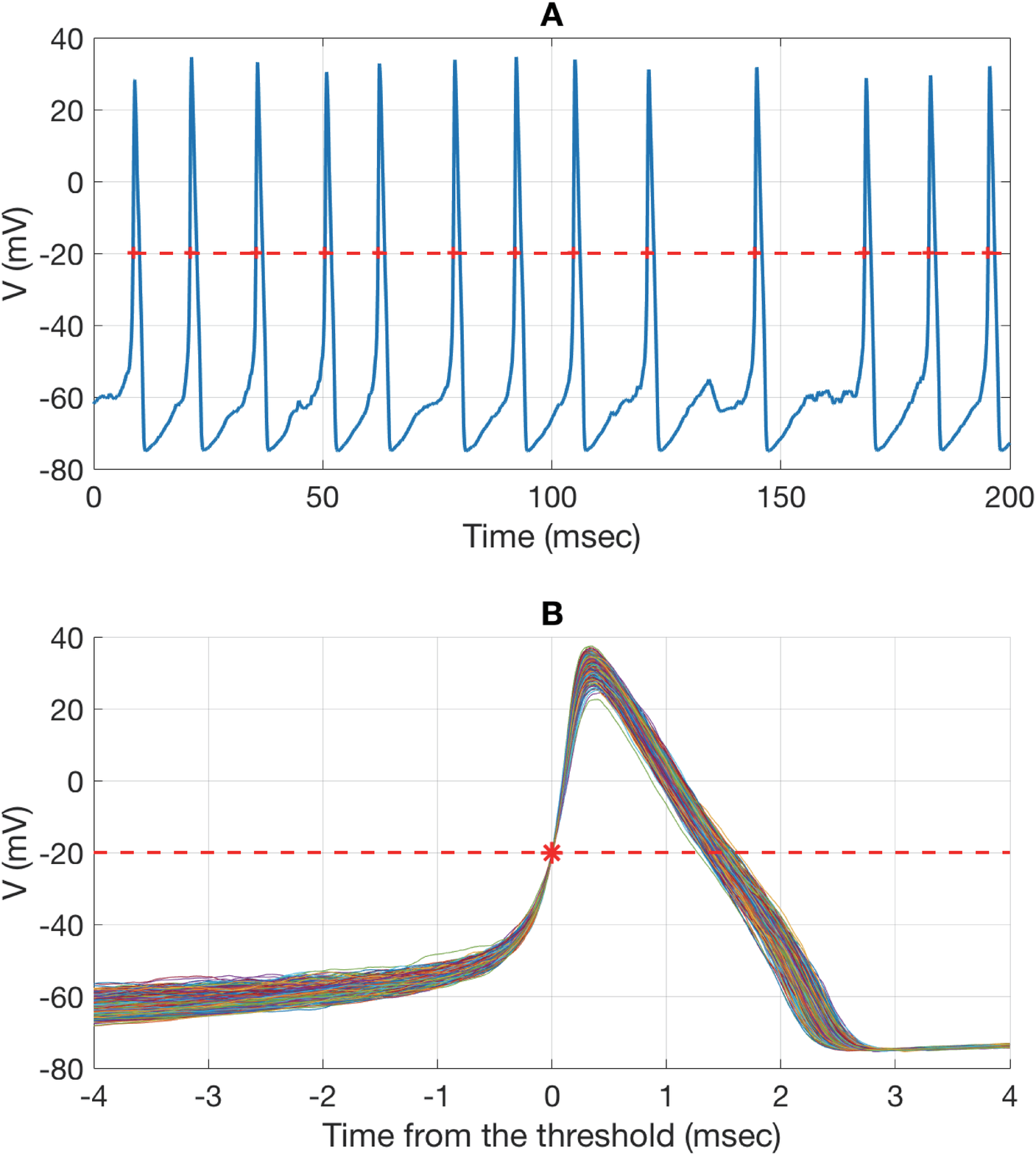}
    \caption{ {Example voltage trace for the 14-D stochastic HH model (eqn~\eqref{eq:langevin-rescaled}). \textbf{(A)} Voltage trace generated by the full 14-D stochastic HH model with $I_\text{app}$=-10 nA; for other parameters see \S\ref{app:notations}.  \textbf{(B)} Ensemble of voltage traces constructed by aligning traces with a voltage upcrossing at $V=-20$ mV (blue star) for 651 cycles. }}\label{fig:V_trace}
\end{figure}

Figure \ref{fig:V_trace} shows a voltage trajectory generated by the 14-D stochastic HH model, under current clamp, with injected current in the  range supporting steady firing.  
The regular periodicity of the deterministic model vanishes in this case.  Nevertheless, the voltage generates spikes, which allows us to introduce a well defined  series of spike times and inter-spike intervals (ISIs). For example, we may  select a reference voltage such as $v=-20$ mV, with the property that within a neighborhood of this voltage, trajectories have strictly positive or strictly negative derivatives ($dV/dt$) with high probability. 

In \cite{Rowat2007NECO}, they suggested that  the stochastic (Langevin) 4-D HH model has a unique invariant stationary joint  density for the voltage and gating variables, as well as producing a stationary point process of spike times.  The ensemble of trajectories may be visualized by aligning the voltage spikes  (Figure \ref{fig:V_trace}b), and illustrates that each trace is either rapidly increasing or else rapidly decreasing as it passes $v=-20$ mV.

In order to give a precise definition of the interspike interval, on which we can base a first-passage time analysis, we will  consider two types of Poincar\'{e} section of the fourteen-dimensional phase space: the ``nullcline'' surface associated with the voltage variable,
\begin{equation}
\mathcal{V}^0=\{(v,\mbm,\mbn)\in\mathcal{D}\given f(v,\mbm,\mbn)=0 \},
\end{equation}
where the rate of change of voltage is instantaneously zero, and an iso-voltage
 sections of the form
\begin{equation}
\mathcal{S}^u=\{(v,\mbm,\mbn)\in\mathcal{D}\given v=u \}.    
\end{equation}
(In \S\ref{sssec:MRTphase} we will define a third type of Poincar\'{e} section, namely isochrons of the mean--return-time function $T(v,\mbn)$ \cite{CaoLindnerThomas2020SIAP}.) 
Figure \ref{fig:dvdt} illustrates the projection of $\mathcal{V}^0$ (green horizontal line) and $\mathcal{S}^{u}$ for $u\in\{-40,10\}$  (red vertical lines) onto the $(V,dV/dt)$ plane.

For any voltage $u$ we can partition the voltage-slice  $\mathcal{S}^u$ into three disjoint components $\mathcal{S}^u=\mathcal{S}^u_0\bigsqcup\mathcal{S}^u_+\bigsqcup\mathcal{S}^u_-$, defined as follows:

\begin{definition}\label{def:Spm}
Given the stochastic differential equations \eqref{eq:langevin-rescaled} defined on the HH domain $\mathcal{D}$, and for a given voltage $u$, the ``null" surface, $\mathcal{S}^{u}_0$ is defined as
$$\mathcal{S}^{u}_0\defn \mathcal{S}^u\cap\mathcal{V}^0=\left\{(v,\mbm,\mbn)\in \mathcal{D} \given v=u\  \& \   f(v,\mbm, \mbn)=0\right\},$$
the ``inward current" surface, $\mathcal{S}^{u}_+$ is defined as
$$\mathcal{S}^{u}_+\defn \left\{(v,\mbm,\mbn)\in \mathcal{D} \given v=u\  \& \    f(v,\mbm,\mbn)>0\right\},$$
and the ``outward current" surface is defined as
$$\mathcal{S}^{u}_-\defn \left\{(v,\mbm,\mbn)\in \mathcal{D} \given v=u\  \& \      f(v,\mbm,\mbn)<0\right\}.$$
\end{definition}
Figure \ref{fig:dvdt} plots $dV/dt$ versus $V$ for roughly 600 cycles, and shows that for certain values of $v$, the density of  trajectories in a neighborhood of $\mathcal{V}^0$ is very small for a finite voltage range (here shown as $-40$ to $+10$ mV).  
Indeed for any $u$, the intersection of the null set $\mathcal{S}^u_0$ has measure zero relative to the uniform measure on $\mathcal{S}^u$, and the probability of finding a trajectory at precisely $V=u$ and $dV/dt=0$ is zero.
From this observation, and because $dV/dt$ is conditionally deterministic, given $\mbn$,
it follows that a trajectory starting from $\mbx\in \mathcal{S}^{u}_+$ will necessarily cross $\mathcal{S}^{u}_-$ before crossing $\mathcal{S}^{u}_+$ again (with probability one).

\paragraph{First-Passage Times} 
Based on this observation, we can give a formal definition of the first passage time as follows.
\begin{figure}
    \centering
          \includegraphics[width=.75\textwidth]{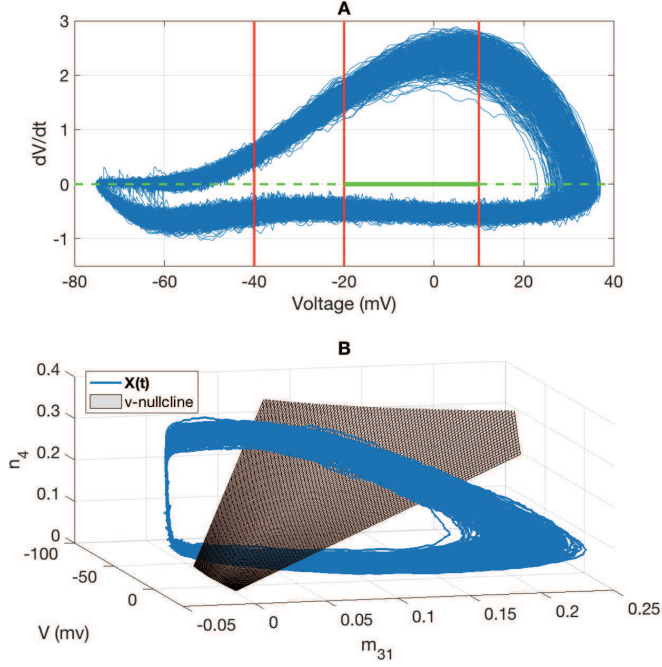}
    \vspace*{-2mm}
    \caption{{Sample trace of $\mbX(t)$ over 1000 cycles.
    \textbf{A:} View of $\mbX(t)$ (blue trace) in $(V,dV/dt)$ phase plane.  Red vertical line: projections of voltage slices $\mathcal{S}^{-40}$, $\mathcal{S}^{-20}$ and $\mathcal{S}^{+10}$.  Green horizontal line: projection of $V$-nullcline $\mathcal{V}^0$; solid portion corresponds to voltage range in second panel.
    \textbf{B:}
    Projection of $\mbX(t)$ on three coordinates $(V,m_{31},n_4)$. 
    Gray surface: subset of $v$-nullcline with $-40\text{ mV }\le v\le +10$ mV.  For this voltage range,  trajectories remain a finite distance away from $\mathcal{V}^0$ with high probability.
    }}\label{fig:dvdt}
\end{figure}

\begin{figure}
    \centering
\includegraphics[width=0.75\textwidth]{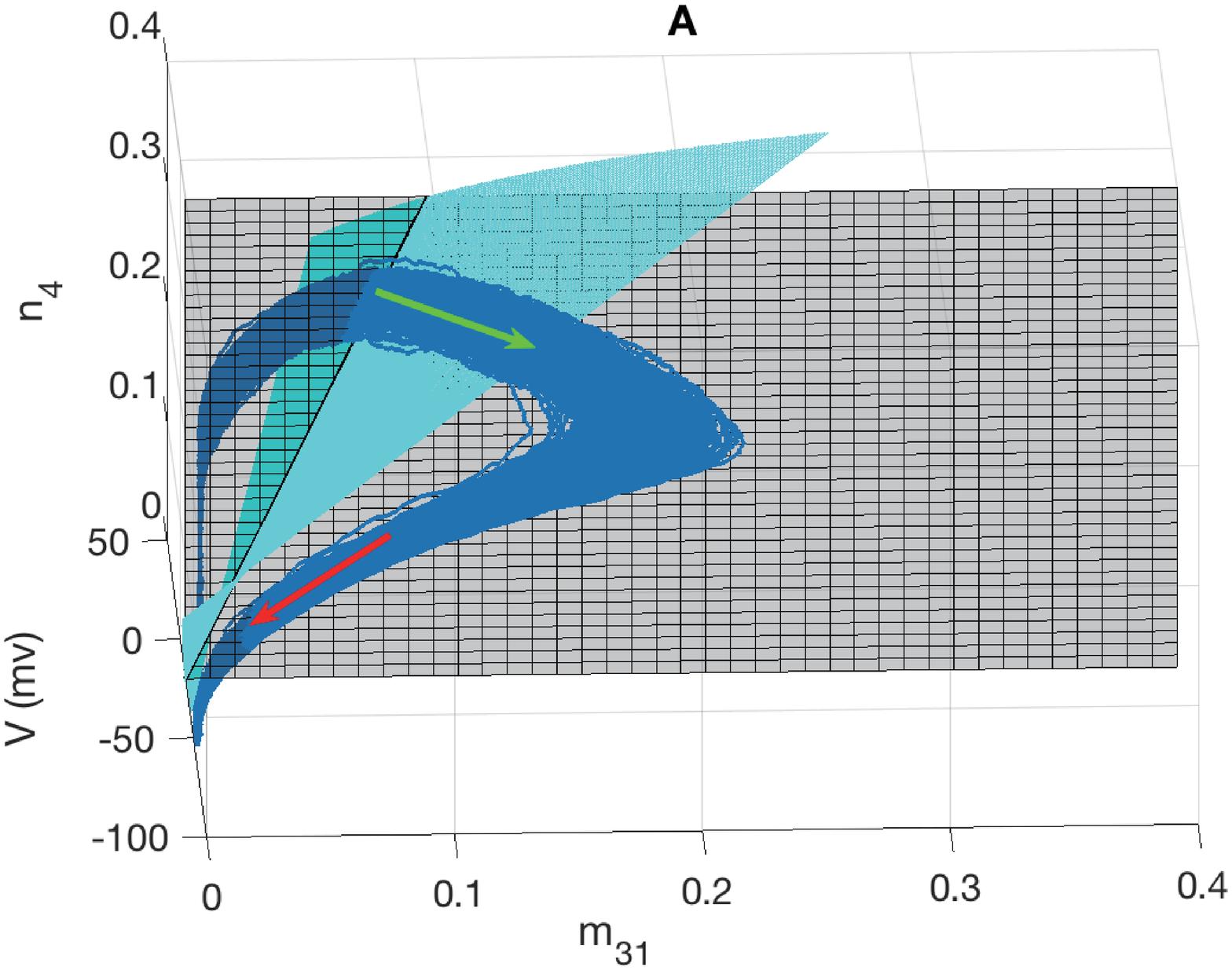}
\includegraphics[width=0.75\textwidth]{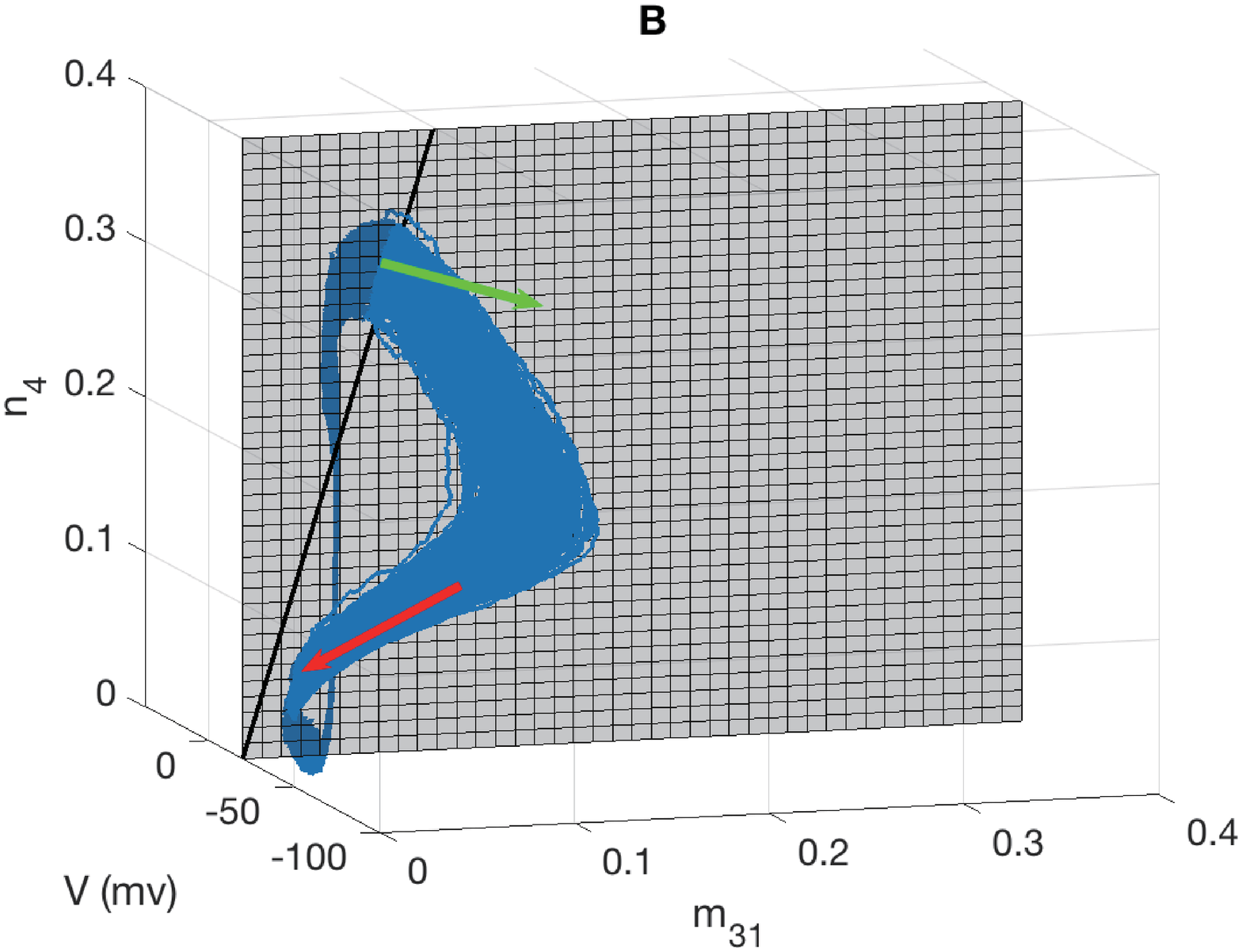}
    \caption{{Intersections of a trajectory (blue trace) with a voltage slice ($\mathcal{S}^{-20}$, grey surface) and $V$-nullcline ($\mathcal{V}^0$, cyan surface).
    \textbf{A} 
    Trajectory $\mbX(t)$ crosses $\mathcal{S}^{-20}_+$ with  increasing voltage component (green arrow). Subsequently, the trajectory crosses $\mathcal{S}^{-20}_-$ with  decreasing voltage component (red arrow).
    The trajectory $X(t)$ does not intersect with the null space for voltage in the range of $[-40,10]$mV with probability 1.
    \textbf{B}
    A special case for A with the null component $\mathcal{S}^{-20}_0$ (black diagonal line) indicated for $v=-20$mV.  The intersection of the stochastic trajectory and $v=-20$mV is partitioned into an inward component $\mathcal{S}^{-20}_+$ (green arrow shows trajectory crossing with $dV/dt>0$) and an outward component $\mathcal{S}^{-20}_-$ (red arrow shows trajectory crossing with $dV/dt<0$). Note that the null component $\mathcal{S}^{-20}_0$ does not intersect with the trajectory with probability one. }}\label{fig:phase}
\end{figure}

\begin{definition}\label{def:FPT}
 Given a section $\mathcal{S}\subset \mathcal{D}$, we define the first passage time (FPT) from a point $\mbx\in\mathcal{D}$ to $\mathcal{S}$, for a stochastic conductance-based model as
\begin{equation}
    \tau(\mbx,\mathcal{S})\defn\inf \{t>0\given \mbX(t)\in\mathcal{S}\,\&\,\mbX(0)=\mbx\}. 
\end{equation}
\end{definition}
Note that, more generally, we can think of $\tau$ as $\tau(\mbx,\mathcal{S},\omega)$, where $\omega$ is a sample from the underlying Langevin process sample space, $\omega\in\Omega$.\footnote{For the $14\times28$D Langevin Hodgkin-Huxley model, $\Omega$ may be thought of as the space of continuous vector functions on $[0,\infty)$ with 28 components -- one for each independent noise source.}
For economy of notation we usually suppress $\omega$, and may also suppress $\mathcal{S}$, or  $\mbx$ when these are clear from context.

In the theory of stochastic processes a \emph{stopping time}, $\tau$, is any random time such that the event $\{\tau\le t\}$ is part of the $\sigma$-algebra generated by the filtration $\mathcal{F}_t$ of the stochastic process from time $0$ through time $t$.  That is, one can determine whether the event defining $\tau$ has occurred or not by observing the process for times up to and including $t$ (see \cite{Oksendal2007}, \S 7.2, for further details). 
\begin{remark}
\label{rem:tau-is-a-stopping-time}
Given any section $\mathcal{S}\subset\mathcal{D}$ and any point $\mbx\in\mathcal{D}$, the first passage time $\tau(\mbx,\mathcal{S})$ is a stopping time.  
This fact will play a critical role in the proof of our main theorem.  
\end{remark}

As Figure \ref{fig:dvdt} suggests, for $-40\le v \le +10$ mV, the probability of finding trajectories in an open \emph{neighborhood} of $\mathcal{S}^v_0$ can be made arbitrarily small by making the neighborhood around $\mathcal{S}_0^v$ sufficiently small. 
This observation has two important consequences. 
First, because the probability of being near the nullcline $\mathcal{S}^v_0$ is vanishingly small, interspike intervals are well defined (cf.~Def.~\ref{def:ISI}, below), even for finite precision numerical simulation and trajectory analysis.  
In addition, this observation lets us surround the nullcline with a small cylindrical open set, through which trajectories are unlikely to pass.  This cylinder-shaped buffer will play a role in defining the mean--return-time phase in \S  \ref{sssec:MRTphase}.

Moreover, as illustrated in Figure \ref{fig:phase}, when $V=-40$, the stochastic trajectory $x$ intersects $\mathcal{S}^{-40}$ at two points within each full cycle, where one is in $\mathcal{S}^{-40}_+$ and one in $\mathcal{S}^{-40}_-$.  
In addition, the trajectory crosses $\mathcal{S}^{-40}_-$ before it crosses $\mathcal{S}^{-40}_+$ again.  
This is a particular feature for conductance-based models in which $dV/dt$ is conditionally deterministic, i.e.~the model includes no current noise.\footnote{In this paper we focus on a Langevin equation, rather than models with discrete channel noise.  
Therefore, our trajectories are diffusions, that have  continuous sample paths (with probability one).  
Therefore, the FPT $\tau(\mbx,\mathcal{S})$ is well defined. For discrete channel-noise models, a slightly modified definition would be required.}

\begin{definition}\label{def:MFPT}
Given any set $\mathcal{S}\subset\mathcal{D}$ (for instance, a voltage-section) and a point $\mbx\in\mathcal{D}$, 
the mean first passage time (MFPT) from $\mbx$ to $\mathcal{S}$, \begin{equation}
        T(\mbx,\mathcal{S})\defn \E[\tau(\mbx,\mathcal{S})],
    \end{equation} and the second moment of the first passage time is defined as
    \begin{equation}
    \label{eq:Second_moment_for_X_and_section}
        S(\mbx,\mathcal{S})\defn \E\left[\tau(\mbx,\mathcal{S})^2\right].
    \end{equation}
\end{definition}

\paragraph{Interspike Intervals} 
Starting from $\mbx\in\mathcal{S}^{v_0}_+$, at time $t=0$, we can identify the sequence of $(\tau,\mbx)$ pairs of crossing times and crossing locations as
 \begin{equation}\label{eq:define_crossing_times}
 \begin{split}
    &(\tau^u_0=0,\mbx^u_0=\mbx),(\tau^d_1,\mbx^d_1),(\tau^u_1,\mbx^u_1),\ldots,(\tau^d_k,\mbx^d_k),(\tau^u_k,\mbx^u_k),\ldots ,\\
    \text{with }& 0=\tau_0^u<\tau_1^d<\tau_1^u<\tau_2^d<\tau_2^u<\ldots<\tau_k^d<\tau_k^u,\ldots 
 \end{split}
\end{equation}
    where $\tau^d_k=\inf\{t>\tau^u_{k-1}\given \mbx\in\mathcal{S}^{v_0}_- \}$ is the $k$th down-crossing time,  $\mbx^d_k\in\mathcal{S}^{v_0}_-$ is the $k$th down-crossing location,  $\tau^u_{k}=\inf\{t>\tau^d_{k}\given \mbx\in\mathcal{S}^{v_0}_+ \}$ is the $k$th up-crossing time, and $\mbx^u_k\in \mathcal{S}^{v_0}_+$ is the $k$th up-crossing location, for all $k\in \mathbb{N}^+$.

Under constant applied current, the HH system has a unique stationary distribution with respect to which the sequence of crossing times and locations have well-defined probability distributions \cite{Rowat2007NECO}.
We define the moments of the interspike interval distribution with respect to this underlying stationary probability distribution.

\begin{definition}\label{def:ISI}
Given a sequence of up- and down-crossings, relative to a reference voltage $v_0$ as above,  
the $k$th interspike interval (ISI), $I_k$ (in milliseconds), of the stochastic conductance-based model is a random variable that is defined as
\begin{equation}
    I_k\defn\tau^u_{k+1}-\tau^u_k
\end{equation}
where $\tau^u_k$ is the $k$th up-crossing time.
The mean ISI is defined as
\begin{equation}
        I\defn \E[I_k]
    \end{equation} and the second moment of the ISI is defined as
    \begin{equation}
        H\defn \E\left[I_k^2\right]
    \end{equation}
    The variance of the ISI is defined as
     \begin{equation}
     \sigma^2_\text{ISI}\defn \E\left[(I-I_k)^2\right],
     \end{equation}
     where $k=1,2,\cdots$.
\end{definition}

It follows immediately that $\sigma_\text{ISI}^2=H-I^2.$

\subsection{Asymptotic phase and infinitesimal phase response curve}
\label{ssec:iPRC}
\begin{figure}[htbp]
\vspace*{-4mm}
    \centering
    \includegraphics[width=10cm]{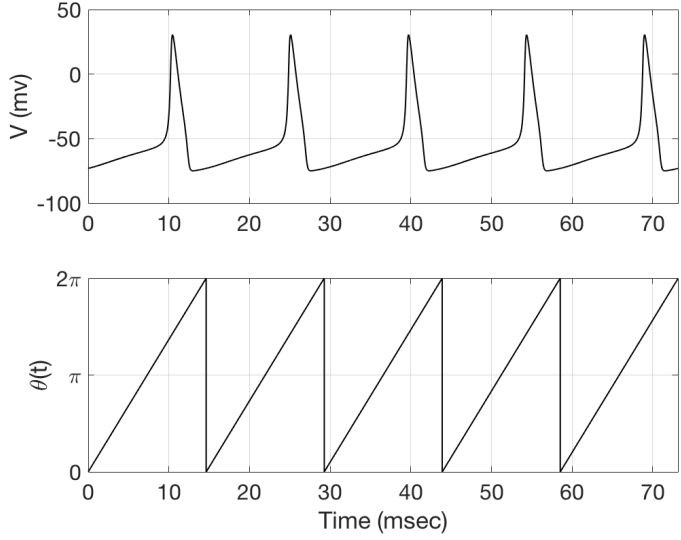}
    \caption{ Sample trace on the limit cycle and its corresponding phase function $\theta(t)$. \textbf{Top:} voltage trace for the deterministic system $d\mbx=F(\mbx)\,dt$ showing a period of $T_0\approx 14.63$ ms. 
    \textbf{Bottom:} The phase function of time saceled from $[0,2\pi)$.}\label{fig:phase_ODE}
\end{figure}
Given parameters in App.~\ref{app:notations} with an applied current $I_\text{app}=10$ nA, the deterministic HH model,
\begin{equation}\label{eq:HHODE}
    \frac{d\mbx}{dt}=F(\mbx)
\end{equation}
 fires periodically with a period $T_0\approx14.63$ msec, as shown in Fig.~\ref{fig:phase_ODE}. 
We assume that the deterministic model has an asymptotically stable limit cycle, $\gamma(t)=\gamma(t+T_0)$. 
The phase of the model at time $t$ can be defined as \cite{Schwemmer2012Springer}
\begin{equation}
    \theta(t)=\frac{(t+T_0\frac{\varphi}{2\pi})\, \text{mod}\, T_0}{T_0}\times 2\pi,
\end{equation}
where mod is the module operation, and $\theta(t)=0$ sets the spike threshold for the 
model. 
The constant $\varphi\in[0,2\pi]$ is the relative phase determined by the initial condition, and there is a one-to-one map between each point on the limit cycle and the phase.
In general, the phase  can be scaled to any constant interval;  popular choices include $[0,1)$, $[0,2\pi)$, and $[0,T)$. Here we take  $\theta\in[0,2\pi)$ (see Fig.~\ref{fig:phase_ODE}).

Winfree and Guckenheimer extended the definition of phase from the limit cycle to the basin of attraction, which laid the foundation for the asymptotic phase function $\phi(\mbx)$ \cite{Guckenheimer1975JMathBiol,Winfree2000,Winfree1974JMB}.
For the system in Eqn.~\eqref{eq:HHODE}, let $\mathbf{x}(0)$ and $\mathbf{y}(0)$ be two initial conditions, one on the limit cycle and one in the basin of attraction, respectively. 
Denote the phase associated to $\mathbf{x}(0)$ as $\theta_0(t)$.
If the solutions  $\mathbf{x}(t)$ and $\mathbf{y}(t)$ satisfy
$$\lim_{t\to \infty}|\mathbf{x}(t)-\mathbf{y}(t)|=0,\ \text{i.e.}\ \lim_{t\to \infty}|\phi(\mathbf{y}(t))-\theta_0(t)|=0,$$
then $\mathbf{y}(0)$ has  asymptotic phase $\theta_0$.
The set of all points sharing the same asymptotic phase comprises  an \emph{isochron}, a level set of $\phi(\mbx)$. 
We also refer to such a set of points as an \emph{iso-phase surface}  \cite{SchwabedalPikovsky2013PRL}. 
By construction, the asymptotic phase function $\phi(\mbx)$ coincides with the oscillator phase $\theta(t)$ on the limit cycle, i.e.~$\theta(t)=\phi(\gamma(t))$. We will assume that $\phi(\mbx)$ is twice-differentiable within the basin of attraction.

The phase response curve (PRC) is defined as the change in phase of an oscillating system in response to a given perturbations.
If the original phase is defined as $\theta_b$ and the phase after perturbation as $\theta_a$, then the PRC is the shift in phase 
$$\Delta(\theta)=\theta_a-\theta_b.$$
In the limit of small instantaneous perturbations, the PRC may be approximated by the  \emph{infinitesimal phase response curve} (iPRC) \cite{Schwemmer2012Springer,Winfree2000}. 
For a deterministic limit cycle, the iPRC $\mbZ(t)$ obeys the adjoint equation \cite{BrownMoehlisHolmes2004NeComp} 
\begin{align}
\label{eq:iPRC1}
&\frac{d\mbZ}{dt}=-\mathcal{J}(\gamma(t))^\intercal\mbZ,\\
&\mbZ(0)=\mbZ(T_0),\\
&\mbZ(0)^\intercal \mbF(\gamma(0))=1 
\label{eq:iPRC3}
\end{align}
where $T_0$ is the period of the deterministic limit cycle, $\gamma(t)$ is the periodic limit cycle trajectory (for the HH equations \eqref{eq:HHODE}, $\gamma(t) \in \R^{14}$) and $\mathcal{J}(t)=D\mbF(\gamma(t))$ is the Jacobian of $\mbF$ evaluated along the limit cycle.
The iPRC $\mbZ(t)$ is proportional to the gradient of the phase function $\phi(\mbx)$ evaluated on the limit cycle.  
For any point $\mbx$ in the limit cycle's basin of attraction, we can  define a timing sensitivity function $\tilde{\mbZ}(\mbx)\defn \frac{T_0}{2\pi} \nabla_\mbx\phi(\mbx)$. 
For the limit cycle trajectory $\gamma(t)$, we have  $\mbZ(t)=\tilde{\mbZ}(\gamma(t))$.  The first component of $\mbZ$, for example, has units of msec/mv, or change in time per change in voltage.  

\subsection{Small-noise expansions}\label{subsec:small_noise_expan}

Given the scaling of the noise coefficients with system size, $\epsilon\propto 1/\Omega$ (cf.~p.~\pageref{page:epsilon_propto_N}), the larger the system size, the smaller the effective noise level.  
For sufficiently small values of $\epsilon$, the solutions to eq.~\eqref{eq:langevin-rescaled} remain close to the determinstic limit cycle; the (stochastic) interspike intervals will remain close to the deterministic limit cycle period $\tbar_0$.
If $\mbX(t)$ is a trajectory of \eqref{eq:langevin-rescaled}, and $\phi(\mbx)$ is any twice-differentiable function, then Ito's lemma  \cite{Oksendal2007} gives an expression for the increment of $\phi$ during a time increment $dt$, beginning from state $\mbX$:
\begin{align}
\label{eq:ItoLemma}
d\phi(\mbX(t))&=(\nabla\phi(\mbX))\cdot d\mbX+\frac\epsilon 2\sum_{ij}\frac{\partial^2\phi(\mbX)}{\partial x_i\partial x_j}\,dt\\
&=\left( \mbF(\mbX)\cdot\nabla\phi(\mbX)+\frac\epsilon 2 \sum_{ij}\frac{\partial^2\phi(\mbX)}{\partial x_i\partial x_j} \right)\,dt+\sqrt{\epsilon}\left(\nabla\phi(\mbX)\right)^\intercal\mathcal{G}(\mbX)\,d\mbW
\nonumber\\
\label{eq:Ldagger}
&=\LL^\dagger[\phi(\mbX)]\,dt+\sqrt{\epsilon}\left(\nabla\phi(\mbX)\right)^\intercal\mathcal{G}(\mbX)\,d\mbW
\end{align}
up to terms of order $dt$.
The operator $\LL^\dagger$ defined by \eqref{eq:ItoLemma}-\eqref{eq:Ldagger}
is the formal adjoint of the Fokker-Planck or Kolmogorov operator \cite{Risken1996}, also known as the generator of the Markov process \cite{Oksendal2007}, or the Koopman operator \cite{LasotaMackey94}.  

Dynkin's formula, which we will use to prove our main result, is closely related to equation \eqref{eq:Ldagger}.  
Let $\mbx\in\mathcal{D}$ and let $\E^\mbx$ denote the probability law for the ensemble of stochastic trajectories beginning at $\mbx$.  
Dynkin's theorem (\cite{Oksendal2007}, \S 7.4)  states that if $\phi$ is a twice-differentiable function on $\mathcal{D}$, and if $\tau$ is any stopping time (cf.~Remark \ref{rem:tau-is-a-stopping-time}) such that $\E^\mbx[\tau]<\infty$, then
\begin{equation}
\label{eq:Dynkin-formula}
    \E^\mbx[\phi(\mbX(\tau)]=\phi(\mbx)+\E^\mbx\left[\int_{0}^\tau \LL^\dagger[\phi(\mbX(s)]\,ds\right].
\end{equation}

\subsection{Iso-phase Sections}\label{sssec:MRTphase} 
For the deterministic model, the isochrons form a system of Poincar\'e sections $\mathcal{S}_\varphi,\,\varphi\in[0,2\pi],$ each with a constant return time equal to the oscillator period $T_0$.
When the system is perturbed by noise, $\epsilon>0$ in \eqref{eq:langevin-rescaled},  we consider a set of ``iso-phase sections"  based on a mean--return-time (MRT)  construction, first proposed by
\cite{SchwabedalPikovsky2013PRL} and rigorously analyzed by \cite{CaoLindnerThomas2020SIAP}.  
As shown in  \cite{CaoLindnerThomas2020SIAP}, the MRT iso-phase surfaces $\mathcal{S}$ are the level sets of a function $T_\epsilon(\mbx)$ satisfying the MRT property.  Namely, if $\mathcal{S}$ is an iso-phase section, then the mean time taken to return to $\mathcal{S}$, starting from any $\mbx\in\mathcal{S}$, after one full rotation, is equal to the mean period, $\tbar_\epsilon$. 

The construction in \cite{CaoLindnerThomas2020SIAP} requires that the Langevin equation \eqref{eq:langevin-rescaled} be defined on a domain with the topology of an $n$-dimensional cylinder, because finding the MRT function $T_\epsilon(\mbx)$ involves specifying an arbitrary ``cut" from the outer to the inner boundary of the cylinder.  
Conductance-based models in the steady-firing regime, where the mean-field equations support a stable limit cycle, can be well approximated by cylindrical domains. In particular, their variables are restricted to a compact range, and there is typically a ``hole" through the domain in which trajectories are exceedlingly unlikely to pass, at least for small noise.

As an example, consider the domain  for the 14D HH equations (recall Defs.~\ref{dfn:HHdomain}), namely  $\mathcal{D}\defn [v_\text{min},v_\text{max}]\times\Delta^7\times \Delta^4$.  
The $p$-dimensional simplex $\Delta^p$ is a bounded set, and, as established by Lemma \ref{Lemma:vmin_vmax}, the trajectories of \eqref{eq:langevin-rescaled} remain within fixed voltage bounds with probability 1, so our HH system operates within a bounded subset of $\R^{14}$.  
To identify a ``hole'' through this domain, note that the set
$$\mathcal{S}^{u}_0\defn \mathcal{S}^u\cap\mathcal{V}^0=\left\{(v,\mbm,\mbn)\in \mathcal{D} \given v=u\  \& \   f(v,\mbm, \mbn)=0\right\},$$ 
which is the intersection of the voltage nullcline $\mathcal{V}^0$ with the constant-voltage section $\mathcal{S}^u$, is rarely visited by trajectories under small noise conditions (Fig.~\ref{fig:dvdt}B).

For $r>0$, we define the open ball of radius $r$ around $\mathcal{S}^u_0$ as
\begin{equation}
    \mathcal{B}_r(\mathcal{S}^u_0)\defn\left\{\mbx\in\mathcal{D}\:\Big|\: \min_{ \mby\in\mathcal{S}^u_0}\left( ||\mbx-\mby||\right)<r   \right\}.
\end{equation}
For the remainder of the paper, we  take the stochastic differential equation \eqref{eq:langevin-rescaled} to be defined on 
\begin{equation}\label{eq:D0}
    \mathcal{D}_0=\mathcal{D}\backslash \mathcal{B}_r(\mathcal{S}^v_0).
\end{equation}
For sufficiently small $r>0$, $\mathcal{D}_0$
is a space homeomorphic to a cylinder in $\R^{14}$.  To see this, consider the annulus  $\mathcal{A}=I_1 \times B^{13}$, where $I_1=[0,2\pi]$, and $B^{13}$ is a simply connected subset of $\R^{13}$. 
That space is homotopy equivalent to a circle $S^1$ by contracting the closed interval parts to a point, and contracting the annulus part to its inner circle.

To complete the setup so that we can apply the theory from \cite{CaoLindnerThomas2020SIAP}, 
we set boundary conditions  $\sum_{ij}n_i(\mathcal{G}\mathcal{G}^\intercal)_{ij}\partial_j T_\epsilon=0$ at reflecting boundaries with outward normal $\mbn$ on both the innner and outer boundaries of the cylinder.  In addition, we choose an (arbitrary) section transverse to the cylinder, and impose  a jump condition $T_\epsilon\to T_\epsilon+\tbar_\epsilon$ across this section, where $\tbar_\epsilon$ is mean oscillator period under noise level $\epsilon$. 

As showed in \cite{CaoLindnerThomas2020SIAP}, this construction allows us to establish a well defined MRT function for a given noise level $\epsilon$, $\tbar_\epsilon(\mbx)$.  We then obtain the iso-phase sections as level sets of $\tbar_\epsilon(\mbx)$. We give a formal definition as follows.

\begin{definition}\label{def:IPI}
Given a fixed noise level $\epsilon\ge 0$, and an iso-phase surface $\mathcal{S}$ for eqn.~\eqref{eq:langevin-rescaled}, we define the $k$th iso-phase interval (IPI) as the random variable
\begin{equation}\label{eq:IPI}
    \Delta_k\defn\mu_{k}-\mu_{k-1},
\end{equation}
where $\{\mu_k\}_{k\in\mathrm{N}^+}$ is a sequence of times at which the trajectory crosses $\mathcal{S}$.
The mean IPI is defined as
\begin{equation}
\label{eq:Tbar_epsilon}
        \tbar_\epsilon\defn \E[\Delta_k]
    \end{equation} and the second moment of the IPI is defined as
    \begin{equation}
    \label{eq:S_epsilon}
        S_\epsilon\defn \E\left[\Delta_k^2\right].
    \end{equation}
    The variance of the IPI is defined as
     \begin{equation} 
     \label{eq:sigma_phi}
     \sigma^2_\phi\defn \E\left[(\tbar_\epsilon-\Delta_k)^2\right].
     \end{equation}
     The moments \eqref{eq:Tbar_epsilon}-\eqref{eq:sigma_phi} are
      evaluated under the stationary probability distribution.
\end{definition}
It follows immediately that for a given noise level $\epsilon,$ we have  $\sigma_\phi^2=S_\epsilon-\Delta_\epsilon^2.$

\begin{remark}
Each iso-phase crossing time, $\{\mu_k\}_{k\in\mathrm{N}^+}$, in Definition~\ref{def:IPI}, is a stopping time.
\end{remark}


\begin{remark}
Because \eqref{eq:langevin-rescaled} is a diffusion with continuous sample paths, it is possible that when $\epsilon>0$ a stochastic trajectory $\mbX(t)$  may make multiple crossings of an iso-phase section $\mathcal{S}$ in quick succession.  Should this be the case, we condition the  crossing times $\mu_k$ on completion of successive circuits around the hole in our cylindrical domain.  That is, given $\mu_k$, we take $\mu_{k+1}$ to be the \emph{first return time} to $\mathcal{S}$ after having completed at least one half a rotation around the domain.  
\end{remark}


\section{Noise Decomposition of the 14-D Stochastic HH Model} \label{sec:noise_decomp}

 Ermentrout et al.~\cite{ErmentroutBeverlinTroyerNetoff2011JCNS} studied the variance of the infinitesimal phase response curve for a neuronal oscillator driven by a white noise current, using a four-dimensional version of the Hodgkin-Huxley model as an example.   
 As a corollary result, they obtained an expression for the variance of the interspike interval, by setting the size of the perturbing voltage pulse to zero.  

 Stochastic shielding \cite{SchmandtGalan2012PRL} allows one to resolve the molecular contributions (per directed edge in the ion channel state transition graph $\mathcal{E}$, cf.~Fig.~\ref{plot:HHNaKgates}) to the variance of ion channel currents  \cite{SchmidtGalanThomas2018PLoSCB,SchmidtThomas2014JMN}, and provides a numerical method for accurate, efficient simulation of Langevin models using a small subset of the independent noise forcing (only for the ``most important edges") \cite{PuThomas2020NECO}. 

 Here we combine the stochastic shielding method with Cao et al.'s mean--return-time phase analysis \cite{CaoLindnerThomas2020SIAP}
 to obtain an analytical decomposition of the molecular sources of timing variability under current clamp.

Prior analysis of stochastic shielding ( \cite{SchmidtGalanThomas2018PLoSCB,SchmidtThomas2014JMN}) assumed voltage clamp conditions, under which the ion channel state process is a stationary Markov process.
Under current clamp, however, fluctuations of channel state determine fluctuations in current, which in turn dictate voltage changes, which then influence channel state transition probabilities, forming a closed loop of statistical interdependence. 
Therefore, the variance of ISI under current clamp becomes more difficult to analyze. 
Nevertheless, in this section, we seek a decomposition of the interspike-interval variance into a sum of contributions from each edge $k\in\mathcal{E}$, e.g.
\begin{align}
     \sigma^2_\text{ISI}(\epsilon)&=\epsilon\sum_{k\in\mathcal{E}}\sigma^2_{\text{ISI},k} + O(\epsilon^2)\\
      \sigma^2_\phi(\epsilon)&=\epsilon\sum_{k\in\mathcal{E}}\sigma^2_{\phi,k} + O(\epsilon^2)
\end{align}
to leading order as $\epsilon\to 0$. 

Theorem \ref{theorem:noise} below gives the detailed form of the decomposition.
 As preliminary evidence for its plausibility, Fig.~\ref{fig:ISI_decomp} shows the coefficient of variation (standard deviation divided by mean) of the ISI under the  stochastic shielding approximation for Langevin model in different scenarios: including noise along a single directed edge at a time (blue bars), or on edges numbered 1 to $k$ inclusive (numbering follows that in Fig.~\ref{plot:HHNaKgates}).
 For large noise (Fig.~\ref{fig:ISI_decomp}a,c), the  effects of noise from different edges combine subadditively.  
 For small noise (Fig.~\ref{fig:ISI_decomp}b,d) contributions of distinct edges to firing variability combine additively.  
 Edges with small contribution to steady-state occupancy under voltage clamp (edges 1-6 for \K, edges 1-18 for \Na, cf.~Fig.~\ref{plot:HHNaKgates}) contribute additively even in the large-noise regime. Thus even in the large-noise regime, stochastic shielding allows accurate simulation of ISI variability using significantly fewer edges for both the sodium and  potassium channels.   \newcommand{\ww}{1.5125in} 
\begin{figure}[htbp] 
   \centering
   \hspace*{-1cm}  
   \includegraphics[width=4.5in]{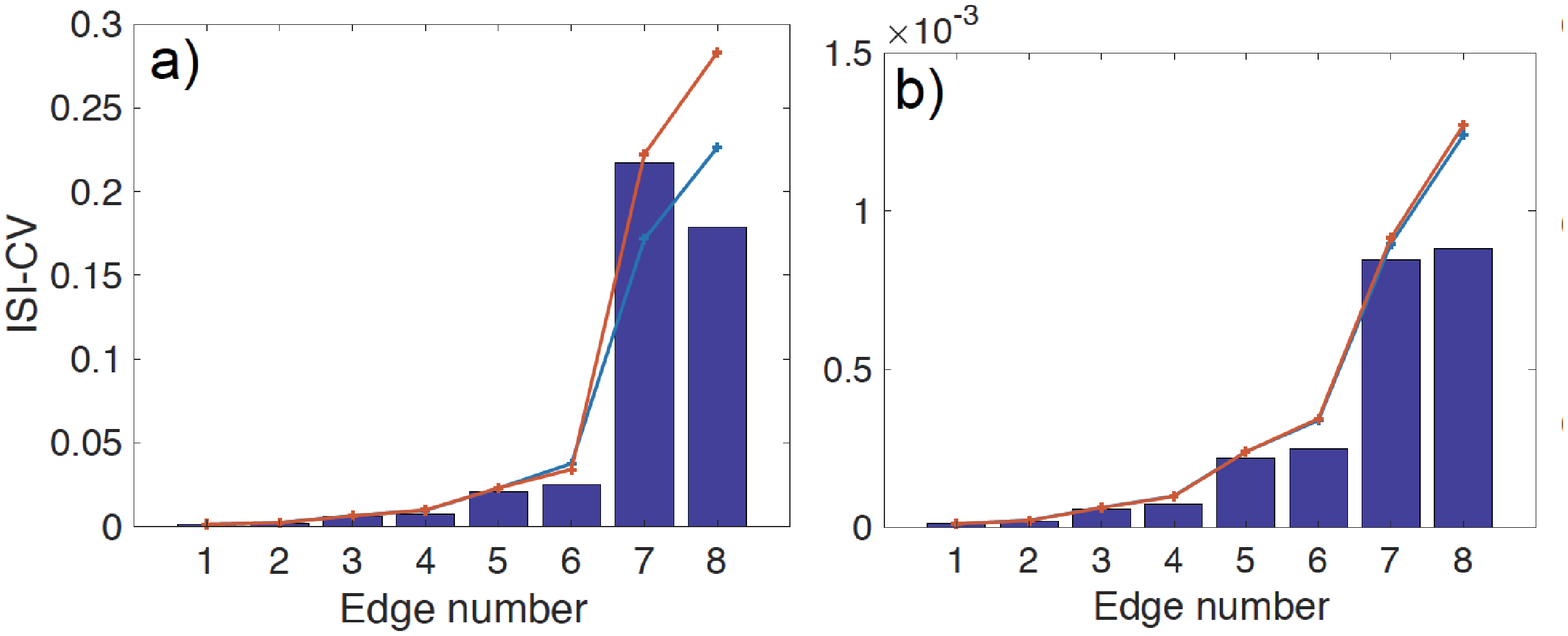}
    \includegraphics[width=4.5in]{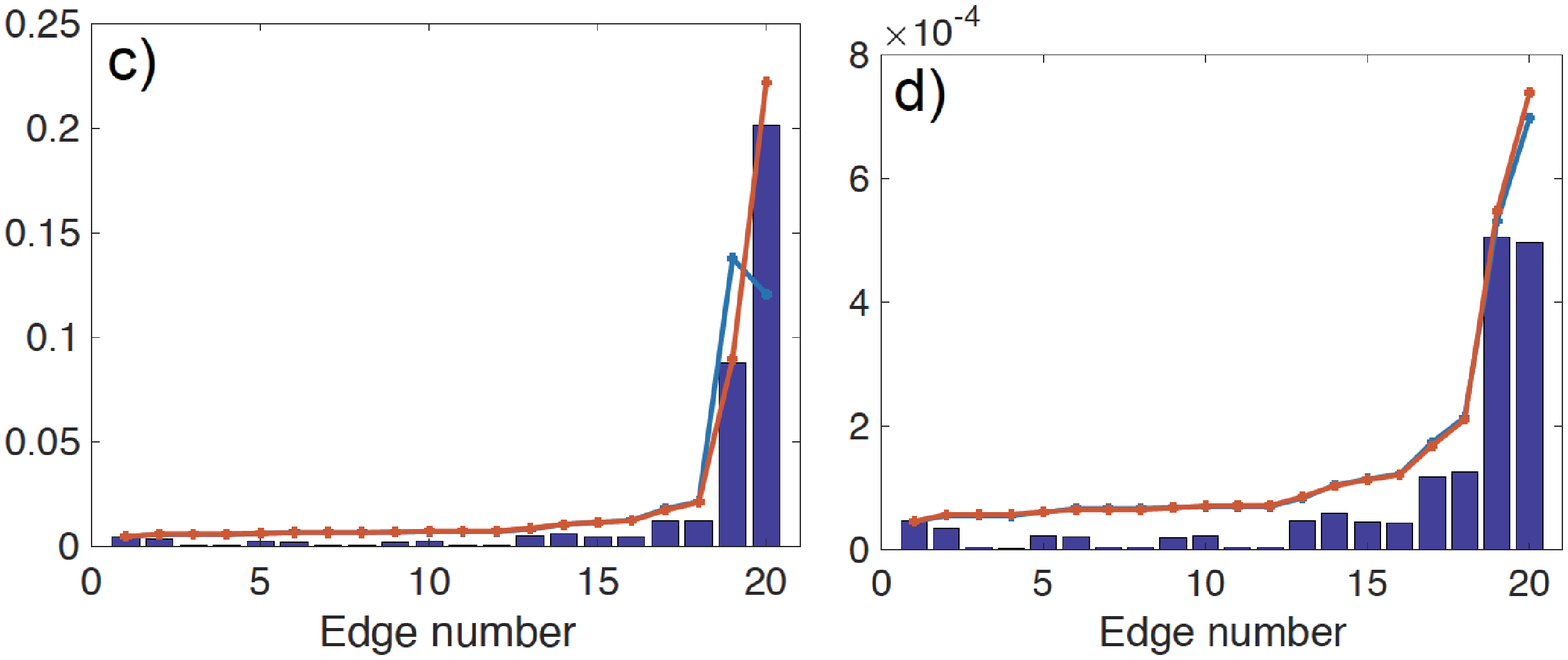}
   \caption{\small Approximate decomposition of interspike interval (ISI) variance into a sum of contributions from each edge for Hodgkin-Huxley model with stochastic K$^+$ and deterministic Na$^+$ gates (a,b) or stochastic Na$^+$ and deterministic K$^+$ gates (c,d).  Bar $n$ shows ISI coefficient of variation (CV) when noise on edge $n$ is included (a,c: $\epsilon=1$, large noise; b,d: $\epsilon=0.01$, small noise).    Blue line shows the CV of ISI when noise on all edges numbered $\le n$ are included.  Red line shows CV predicted by a linear superposition of the form \eqref{eq:ISI_channel_noise_decomp}. 
 } \label{fig:ISI_decomp}
\end{figure}

\subsection{Assumptions for Decomposition of the Full Noise Model}
Consider a Langevin model for a single-compartment conductance-based neuron \eqref{eq:langevin-rescaled}.  
We organize the state vector into the voltage component followed by fractional gating variables as follows:
\begin{equation}\label{eq:state-variables}
    \mbx=(v,q_1,q_2,\ldots,q_N)^\intercal.
\end{equation}
Here, $N$ is the number of nodes in the union of the ion channel state graphs.
For example, for the HH system, $N=13$, and we would write $q_1=m_{00},\ldots,q_8=m_{31}$ for the sodium gating variables, and $q_9=n_0,\ldots,q_{13}=n_4$ for the potassium gating variables.   
Similarly, we enumerate the $K$ edges occurring in the union of the ion channel state graphs, and write the stoichiometry vector $\zeta_k\in\R^{N+1}$ for transition $k$, taking source $i(k)$ to destination $j(k)$, in terms of  $(N+1)$-dimensional unit vectors $\mbe^{N+1}_i\in\R^{N+1}$ as $\zeta_k=-\mbe^{N+1}_{i(k)}+\mbe^{N+1}_{j(k)}.$
In order to study the contributions of individual molecular transitions to spike-time variability, we develop asymptotic expansions of the first and second moments of the distribution of iso-phase surface crossing times (iso-phase interval distribution see Def.~\ref{def:IPI} above) in the small $\epsilon$ limit.

Before formally stating the theorem, we make the following assumptions concerning the system \eqref{eq:langevin-rescaled}:
\begin{itemize}
    \item[A1] We assume the deterministic dynamical system $\frac{d\mbX}{dt}=\mbF(\mbx)$ 
    has an asymptotically, linearly stable limit cycle $\mbx=\gamma(t)$ with finite period $T_0$, and asymptotic phase function $\phi(\mbx)$ defined throughout the limit cycle's basin of attraction such that $\frac{d\phi(\mbx(t))}{dt}=\frac{2\pi}{T_0}$ along deterministic trajectories, and a well defined infinitesimal phase response curve (iPRC), $\mbZ(t)=\nabla\phi(\gamma(t))$.  
    \item[A2] We assume that the $(N+1)\times K$ matrix $\mathcal{G}$ has the form
    \begin{equation}
        \mathcal{G}(\mbx)=\sum_{k=1}^K\left( \zeta_k\mbr_k\right)\sqrt{\alpha_k(v)q_{i(k)}}
    \end{equation}
    where $\mbr_k=\left(\mbe_k^K\right)^\intercal$ is an $K$-dimensional unit row vector with all zero components except in the $k$th entry, $\alpha_k(v)$ is the voltage-dependent per capita transition rate along the $k$th directed edge, and the $q_{i(t)}$ denote channel state occupancy probabilities as described above (cf.~\eqref{eq:state-variables}). 
    \end{itemize}
    \begin{remark}
The product $\zeta_k\mbr_k$ is a $(N+1)\times K$ sparse matrix, containing zeros everywhere except in the $k$th column. 
Each column conveys the impact of an independent noise source on the state vector \cite{PuThomas2020NECO}.
\end{remark}
        \begin{itemize}
    \item[A3] We assume that for sufficiently small noise, $0<|\epsilon|\ll 1$,  
     we have a well defined joint stationary probability distribution in the voltage $V$ and the gating variables  $Q_1,\ldots,Q_N$ with a well defined mean period $\tbar_\epsilon$ and mean--return-time phase function $T_\epsilon(\mbx)$. Moreover, we assume that
    the mean period,
    the MRT function, and the second moment  function all have well defined series expansions:
    \begin{align}
    \label{eq:Taylor-for-Te}
        \tbar_\epsilon&=\tbar_0+\epsilon\tbar_1+O(\epsilon^2) \\
        \label{eq:Taylor-for-T}
        T_\epsilon(\mbx)&=T_0(\mbx)+\epsilon T_1(\mbx)+O(\epsilon^2) \\
        \label{eq:Taylor-for-S}
        S_\epsilon(\mbx)&=S_0(\mbx)+\epsilon S_1(\mbx)+O(\epsilon^2),
    \end{align}
    as $\epsilon\to 0$. 
\end{itemize}
\begin{remark}
Note that the expansions \eqref{eq:Taylor-for-Te}-\eqref{eq:Taylor-for-S} may break down in the small-$\epsilon$ limit for noise-dependent oscillators, such as the heteroclinic oscillator \cite{ThomasLindner2014PRL} or ecological quasi-cycles \cite{MckaneNewman2005PRL}, but should remain valid for finite-period limit cycles such as the Hodgkin-Huxley system in the periodic spiking regime. 
\end{remark}

\subsection{Noise Decomposition Theorem}
\begin{theorem}[Noise Decomposition Theorem]
\label{theorem:noise}
Let $\mbx=(V_0,\mbQ_0)\in\mathcal{S}_0$ be the point on the deterministic limit cycle such that $\phi(\mbx)=0$ (i.e.~assigned to ``phase zero"), and let $\E^\mbx$ denote expectation with respect to the law of trajectories with initial condition $\mbx$, for fixed $\epsilon\ge 0$.
Under assumptions A1-A3, the variance $\sigma^2_\phi$ of the $\mathcal{S}_0$-isochron crossing times (iso-phase intervals, or IPI) for conductance-based Langevin models (eqn.~\eqref{eq:langevin-rescaled}) decomposes into additive contributions from each channel-state transition, in the sense that
\begin{align}\label{eq:ISI_channel_noise_decomp}
\sigma^2_\phi=&\sum_{k\in\text{all edges}}\sigma^2_{\phi,k}\\
=&\epsilon \sum_k\int_{0}^{\tbar_0}\E^\mbx
 \left(  \alpha_k(V(t))Q_{i(k)}(t)
 \left(\zeta_k^\intercal \tilde{\mbZ}(\mbX(t))  \right)^2 \right)\,dt
 + O\left(\epsilon^2\right),\label{eq:theorem1} 
\end{align}
as $\epsilon\to 0^+$.  The function $\mbX(t)=(V(t),Q_1(t),\ldots,Q_N(t))^\intercal$ denotes a stochastic trajectory of \eqref{eq:langevin-rescaled} with  initial condition $\mbx$.
\end{theorem}
\begin{remark}
The theorem holds independently of the choice of the initial point $\mbx$ on the deterministic limit cycle, in the sense that choosing a different base point would just shift the endpoints of the interval of integration; since the deterministic limit cycle is periodic with period $\overline{T}_0$, the resulting expression for $\sigma_\phi^2$ is the same.
See Corollary \ref{corollary}.
\end{remark}
\begin{remark}
The proof relies on Dynkin's formula, first--passage-time moment calculations, and a small noise expansion.  
The right hand side of \eqref{eq:ISI_channel_noise_decomp} leads to an approximation method based on sampling stochastic limit cycle trajectories, which we show below gives an accurate estimate for $\sigma^2_\phi$.
\end{remark}

\begin{remark}
Although the interspike intervals (ISI) determined by voltage crossings are not strictly identical to the iso-phase intervals (IPI) defined by level crossings of the function $T_\epsilon(\mbx)$, we nevertheless expect that the variance of the IPI, and their decomposition, will provide an accurate approximation to the variance of the ISI.  In \S\ref{sec:Numerical} we show numerically that the decomposition given by \eqref{eq:ISI_channel_noise_decomp} predicts the contribution of different directed edges to the \emph{voltage}-based ISIs with a high degree of accuracy.
\end{remark}

Before proving the theorem, we  state and prove two ancillary lemmas.

\begin{lemma}\label{Lemma1} Fix a cylindrical domain $\mathcal{D}_0$ (as in equation \eqref{eq:D0}) and an iso-phase section $\mathcal{S}_0$ transverse to the vector field $\mbF$.  
If the mean period $\tbar_\epsilon$ and MRT function $T_\epsilon(\mbx)$ have  Taylor expansions \eqref{eq:Taylor-for-Te} and \eqref{eq:Taylor-for-T}, then
the unperturbed isochron function $T_0$ and the  sensitivity of the isochron function to small noise $T_1$ satisfy
\begin{align}
 &\mbF(\mbx)\cdot\nabla T_0(\mbx)=-1,\\
&\mbF(\mbx)\cdot\nabla T_1(\mbx)=-\frac12\sum_{ij}(\mathcal{G}\mathcal{G}^\intercal)_{ij}\partial^2_{ij}T_0(\mbx),\\
&\left.\sum_{ij}n_i(\mathcal{G}\mathcal{G}^\intercal)_{ij}\partial_j T_1\right|_{\partial\mathcal{D}}=0, 
\end{align}
where $T_1\to T_1+\tbar_1\ \text{and}\ T_0\to T_0+\tbar_0 \text{ across }\mathcal{S}_0$, and $\mbn$ is the outward normal to the boundary $\partial\mathcal{D}$.
\end{lemma} 
Note that $\tbar_1$ may be determined from the stationary solution of the forward equation for $0<\epsilon$, or through Monte Carlo simulations (in some cases $\tbar_1\equiv 0$).

\begin{proof}
For all noise levels $\epsilon\ge 0$, from \cite{Gardiner2004} (Chapter 5, equation 5.5.19), the MRT function $T_\epsilon(x)$ satisfies 
 \begin{equation}
\label{eq:T1}
\Ld\left[T_\epsilon\right]=\mbF\cdot\nabla T_\epsilon + \frac\epsilon 2\sum_{ij}(\mathcal{G}\mathcal{G}^\intercal)_{ij}\partial^2_{ij}T_\epsilon=-1,
\end{equation}
together with adjoint reflecting boundary conditions at the edges of the domain $\mathcal{D}$ with outward normal vector $\mathbf{n}$
\begin{align}
\left.\sum_{ij}n_i(\mathcal{G}\mathcal{G}^\intercal)_{ij}\partial_j T_\epsilon\right|_{\partial\mathcal{D}}=0
\end{align}
and the jump condition is specified as follows. When $x$ increases across the reference section $\mathcal{S}$ in the ``forward direction", i.e., in a direction consistent with the mean flow in forwards time, the function $T_\epsilon\to T_\epsilon + \tbar_\epsilon.$ Note that since $T_0\to T_0+\tbar_0$, we also have $T_1\to T_1+\tbar_1$ across the same Poincar\'{e} section, for consistency.

Substituting the expansion \eqref{eq:Taylor-for-T} into \eqref{eq:T1} gives
\begin{align}
\label{eq:T2}
-1=&\mbF\cdot\nabla \left(T_0(\mbx)+\epsilon T_1(\mbx)+O(\epsilon^2)\right) + \frac\epsilon 2\sum_{ij}(\mathcal{G}\mathcal{G}^\intercal)_{ij}\partial^2_{ij}\left(T_0(\mbx)+\epsilon T_1(\mbx)+O(\epsilon^2)\right), \\
=&\mbF\cdot\nabla T_0(\mbx)+\mbF\cdot\nabla\left(\epsilon T_1(\mbx)+O(\epsilon^2)\right) +\frac\epsilon 2\sum_{ij}(\mathcal{G}\mathcal{G}^\intercal)_{ij}\partial^2_{ij}T_0(\mbx) \\
&+\frac\epsilon 2\sum_{ij}(\mathcal{G}\mathcal{G}^\intercal)_{ij}\partial^2_{ij}\left(\epsilon T_1(\mbx)+O(\epsilon^2)\right) \nonumber
\end{align}
Note that, when $\epsilon=0$
 \begin{equation} \label{eq:FT0}
 \mbF\cdot\nabla T_0(\mbx)=-1,
 \end{equation}
consistent with $T_0$ being equal to minus the asymptotic phase of the limit cycle (up to an additive constant). 
On the other hand, for $\epsilon\neq0$, by comparing the first order term, the sensitivity of the isochron function to small noise $T_1$ satisfies
\begin{align} \label{eq:FT1}
\mbF\cdot\nabla T_1(\mbx)=-\frac12\sum_{ij}(\mathcal{G}\mathcal{G}^\intercal)_{ij}\partial^2_{ij}T_0,\quad 
\left.\sum_{ij}n_i(\mathcal{G}\mathcal{G}^\intercal)_{ij}\partial_j T_1\right|_{\partial\mathcal{D}}=0, 
\end{align}
where $T_1\to T_1+\tbar_1\text{ across }\mathcal{S}$, and $\mbn$ is the outward normal to the boundary $\partial\mathcal{D}$,
thus we proved Lemma \ref{Lemma1}.
\end{proof} 

Our next lemma concerns the second moment of the first passage time from a point $\mbx\in\mathcal{D}$ to a given iso-phase section $\mathcal{S}_0$, that is, $S_\epsilon(\mbx)\defn E\left[\tau(\mbx,\mathcal{S}_0)^2\right]$, cf.~\eqref{eq:Second_moment_for_X_and_section}.

\begin{lemma}\label{lem:S1}
Suppose the assumptions of Lemma \ref{Lemma1} hold, and assume in addition that $S_\epsilon$ has a Taylor expansion \eqref{eq:Taylor-for-S} for  small $\epsilon$.
The second moment, $S_0(\mbx)$, of the first passage time $\tau(\mbX,\mathcal{S})$ from a point $\mbx$ to a given isochron section $\mathcal{S}_0=\{T_\epsilon(\mbx)=\text{const}\}$, and its first order perturbation, $S_1(\mbx)$, satisfy  \begin{align}
&\mbF\cdot\nabla S_0=-2T_0\\
&\mbF\cdot\nabla  S_1 + \frac{1}{2}\sum_{ij}(\mathcal{G}\mathcal{G}^\intercal)_{ij}\partial^2_{ij}S_0=-2 T_1.
\end{align}
\end{lemma} 
\begin{proof}
Following \cite{Gardiner2004} (Chapter 5, equation 5.5.19), the second moment $S_\epsilon(\mbx)$ of the first passage time from a point $\mbx$ to a given isochron $T_\epsilon(\mbx)=\text{const}$, satisfies  
$$\mathcal{L}^\dagger[S_\epsilon]:=\mbF\cdot\nabla S_\epsilon + \frac\epsilon 2\sum_{ij}(\mathcal{G}\mathcal{G}^\intercal)_{ij}\partial^2_{ij}S_\epsilon=-2 T_\epsilon.$$
Substituting in the Taylor expansions \eqref{eq:Taylor-for-Te}-\eqref{eq:Taylor-for-S}, we have to order $O(\epsilon)$
\begin{eqnarray}
\mbF\cdot\nabla (S_0+\epsilon S_1) + \frac\epsilon 2\sum_{ij}(\mathcal{G}\mathcal{G}^\intercal)_{ij}\partial^2_{ij}(S_0+\epsilon S_1)=-2 (T_0+\epsilon T_1)+O(\epsilon^2).
\end{eqnarray}
Setting $\epsilon=0$, we see that
 \begin{equation}
 \mbF\cdot\nabla S_0=-2T_0.
  \end{equation}
For $\epsilon>0$, the first order terms yield
\begin{eqnarray} \label{eq:S1}
\mbF\cdot\nabla  S_1 + \frac{1}{2}\sum_{ij}(\mathcal{G}\mathcal{G}^\intercal)_{ij}\partial^2_{ij}S_0=-2 T_1.
\end{eqnarray}
Therefore, we complete the proof of Lemma \ref{lem:S1}.
\end{proof}

\subsection{Proof of Theorem \ref{theorem:noise}}\label{sec:proof}
\begin{proof} 
We divide the proof of the Theorem into three steps.
\begin{enumerate}
\item First, we will calculate the infinitesimal generator for the variance of the iso-phase interval (IPI).

For fixed noise level $\epsilon> 0$, the variance of IPI, $\sigma^2_\phi$ is equal to the expected value of $V_\epsilon=S_\epsilon-T_\epsilon^2$, evaluated at the isochron $T=\text{const}+\tbar$.
Note that when $\epsilon=0$, the system is deterministic and the iso-phase interval has a zero variance, i.e., $V_0\equiv 0$. Expanding $S_\epsilon=S_0+\epsilon S_1+O(\epsilon^2)$ and $V_\epsilon=\epsilon V_1+O(\epsilon^2)$ to first order in $\epsilon\ll 1$, then
\begin{eqnarray}
V_\epsilon&=&V_0+\epsilon V_1+O(\epsilon^2) \\
&=&S_\epsilon-T^2_\epsilon \\
&=& S_0+\epsilon S_1+O(\epsilon^2)-\left(T_0(\mbx)+\epsilon T_1(\mbx)+O(\epsilon^2)\right)^2 \\
&=&S_0-T^2_0+\epsilon(S_1-2T_0T_1)+O(\epsilon^2),
\end{eqnarray}
thus,
\begin{eqnarray}
S_0&=&T^2_0\\
S_1&=&V_1+2T_0T_1.
\end{eqnarray}
Plug the above results into equation \eqref{eq:S1} (Lemma \ref{lem:S1}), we can obtain that
\begin{eqnarray} 
\mbF\cdot\nabla  (V_1+2T_0T_1) + \frac{1}{2}\sum_{ij}(\mathcal{G}\mathcal{G}^\intercal)_{ij}\partial^2_{ij}T^2_0=-2 T_1.
\end{eqnarray}
By the product rule and use equations \eqref{eq:FT0}, and \eqref{eq:FT1}, 
\begin{eqnarray} 
\mbF\cdot\nabla  (2T_0T_1) &=&2T_1\mbF\cdot\nabla  (T_0)+2T_0\mbF\cdot\nabla  (T_1) \\
&=&-2T_1-T_0\sum_{ij}(\mathcal{G}\mathcal{G}^\intercal)_{ij}\partial^2_{ij}T_0.
\end{eqnarray}

Therefore,
\begin{eqnarray} 
\mbF\cdot\nabla V_1-2T_1-T_0\sum_{ij}(\mathcal{G}\mathcal{G}^\intercal)_{ij}\partial^2_{ij}T_0 + \frac{1}{2}\sum_{ij}(\mathcal{G}\mathcal{G}^\intercal)_{ij}\partial^2_{ij}T^2_0=-2 T_1.
\end{eqnarray}
Since $\partial^2_{ij}T^2_0=\partial_{i}(2T_0\partial_{j}T_0)=2\partial_{i}T_0\partial_{j}T_0+2T_0\partial^2_{ij}T_0,$
it follows that
\begin{eqnarray}
&\mbF\cdot\nabla V_1 -T_0\sum_{ij}(\mathcal{G}\mathcal{G}^\intercal)_{ij}\partial^2_{ij}T_0
+\sum_{ij}(\mathcal{G}\mathcal{G}^\intercal)_{ij}\partial_{i}T_0\partial_{j}T_0
+\sum_{ij}(\mathcal{G}\mathcal{G}^\intercal)_{ij}T_0\partial^2_{ij}T_0=0 \nonumber\\
&\mbF\cdot\nabla V_1=-\sum_{ij}(\mathcal{G}\mathcal{G}^\intercal)_{ij}\partial_{i}T_0\partial_{j}T_0 \label{eq:FV1}
\end{eqnarray}
Finally, 
\begin{eqnarray} 
\mathcal{L}^\dagger[V_\epsilon]&=&\mathcal{L}^\dagger[V_0+\epsilon V_1+O(\epsilon^2)]\\
&=&\epsilon\mathcal{L}^\dagger[ V_1]+O(\epsilon^2)\\
&=&\epsilon\left(\mbF\cdot\nabla V_1 + \frac\epsilon 2\sum_{ij}(\mathcal{G}\mathcal{G}^\intercal)_{ij}\partial^2_{ij}V_1\right)+O(\epsilon^2) \\
&=&\epsilon\left(\mbF\cdot\nabla V_1\right)+O(\epsilon^2)\\
&=&-\epsilon\sum_{ij}(\mathcal{G}\mathcal{G}^\intercal)_{ij}\partial_{i}T_0\partial_{j}T_0+O(\epsilon^2), \label{LV1}
\end{eqnarray}
where we used $V_0\equiv 0$ and applied equation \eqref{eq:FV1}.


\item Secondly, we will show that for first-order transition networks underlying the molecular ion channel process, the decomposition $\mathcal{G}\mathcal{G}^\intercal=\sum_{k\in\mathcal{E}}\mathcal{G}_k\mathcal{G}_k^\intercal$ is exact.

To see this, note that $\mathcal{G}$ can be written as a sum of 29 sparse matrix with one zero matrix and 28 rank one matrix. The $k$th rank one matrix consists of the transition due to the $k$th edge and there are 28 edges in the 14-D HH model. The  $k$th column of the $k$th rank one matrix equals to a stoichiometry vector times the square root of the corresponding state occupancy and zeros otherwise. For example, the $k$th column of $\mathcal{G}$ is given by
$$\mathcal{G}_k=\zeta_k\sqrt{\alpha_k(v)\mbX_{i(k)}},$$
where $\zeta_k$ is the stoichiometry vector, $\alpha_k$ is the voltage-dependent \textit{per capita} transition rate, and  $\mbX_{i(k)}$ is the population vector component at the source node $i(k)$ for transition number $k$.
\begin{align}
\mathcal{G}\mathcal{G}^\intercal&=(\mathcal{G}_1+\mathcal{G}_2+\cdots+\mathcal{G}_{28})(\mathcal{G}_1+\mathcal{G}_2+\cdots+\mathcal{G}_{28})^\intercal\\
&=\sum_{k=1}^{28} \mathcal{G}_k\mathcal{G}_k^\intercal \label{eq:GG}\\
&=\sum_{k=1}^{28} \alpha_k(v)\mbX_{i(k)} \zeta_k \zeta_k^\intercal \label{eq:GG2}
\end{align}
where \eqref{eq:GG} holds because $\mathcal{G}_i\mathcal{G}_j^\intercal=0$ when $i\neq j$. 

 Note that $\partial_{i}T_0\partial_{j}T_0
 =\omega^{-2}\partial_{i}\phi(\mbx)\partial_{j}\phi(\mbx)
 =\tilde\mbZ_i(\mbx)\tilde\mbZ_j(\mbx)$, with $\omega\equiv 2\pi/\tbar_0$, because $\phi$ is normalized to range from 0 to $2\pi$, and $T_0$ ranges from 0 to $\tbar_0$.
\begin{align}
&\sum_{ij}(\mathcal{G}\mathcal{G}^\intercal)_{ij}\partial_{i}T_0\partial_{j}T_0=\sum_{ij}(\mathcal{G}\mathcal{G}^\intercal)_{ij}\tilde{\mbZ}_i\tilde{\mbZ}_j\\
&=\sum_{k=2}^{29}\sum_{ij} ( \mathcal{G}_k\mathcal{G}_k^\intercal)_{ij}\tilde{\mbZ}_i\tilde{\mbZ}_j  \\
&=\sum_{k=1}^{28}\Big(\alpha_k(v)\mbX_{i(k)}\sum_{ij} ( \zeta_k \zeta_k^\intercal)_{ij}\tilde{\mbZ}_i\tilde{\mbZ}_j\Big) \\
&=\sum_{k=1}^{28}\alpha_k(v)\mbX_{i(k)} \big[\tilde{\mbZ}_{i(k)}^2+\tilde{\mbZ}_{j(k)}^2-2\tilde{\mbZ}_{i(k)}\tilde{\mbZ}_{j(k)}\big] \label{eqZ}\\
&=\sum_{k=1}^{28}\alpha_k(v)\mbX_{i(k)} \big[\tilde{\mbZ}_{i(k)}-\tilde{\mbZ}_{j(k)}\big]^2\\
&=\sum_{k=1}^{28}\alpha_k(v)\mbX_{i(k)}  \left(\zeta_k^\intercal\tilde{\mbZ}\right)^2, \label{etaZ}
\end{align}
where $i(k)$ and $j(k)$ are the source and sink nodes for transition number $k$. Equation \eqref{eqZ} holds because the $k$th edge only involves two nodes.

\item Finally, we will apply Dynkin's formula to complete the rest of the proof. 

For a stopping time $\tau(\mbx)$ with $\E^\mbx\left( \tau\right)<\infty$, by Dynkin's formula \eqref{eq:Dynkin-formula}, the expected IPI variance starting from $\mbx$ is
\begin{eqnarray}
\E^\mbx\left( V_\epsilon(\mbX(\tau))\right)=V_\epsilon(\mbx)+\E^\mbx
\left(\int_0^\tau \mathcal{L}^\dagger[V_\epsilon (\mbX(s))] ds\right) \label{eq:Dynkin}
\end{eqnarray}
The first return time $\tau$ is the time at which the trajectory $\mbX(t)$ first returns to the isochron $\mathcal{S}_0$, therefore $\mbX(\tau)\in\mathcal{S}_0$ and the time left to reach $\mathcal{S}_0$ from the random location $\mbX(\tau)$ is guaranteed to be zero.  
That is, $V_\epsilon(\mbX(\tau))=0$ with probability 1.  
Hence, $\E^\mbx\left( V_\epsilon(\mbX(\tau))\right)\equiv 0$ for all $\mbx\in\mathcal{S}_0$. 

Fix a mean--return-time isochron $\mathcal{S}_0$, the mean return time from any initial location $\mbx\in\mathcal{S}_0$ back to $\mathcal{S}_0$, after completing one rotation is exactly $\overline{T}_\epsilon$, by construction.  
However, in principle, the \emph{variance} of the return time might depend on the initial location within the isochron.  
We next show that, to leading order in $\epsilon$, this is not the case, that is, the MRT isochrons have uniform first \emph{and} second moment properties.

Using equations \eqref{LV1}, \eqref{etaZ} and \eqref{eq:Dynkin}, we obtain 
\begin{align}
    V_\epsilon(\mbx)&=-\E^\mbx\left(\int_0^\tau\LL^\dagger\left[V_\epsilon(\mbX(s))\right]\,ds\right)\\
    &=\E^\mbx
\left(\int_0^{\tau} \epsilon\sum_{ij}(\mathcal{G}\mathcal{G}^\intercal)_{ij}\partial_{i}T_0\partial_{j}T_0 ds\right)+\oee \\
&=\epsilon \sum_{k=1}^{28}\E^\mbx
\left(\int_0^{\tau}\alpha_k(v)\mbX_{i(k)}  \left(\zeta_k^\intercal\tilde{\mbZ}\right)^2ds\right)+\oee, \label{eq:veps1}
\end{align}
where the integrals are evaluted along a stochastic trajectory $\mbX(t)$ with $\mbX(0)=\mbx$ and $\mbX(\tau)\in\mathcal{S}_0$, one rotation later.
%
Holding the deterministic zero-phase isochron $\mathcal{S}_0$ fixed,  and choosing an  arbitrary point $\mby\in\mathcal{D}$, we have, by definition, \begin{equation}
\E^\mby[\tau(\mby)]=T_\epsilon(\mby)=T_0(\mby)+\epsilon T_1(\mby)+\oee.
\end{equation}

Therefore, starting from  an initial condition $\mbx\in\mathcal{S}_0$ \emph{one period earlier}, we have
\begin{align}
    \E^\mbx
\left(\int_0^{\tau}\alpha_k(v)\mbX_{i(k)}  \left(\zeta_k^\intercal\tilde{\mbZ}\right)^2ds\right) = \E^\mbx
\left(\int_0^{\overline{T}_0}\alpha_k(v)\mbX_{i(k)}  \left(\zeta_k^\intercal\tilde{\mbZ}\right)^2ds\right)+\oe. \label{eq:veps2}
\end{align}
This relation follows immediately from our assumptions, because, for  $\mbx\in\mathcal{S}_0$,
\begin{align}
\nonumber
    &\left|\E^\mbx
\left(\int_0^{\tau}\alpha_k(V(s))\mbX_{i(k)}(s) \left(\zeta_k^\intercal\tilde{\mbZ}(s)\right)^2 ds\right) - \E^\mbx
\left(\int_0^{\overline{T}_0}\alpha_k(V(s))\mbX_{i(k)}(s)  \left(\zeta_k^\intercal\tilde{\mbZ}(s)\right)^2 ds\right)\right|\\
&=
\left|
\E^\mbx
\left(\int_{\overline{T}_0}^{\tau}\alpha_k(V(s))\mbX_{i(k)}(s)  \left(\zeta_k^\intercal\tilde{\mbZ}(s)\right)^2 ds\right)\right|
\\
&\le C_1\left|\E^\mbx(\tau(\mbx)-\overline{T}_0)\right|
=C_1\left|\E^\mbx(\tau(\mbx))-\E^\mbx(\overline{T}_0)\right|
=C_1\left|\overline{T}_\epsilon-\overline{T}_0\right|\\
&=\epsilon C_1 \overline{T}_1+\oee.
\label{eq:bound-on-difference-of-integrals-last-line}
\end{align}

Here $C_1$ is a positive constant bounding the integrand $\alpha_k(v(t))\mbX_{i(k)}(t)\left(\eta_k^\intercal\tilde{\mbZ}(t)\right)^2$.
From Remark \ref{rem:alpha-is-bounded}, $\alpha_k \le \alpha_\text{max}$.  
By definition, $0\le \mbX_i\le 1$ for each $i$.  
For each edge $k$, $|\zeta_k|=\sqrt{2}$.  
Since $\tilde{\mbZ}$ is continuous and periodic, $|\tilde{\mbZ}|$ is bounded by some constant $z_\text{max}$.  
Therefore setting $C_1=\sqrt{2}\alpha_\text{max}z_\text{max}$ satisfies \eqref{eq:bound-on-difference-of-integrals-last-line}.  

Because the initial point $\mbx\in\mathcal{S}_0$ was located at an arbitrary radius along the specified mean--return-time isochron, the calculation above shows that $\sigma_\phi^2=\E[V_\epsilon(\mbx)\given \mbx\in\mathcal{S}_0]$ is uniform across the isochron $\mathcal{S}_0$, to first order in $\epsilon$.  
Thus, for small noise levels, the MRT isochrons enjoy not only a uniform mean return time, but also a uniform variance in the return time, at least in the limit of small noise.  

Finally, we note that $\sigma_\phi^2$ (equivalently, and $V_\epsilon(\mbx)$) combine a sum of contributions over a finite number of edges.   From equations \eqref{eq:veps1} and \eqref{eq:veps2}, the variance of the inter-phase interval is given by 
\begin{eqnarray}
\sigma^2_\phi
&=&\epsilon \sum_{k=1}^{28}\E^\mbx
\left(\int_0^{\overline{T}_0}\alpha_k(V(s))\mbX_{i(k)} (s) \left(\zeta_k^\intercal\tilde{\mbZ}(s)\right)^2ds\right)+\oee \label{eq:IPI_final}.
\end{eqnarray}

 To complete the proof, note that \eqref{eq:theorem1} follows from \eqref{eq:theorem1a} by exchange of expectation $\E^\mbx[\cdot]$ with (deterministic) integration $\int_0^{\overline{T}_0}[\cdot]\,dt$.
This completes the proof of Theorem \eqref{theorem:noise}.

\end{enumerate}
\end{proof}

The choice of the initial reference point $\mbx$ or isochron $\mathcal{S}_0$ in \eqref{eq:theorem1a} was arbitrary and the variance of IPI is uniform to the first order.  
Therefore, the inter-phase-interval variance may be uniform (to first order) almost everywhere in $\mathcal{D}$.
We can then replace the integral around the limit cycle in \eqref{eq:theorem1a} with an integral over $\mathcal{D}$ with respect to the stationary probability distribution.
Thus we have the following
\begin{corollary}
\label{corollary}
Under the assumptions of Theorem \ref{theorem:noise}, the inter-phase-interval variance satisfies
\begin{eqnarray}
\sigma^2_\phi
&=&\epsilon \tbar_0 \sum_{k=1}^{28}\E
\left(\alpha_k(V)\mbX_{i(k)}  \left(\zeta_k^\intercal\tilde{\mbZ}(\mbX)\right)^2\right)+\oee,  \label{eq:theorem1a}
\end{eqnarray}
as $\epsilon\to 0$,
where $\E$ denotes expectation with respect to the stationary probability density for \eqref{eq:langevin-rescaled}.
\end{corollary}
\begin{remark}
Because the variance of the IPI, $\sigma_\phi^2$, is uniform regardless the choice of the reference iso-phase section, we will henceforth refer it as  $\sigma^2_\text{IPI}$ throughout the rest of this paper.
\end{remark}

 Now we have generalized the edge important measure introduced in \cite{SchmidtThomas2014JMN} for the voltage-clamp case to the current clamp case with weak noise. 
 In the next section we leverage Theorem \ref{theorem:noise} to estimate the inter-phase interval variance in two different ways: by averaging over one period of the deterministic limit cycle (compare \eqref{eq:IPI_final})
or by averaging over a long stochastic simulation (compare \eqref{eq:theorem1a}).  
As we will see below, both methods give excellent agreement with direct measurement of the inter-phase interval variance. 

\section{Numerical Results}
\label{sec:Numerical}
Theorem \ref{theorem:noise} and Corollary \ref{corollary} assert that for sufficiently weak levels of channel noise, the contributions to inter-phase interval variance made by  each individual edge in the channel state transition graph (cf.~Fig.~\ref{plot:HHNaKgates}) combine additively.
Moreover, the relative sizes of these contributions provide a basis for selecting a subset of noise terms to include for running efficient yet accurate Langevin simulations, using the stochastic shielding approximation \cite{PuThomas2020NECO}.
In this chapter, we test and illustrate several aspects of these results numerically.  

First, we confront the fact that the inter-\emph{phase}-intervals and the inter-\emph{spike}-intervals are not equivalent, since iso-voltage surfaces do not generally coincide with isochronal surfaces \cite{WilsonErmentrout2018SIADS_operational}.
Indeed, upon close examination of the ISI variance in both real and simulated nerve cells, we find that the voltage-based $\sigma^2_\text{ISI}$ is not constant, as a function of voltage, while the phase-based $\sigma^2_\text{IPI}$ remains the same regardless of the choice of reference isochron.
Nevertheless, we show that the voltage-based ISI variance is well approximated -- to within a few percent -- by the phase-based IPI variance, and therefore, the linear decomposition of Theorem \ref{theorem:noise} approximately extends to the ISI variance as well.

Second, after showing that the linear decomposition of the ISI variance holds at sufficiently small noise levels, we explore the range of noise levels over which the linear superposition of edge-specific contributions to ISI variance holds.  
Consistent with the basic stochastic shielding phenomenon, we find that the variability resulting from noise along edges located further from the observable transitions scales linearly with noise intensity, $\epsilon$ even for moderate noise levels, while the linear scaling of eqn.~\eqref{eq:IPI_final} breaks down sooner with increasing noise for edges closer to observable transitions.  

Finally, we explore the accuracy of a reduced representation using only the two most important  edges from the \K channel and the four most important edges from the \Na channel, over a wide range of noise intensities.  
Here, we find that removing the noise from all but these six edges still gives an accurate representation of the ISI variance far beyond the weak noise regime, despite the apparent breakdown of linearity.

In this section, the variance of ISIs and IPIs are calculated to compare with the predictions using Theorem~\ref{theorem:noise}. 
First, we numerically show that there is a small-noise region within which Theorem~\ref{theorem:noise} holds, for each individual edge, as well as for the whole Langevin model (cf.~\eqref{eq:langevin-rescaled}).
We have two numerical approaches to evaluating the theoretical contributions.  The first method involves integrating once around the deterministic limit cycle while evaluating the local contribution to timing variance at each point along the orbit. 
This approach derives from the theorem, cf.~\eqref{eq:theorem1} or \eqref{eq:IPI_final}, which we refer as the ``limit cycle prediction". 
The second approach derives from the corollary, \eqref{eq:theorem1a}: we average the expected local contribution to timing variation over a long stochastic trajectory.  
More specifically, equation \eqref{eq:theorem1a} gives a theoretical value of the average leading-order contribution mass function, $\mathcal{P}_k$, for the $k^{th}$ edge, as
\begin{equation}\label{eq:Pk}
    \mathcal{P}_k:=\mathbb{E}\left[\alpha_k(V)\mbX_{i(k)}\left(\zeta_k^\intercal \tilde{\mbZ}(\mbX)  \right)^2\right],
\end{equation}
where $\mathbb{E}(\cdot)$ is the mean with respect to the stationary probability distribution of the stochastic limit cycle. 
Given a sample trajectory $\mbX(t)$, we approximate the iPRC near the limit cycle, $\tilde{\mbZ}(\mbX(t))$, by  using the phase response curve of the deterministic limit cycle  
\begin{equation} \label{eq:pmp}
\tilde{\mbZ}(\mbX(t))\approx    \hat{\mbZ}(\mbX(t))\defn     \mbZ\left(
    \underset{s}{\operatorname{argmin}} \left|\bigg(\gamma(s)-\mbX(t)\bigg)^\intercal \mbZ(s)\right|\right),
\end{equation}
where $\gamma$ is a point on the deterministic limit cycle and $\mbZ$ is the infinitesimal phase response curve on the limit cycle (cf.~\S\ref{ssec:iPRC}). 
The predicted contribution of the $k^{th}$ edge to the IPI variance with average period $T_0$, is therefore 

From Corollary \ref{corollary} we have
\begin{equation}
\label{eq:point-mass-prediction}
    \sigma^2_\text{IPI}\approx \epsilon \tbar_0 \sum_k \mathcal{P}_k.
\end{equation}
We call $\mathcal{P}_k$ the \emph{point mass prediction} for the contribution of the $k$th edge to the inter-phase interval variance.

For small noise, both approaches give good agreement with the directly measured IPI variance, as we will see in Fig.~\ref{fig:ISI_K_vth}.

To numerically calculate the contribution for each directed transition in Fig.~\ref{plot:HHNaKgates}, we apply the stochastic shielding (SS) technique proposed by \cite{SchmandtGalan2012PRL}, simulating the Langevin process with noise from all but a single edge suppressed.
Generally speaking, the SS method approximates the Markov process using fluctuations from only a subset of the transitions, often the observable transitions associated to the opening states.
Details about how stochastic shielding can be applied to the $14\times 28$D Langevin model is discussed in our previous paper \cite{PuThomas2020NECO}. 

All numerical simulations for the Langevin models  use the same set of parameters, which are specified in Tab.~\ref{tab:parameters} with given noise level $\epsilon$ in eqn.~\eqref{eq:langevin-rescaled}.
We calculate the following quantities: the point mass prediction $\mathcal{P}_k$, using exact stochastic trajectories \eqref{eq:Pk}; the predicted contributions by substituting the stochastic terms in \eqref{eq:IPI_final} with the deterministic limit cycle; the variance and standard deviation of the interspike intervals ($\sigma^2_\text{ISI}$); and the variance and standard deviation of the isophase intervals ($\sigma^2_\text{IPI}$). 

In addition to numerical simulations, we will also present several observations of experimental recordings.
Data in Fig.~\ref{fig:ISI_num_change} and Fig.~\ref{fig:ISI_change} were recorded {\it{in vitro}} in Dr.~Friel's laboratory from intact wild type Purkinje cells with synaptic input blocked (see \S\ref{sec:experimental-methods} for details). 
We analyzed fourteen different voltage traces from cerebellar Purkinje cells from wild type mice, and seventeen from mice with the \emph{leaner} mutation. 
The average number of full spike oscillations is roughly 1200 for wild type PCs (fourteen cells) and 900 for leaner mutation (seventeen cells).

\subsection{Observations on $\sigma^2_\text{ISI}$ and $\sigma^2_\text{IPI}$} \label{sec:obers_ISI}

When analyzing voltage recordings from {\it{in vitro}} Purkinje cells (PCs) and from simulation of the stochastic HH model, we have the following observations.
First, given a particular (stochastic) voltage trace, the number of interspike intervals (cf.~Def.~\ref{def:ISI}) varies along with the change in voltage threshold used for identifying spikes.  Second, within a range of voltage thresholds for which the number of ISIs is constant, 
the variance of the interspike interval distribution, $\sigma^2_\text{ISI}$ (cf.~Def.~\ref{def:ISI}), which is obtained directly from the voltage recordings, nevertheless varies as a function of the threshold used to define the spike times. 
Thus the ISI variance, a widely studied quantity in the field of computational neuroscience \cite{GutkinErmentrout1998NC,Lindner2004PRE_interspike,Netoff2012PRCN,Shadlen1998JNe,Stein1965BioJ}, is not invariant with respect to the choice of voltage threshold.
To our knowledge this observation has not been previously reported in the neuroscience literature.\footnote{{Throughout this section, we use the term ``threshold" in the data analysis sense of a Schmitt trigger \cite{Schmitt1938IOP}, rather than the physiological sense of a spike generation mechanism.}}

Fig.~\ref{fig:ISI_num_change} plots the histogram of voltage from a wild type PC and number of spikes corresponding to voltage threshold ($V_\text{th}$) in the range of $[-60,-10]$ mV. 
Setting the threshold excessively low or high obviously will lead to too few (or no) spikes.  
As the threshold increases from excessively low values, the counts of threshold-crossing increases.
For example, when $V_\text{th}$ is in the after hyper-polarization (AHP) range (roughly $-58\lesssim V_\text{th}\lesssim-48$ mV in Fig.~\ref{fig:ISI_num_change}) the voltage trajectory may cross the threshold multiple times before it finally spikes.
As illustrated in Fig.~\ref{fig:ISI_num_change}, the number of spikes is not a constant as the threshold varies, therefore, the mean and variance of ISI are not well-defined in the regions where extra spikes are counted.
To make the number of spikes accurately reflect the number of full oscillation cycles, in what follows we will only use   thresholds in a voltage interval that induces the correct number of spikes.
Note that, for a given voltage trace and duration ($T_\text{tot}$), if two voltage threshold generate the same number of spikes ($N_\text{spike}$), the mean ISI would be almost identical, approximately $T_\text{tot}/N_\text{spike}$. This observation holds for both experimental recordings and numerical simulations.
\begin{figure}[htbp]
   \centering
\includegraphics[width=1\textwidth]{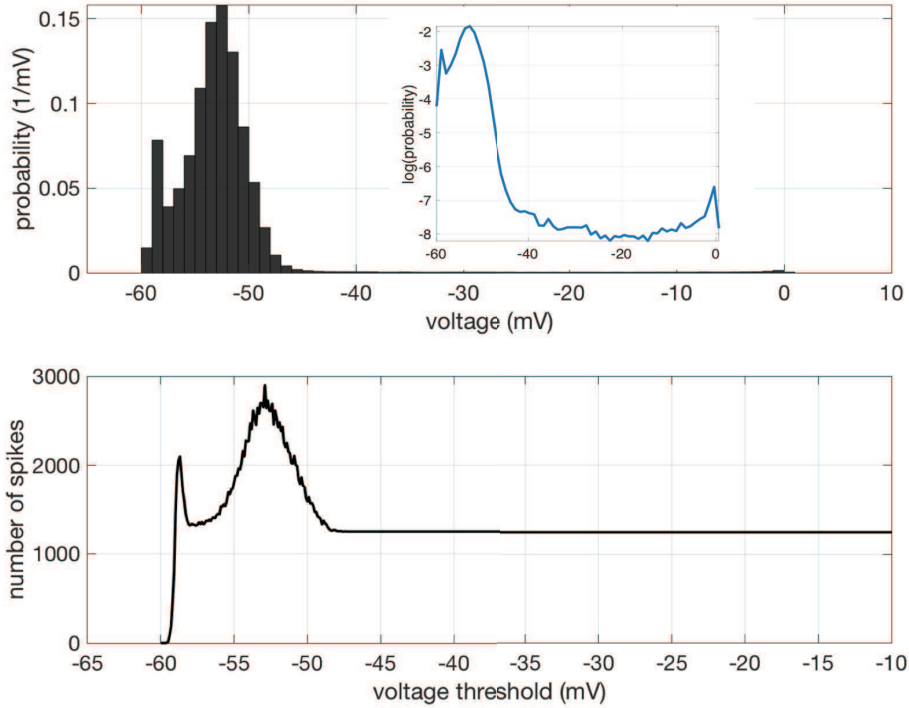}
\caption[Histogram of voltage recordings and number of spikes as a function of voltage threshold]{Histogram of voltage and number of spikes as a function of voltage threshold ($V_\text{th}$) for one wild type Purkinje cell (same data as in Fig.~\ref{fig:WTLE_stats}A \& C). 
The number of ISIs is found by counting up-crossing times as defined in Def.~\ref{def:ISI}. 
For this particular trajectory, the AHP phase locates roughly in the interval $[-58,-48]$ mV.  
The trajectory has 1248 full oscillation cycles. 
When $V_\text{th}$ is near $-60$ mV, it captures fewer spikes than the true value, and when $V_\text{th}\in[-57,-48]$, it tends to overestimate the number of spikes. For $V_\text{th}\in[-48,-10]$, the number of spikes is a constant (1248) that matches the number of full oscillations. 
}\label{fig:ISI_num_change}
\end{figure}

Next we address the sensitivity of the interspike interval to the voltage threshold, within the range over which the number of ISIs is invariant.  
(By ``threshold" we refer throughout to the voltage level used to detect the occurrence and measure the timing of an action potential, rather than a physiological threshold associated with a spike-generation mechanism.)  

From the earliest days of quantitative neurophysiology, the extraction of spike
timing information from voltage traces recorded via microelectrode has relied on
setting a fixed voltage threshold (originally called a Schmitt trigger, after the
circuit designed by O.H. Schmitt \cite{Schmitt1938IOP}).  To our knowledge, it has invariably been
assumed that the choice of the threshold or trigger level was immaterial,
provided it was high enough to avoid background noise and low enough to capture
every action potential \cite{Gerstein1960Science,Mukhametov1970WOL}.  This assumption
is generally left implicit.  
Here, we show that, in fact, the choice of the trigger level (the voltage threshold used for identifying spike timing) can cause a change in the variance of the interspike interval for a given spike train by as much as 5\%.  

\begin{figure}[htbp]
   \centering
\includegraphics[width=1\textwidth]{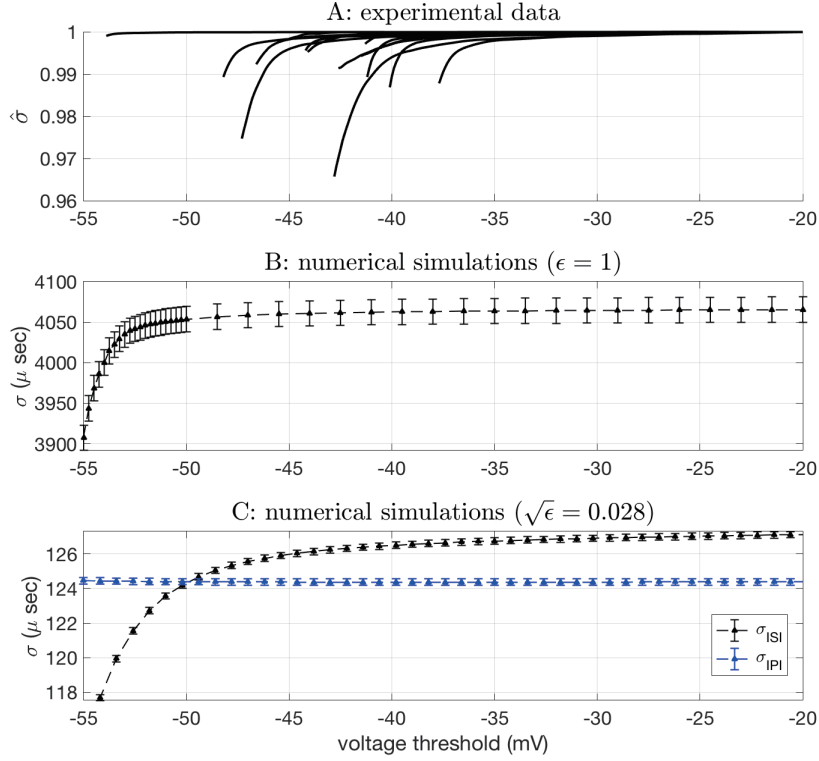}
\caption[Standard deviation of the interspike intervals and iso-phase intervals]{Standard deviation of the interspike intervals ($\sqrt{\sigma^2_\text{ISI}}$, cf.~Def.~\ref{def:ISI}) and standard deviation of the iso-phase intervals ($\sqrt{\sigma^2_\text{IPI}}$, cf.~Def.~\ref{def:IPI}) as a function of voltage threshold.  
\textbf{A:} Rescaled ISI standard deviation ($\hat{\sigma}$) obtained from experimental data recordings from 14 wild type Purkinje cells.
For experimental methods see \S\ref{sec:experimental-methods}.
For each cell, $\sqrt{\sigma^2_\text{ISI}}$ is calculated using voltage threshold ranging from -55 mV to  -20 mV, and scaled by dividing the stand deviation at voltage=-20 mV. 
\textbf{B, C:} standard deviation of ISI when $\epsilon=1$ and $\sqrt{\epsilon}=0.028$ in equation \eqref{eq:langevin-rescaled}, respectively. 
For each voltage threshold, 500 different traces are generated with each trace containing roughly 1000 interspike intervals.
Error bars indicate the 95\% confidence interval of $\sqrt{\sigma^2_\text{ISI}}$ at each threshold. Note the vertical axis is in $\mu$sec.
In \textbf{C}, each value of $\sigma^2_\text{IPI}$ is calculated for the mean--return-time isochron intersecting the deterministic limit cycle at the voltage specified.
 }\label{fig:ISI_change}
\end{figure}

Fig.~\ref{fig:ISI_change} provides evidence both from experimental traces recorded \textit{in vitro}, and from numerical simulations, that $\sigma^2_\text{ISI}$ is sensitive to the voltage threshold defining spike times.
In Fig.~\ref{fig:ISI_change}~A, we superimpose ISI standard deviations from fourteen wild type Purkinje cells, plotted as functions of the the trigger voltage $V_\text{th}$.  
We rescale each plot by the  standard deviation of the ISI for each cell at $V_\text{th}=-20$ mV, which we define as $\bar{\sigma}$. 
As shown in Fig.~\ref{fig:ISI_change}~A, the cells recorded in vitro have a clear variability in the standard deviation as the voltage threshold changes.
Specifically, the standard deviation of ISI gradually increases as voltage threshold increases and then remains constant as the threshold approaches the peak of the spikes.
Two of the cells have larger variations in the standard deviation, with roughly a $3\%-4\%$ change; nine of them have a $1\%-3\%$ change; and three of them show $0.1\%-1\%$ change.

We applied a similar analysis  to seventeen PCs with the leaner mutation \cite{Walter2006NatNeuro}.  
In this case,  one cell had a variation of roughly $1\%$ in the standard deviation, five cells with variations around $0.2\%$, and the remaining without an obvious change (data not shown). 
This difference between cells derived from wild type and leaner mutant mice is an interesting topic for future study.

We observe a similar variability of $\sigma^2_\text{ISI}$ in numerical simulations using our stochastic Langevin HH model (cf.~eqn.~\eqref{eq:langevin-rescaled}). 
Fig.~\ref{fig:ISI_change}~B and C plots two examples showing the change in $\sigma^2_\text{ISI}$ as voltage threshold varies.
For a given noise level ($\epsilon$) and a voltage threshold ($V_\text{th}$), a single run simulates a total time of 9000 milliseconds (ms), with a time step of 0.008 ms, consisting of at least 600 ISIs, which was collected as one realization for the corresponding $\sigma_\text{ISI}(\epsilon,V_\text{th})$.
The mean and standard deviation of the $\sigma_\text{ISI}(\epsilon,V_\text{th})$ is calculated for 1,000 realizations of the aforementioned step for each pair of $\epsilon$ and $V_\text{th}$.
The error bars in Fig.~\ref{fig:ISI_change}~B and C indicate 95$\%$ confidence intervals of the standard deviation.
As illustrated in Fig.~\ref{fig:ISI_change}~B and C, the standard deviation gradually increases as the trigger threshold increases during the AHP, and this trend is observed for both small and large noises.
When $\epsilon=1$, the noisy system in eqn.~\eqref{eq:langevin-rescaled} is not close to the deterministic limit cycle, and there is not a good approximation for the phase response curve. 
When $\sqrt{\epsilon}=0.028$, the system eqn.~\eqref{eq:langevin-rescaled} can be considered in the small-noise region and thus we can find a corresponding phase on the limit cycle as the asymptotic phase.
As shown in Fig.~\ref{fig:ISI_change}~C, unlike the variance of ISI, the variance of IPI is invariant with the choice of the phase threshold ($\phi$). 

\subsection{Numerical Performance of the Decomposition Theorem}\label{sec:num_perf_theorem}

\begin{figure}[htbp] 
   \centering
\includegraphics[width=1\textwidth]{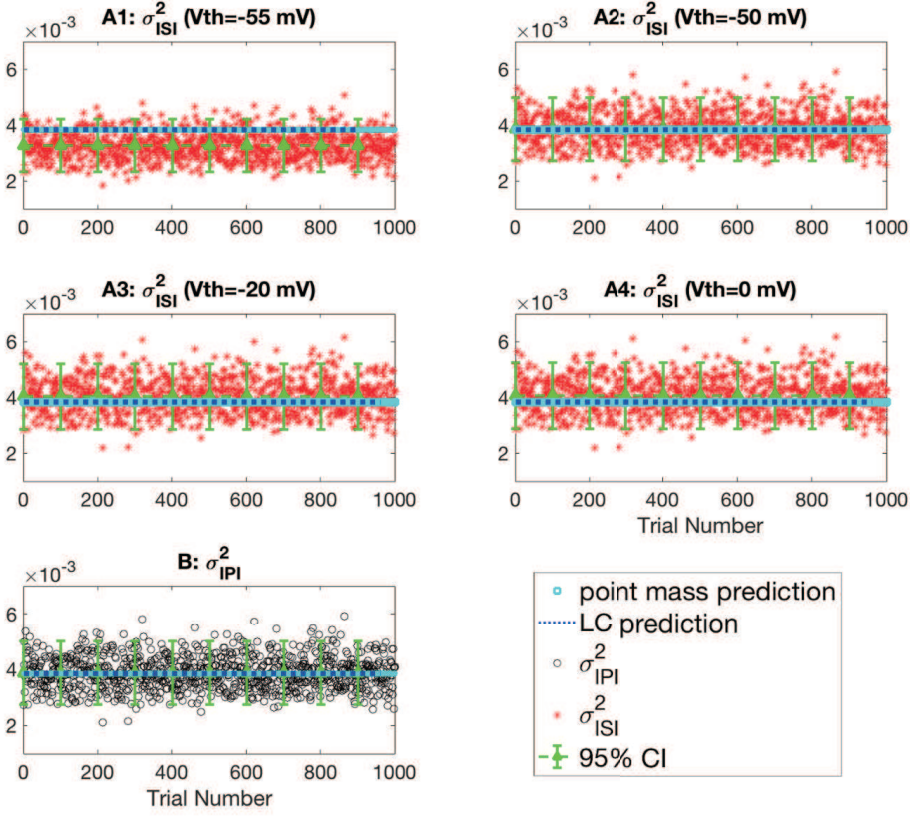}
\caption[Comparison of variance of ISIs ($\sigma^2_\text{ISI}$) and IPIs ($\sigma^2_\text{IPI}$)]{Variance of ISIs ($\sigma^2_\text{ISI}$) and IPIs ($\sigma^2_\text{IPI}$) with only \K~edges included using the stochastic shielding method. 
Cyan dots: point mass prediction (cf.~eqn.~\eqref{eq:pmp}). Solid blue line: limit cycle prediction.
1000 repeated simulations are plotted and for each of the simulation, more than 1000 ISIs (or IPIs) are recorded. 
Each sample point in the plot represents the variance of ISIs and IPIs for one realization.
\textbf{A1-4}: Voltage threshold $V_\text{th}=-55, -50, -20, 0$ mV, with noise level $\sqrt{\epsilon}=0.028$ (effective number of \K~channels $N_\text{tot}\approx2.30\times 10^{6}$).
\textbf{B}: Iso-phase section is the mean-return-time isochron  intersecting the deterministic limit cycle at $V_\text{th}=-50$ mV. }\label{fig:ISI_K_vth}
\end{figure}

In this section, we will apply estimation methods based on Theorem~\ref{theorem:noise} and Corollary~\ref{corollary} to the decomposition of variance of interspike intervals (ISIs, $\sigma^2_\text{ISI}$) and variance of inter-phase intervals (IPIs, $\sigma^2_\text{IPI}$), and numerically test their performance.

Fig.~\ref{fig:ISI_K_vth} presents a detailed comparison of the predicted and measured values of $\sigma^2_\text{ISI}$ and $\sigma^2_\text{IPI}$, when the simulations only include noise from the~\K~channels. 
The channel noise generated by the~\Na~edges is suppressed by applying the stochastic shielding (SS) method to eqn.~\eqref{eq:langevin-rescaled}.
For each plot in Fig.~\ref{fig:ISI_K_vth}, 1000 repeated trials are collected and each trial simulates a total time of 15,000 milliseconds which generates more than 1000 ISIs or IPIs.
Given our previous observation that $\sigma^2_\text{ISI}$ depends on the choice of voltage threshold, we selected four different voltage thresholds for  comparison.

In Fig.~\ref{fig:ISI_K_vth}, red dots in panels A1-A4 mark the ISI variance measured directly from simulated voltage traces, using the indicated $V_\text{th}$ as the trigger voltage. 
Green stem-and-line marks show the mean and 95\% confidence intervals of the direct ISI variance measure, calculated from all 1000 samples.  
The blue dotted line shows the ISI variance predicted from the limit cycle based estimate of the IPI variance (eq.~\eqref{eq:IPI_final}),
and cyan squares show individual estimates using the point-mass prediction (eq.~\eqref{eq:point-mass-prediction}).  
Note each point mass is an independent random variable; these estimates cluster tightly around the limit cycle based estimate.  
Panel B shows the variance of the inter-phase intervals calculated directly from the same 1000 trajectories (as described below), marked in black circles.  
Green stem-and-line marks show the mean and 95\% C.I.~for the IPI variance.  
The blue dotted line and cyan squares represent the same LC-based and point mass based IPI variance estimates as in A1-A4.

As shown in Fig.~\ref{fig:ISI_K_vth} (A1, A2, A3, A4 and B)  the point mass prediction and the LC prediction of the IPI variance give almost the same result.
Specifically, the LC prediction  $\approx3.84\times 10^{-3}$ and the mean of the point mass predictions $\approx3.83\times 10^{-3}$ with a variance of $\approx6.3\times 10^{-11}$. 
Therefore, the LC prediction based on Corollary \ref{corollary} gives a good approximation to the point mass prediction based directly on Theorem \ref{theorem:noise}.
For a given edge (or a group of edges) the LC prediction depends linearly on the scaling factor, $\epsilon$, and can be easily calculated for various noise levels. 
Throughout the rest of this section, we will use the LC prediction as our predicted contribution from the decomposition theorem.

The asymptotic phase is calculated using equation \eqref{eq:pmp} for each point on the stochastic trajectory.
For a given voltage threshold, $V_\text{th}$,
the corresponding iso-phase section is the mean-return-time isochron intersecting the deterministic limit cycle at $V_\text{th}$.
As previously observed, the variance of the IPIs is invariant with respect to the choice of the reference iso-phase section.
As shown in Fig.~\ref{fig:ISI_K_vth} B, 
the prediction of variance of IPIs ($\approx3.83\times 10^{-3}$ {$\text{ms}^2$}) has a good match with the  mean value of numerical simulations ($\approx3.85\times 10^{-3}$ {$\text{ms}^2$}).
The $95\%$ confidence interval of the IPIs are also plotted in  Fig.~\ref{fig:ISI_K_vth} B, which further indicates the reliability of the prediction.

As shown in Fig.~\ref{fig:ISI_K_vth} (A1, A2, A3 and A4), with $V_\text{th}\in[-55,0]$ mV, the numerical realizations of $\sigma^2_\text{ISI}$ are close to the predictions from the main theorem.
However, the accuracy depends on the choice of the voltage threshold.
As noted above, when the trigger voltage $V_\text{th}$ is set below $-50$mV (for example, $-55$mV in Fig.~\ref{fig:ISI_K_vth},A1), the measured variance of ISIs falls below the value predicted from the IPI variance.
When $V_\text{th}\approx-50$mV, the empirically observed value of $\sigma^2_\text{ISI}$ gives the best match to the IPI variance (cf.~Fig.~\ref{fig:ISI_change},C).
When the trigger voltage $V_\text{th}$ exceeds $-50$mV (for example, $-20$mV in Fig.~\ref{fig:ISI_K_vth},A3, and $0$mV in Fig.~\ref{fig:ISI_K_vth},A4), the empirically observed variance of the ISIs is consistently higher than the IPI variance.
Nevertheless, although the empirically observed numerical values of $\sigma^2_\text{ISI}$ ($\approx 4.00\times10^{-3} \text{ms}^2$) overestimate the IPI-derived value when $V_\text{th}>-50$mV, they remain close to the IPI value.
Fig.~\ref{fig:ISI_K_vth} panels A1-4 show that even though the IPI-based prediction of the ISI variance works best when the trigger voltage is set to  $V_\text{th}\approx -50$mV, the IPI-based variance falls within the  $95\%$ confidence interval of $\sigma^2_\text{ISI}$ regardless of the value of $V_\text{th}$ chosen. 
Therefore, we can conclude that Theorem~\ref{theorem:noise} and Corollary~\ref{corollary} give a good approximation to the value of $\sigma^2_\text{ISI}$, at least at noise level $\sqrt{\epsilon}=0.028$.

Practically, the voltage-based interspike interval variance, $\sigma^2_\text{ISI}$, is a more widely used quantity \cite{GutkinErmentrout1998NC,Lindner2004PRE_interspike,Netoff2012PRCN,Shadlen1998JNe,Stein1965BioJ} because it can be calculated directly from  electrophysiological recordings. 
The inter-phase interval variance, $\sigma^2_\text{IPI}$, however, can not be directly measured or calculated.
Even given the stochastic model with its realizations, calculating the asymptotic phase and finding the IPIs are numerically expensive. 
Despite its lack of consistency, as shown in Fig.~\ref{fig:ISI_K_vth} (A3 and A4), the $\sigma^2_\text{ISI}$ can approximately be decomposed using Theorem~\ref{theorem:noise} and Corollary~\ref{corollary}, which offer predicted values that fall in the 95\% confidence interval of $\sigma^2_\text{ISI}$.

\begin{figure}[htbp]
   \centering
\includegraphics[width=1\textwidth]{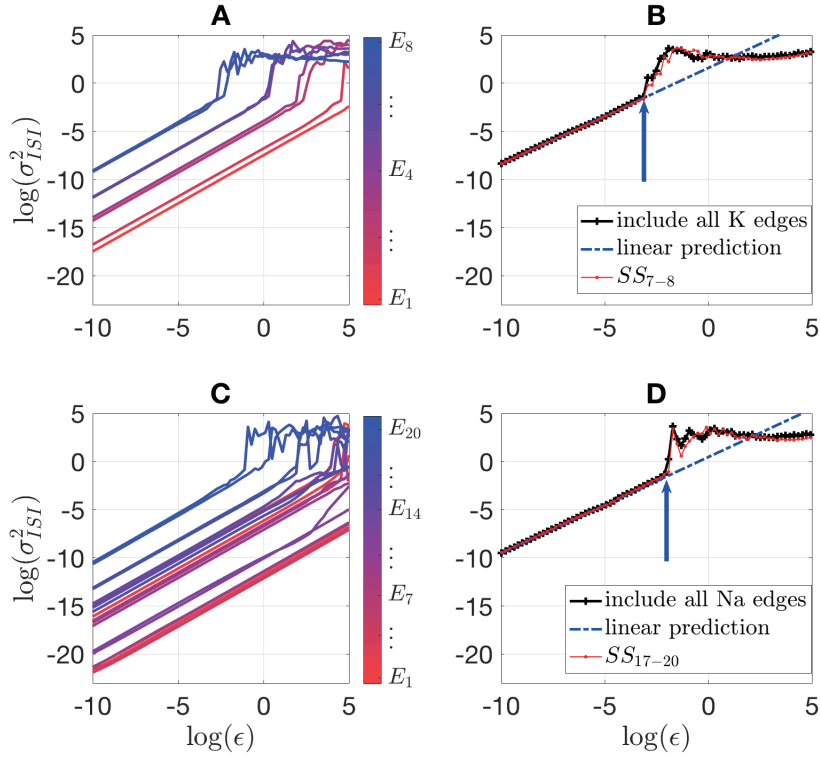}
\caption[Numerical performance of the decomposition for the ISI variance for \Na~and \K~kinetics]{Numerical performance of the decomposition for the ISI variance for \Na~and \K~kinetics.  
\textbf{A\&C}: $\sigma^2_{\text{ISI}}$ for individual \K~(\textbf{A}) and \Na~(\textbf{C}) edges.  
$E_k$ marks the $k$th edge, $k\in\{1,\ldots,8\}$ for \K~ and $k\in\{1,\ldots,20\}$ for \Na.
\textbf{B\&D:} Linearity of superposition for \K~(\textbf{B}) and \Na~(\textbf{D}) channels. 
See text for details.
$\sigma^2_\text{ISI}$ with noise from all \K~(\Na) edges included (black line), with only the most significant edges using stochastic shielding (red line) included, and the linear prediction from Theorem \ref{theorem:noise} (blue dashed line).
 }\label{fig:ISI_Na_K_all}
\end{figure}

Fig.~\ref{fig:ISI_Na_K_all} summarizes the overall fit of the decomposition of variance of ISIs to the prediction from Theorem~\ref{theorem:noise} and Corollary~\ref{corollary}.
We applied the stochastic shielding method by including each directed edge separately in the transition graph (cf.~Fig.~\ref{plot:HHNaKgates}).
In Fig.~\ref{fig:ISI_Na_K_all} (B and D), the variance of the ISIs is compared with the  value obtained with the limit cycle based prediction from eqn.~\eqref{eq:IPI_final}. 

Fig.~\ref{fig:ISI_Na_K_all} (A and C) shows the log-log plot for the ISI variance ($\sigma^2_\text{ISI}$) of each individual edge as a function of the noise level, $\epsilon$, in the range of $[e^{-10},e^{5}]$, measured via direct numerical simulation using $V_\text{th}=-20$ mV. 
The color for each edge ranges from red to blue according an ascending order of edge numbers (1-8 for \K~and 1-20 for \Na). 
The total effective number of \Na~channels is $\mtot=M_\text{ref}/\epsilon$ and of \K~channels is   $\ntot=N_\text{ref}/\epsilon$, where the reference channel numbers are $M_\text{ref}=6000$ and $N_\text{ref}=1800$ (described in~\S\ref{subsec:small_noise_expan}).
That is, we consider ranges of channel numbers $40\lesssim \mtot\lesssim 1.3\times 10^{8}$ for \Na~and $12\lesssim\ntot\lesssim 4.0\times 10^7$ for \K.
Thus, we cover the entire range of empirically observed single-cell channel populations (cf.~Tab.~\ref{tab:Na_K_channel_num}).

As shown in Fig.~\ref{fig:ISI_Na_K_all} (A and C), the linear relation between $\sigma^2_\text{ISI}$ and $\epsilon$  predicted from Theorem \ref{theorem:noise} is numerically observed for all 28 directed edges in the \Na~and~\K~transition graphs (cf.~\ref{plot:HHNaKgates}) for small noise.
The same rank order of edge importance discussed in \cite{PuThomas2020NECO} is also observed here in the small noise region.
Moreover, the smaller the edge importance measure for an individual edge, the larger the value of $\epsilon$ before observing a breakdown of linearity. 

Fig.~\ref{fig:ISI_Na_K_all} (B and D) presents the log-log plot for the ISI variance ($\sigma^2_\text{ISI}$, black solid line) when including noise only from the \K~edges and \Na~edges, respectively. 
As in panels A and C, the noise level, $\epsilon$ is in the range of $[e^{-10},e^{5}]$.
The LC prediction for eqn.~\eqref{eq:IPI_final} from Theorem~\ref{theorem:noise} when including noise from only the \K~(or \Na)~channels is plotted in dashed blue. 
For example, the linear noise prediction for the potassium channels alone is
\begin{equation}
  \sigma^2_\text{ISI}\approx \sum_{k=1}^\mathcal{E_\text{K}}\sigma^2_{\text{ISI},k}  
\end{equation}
where $\mathcal{E}_\text{K}=8$ (similarly, $\mathcal{E}_\text{Na}=20$), and $\sigma^2_{\text{ISI},k}$ is the LC prediction for the $k^{th}$ edge. 
As shown in Fig.~\ref{fig:ISI_Na_K_all} panel B, the linear prediction matches well with the numerically calculated  $\sigma^2_\text{ISI}$ up to $\ln(\epsilon)\approx-3.0$ (indicated by the blue arrow) which corresponds to  approximately 36,000 \K~channels.
For \Na, the theorem gives a good prediction of the numerical $\sigma^2_\text{ISI}$ up to $\ln(\epsilon)\approx-1.9$ (indicated by the blue arrow) which corresponds to approximately 40,000 \Na~channels.
These channel population sizes are consistent with typical single-cell ion channel populations, such as the population of \Na~channels in the node of Ranvier, or the \Na~and \K~channels in models of the soma of a cerebellar Purkinje cell (cf.~Tab.~\ref{tab:Na_K_channel_num}).  

Finally, we apply stochastic shielding (SS) to both the \K and \Na channels by only including noise from
the edges making the largest contributions in Fig.~\ref{fig:ISI_Na_K_all} panels A and C.
For the \K~channel, we include edges 7 and 8, and for \Na, we include edges 17, 18, 19 and 20.
As shown in Fig.~\ref{fig:ISI_Na_K_all} panels B and D, the SS method (solid red line) gives a good match to the overall $\sigma^2_\text{ISI}$ for all noise intensities  $\epsilon\in[e^{-10},e^{5}]$, with numbers of \K~channels $\ge12$ and \Na~channels $\ge40$.

\begin{figure}[htbp] 
   \centering
\includegraphics[width=1\textwidth]{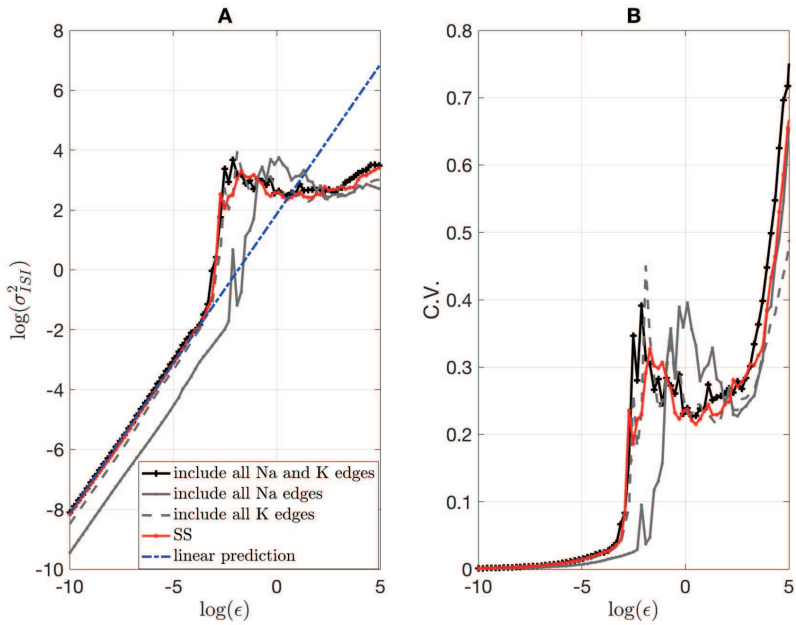}
\caption[Overall performance of the decomposition theorem on ISI variance]{Numerical performance of the decomposition for the ISI variance of the full system. 
\textbf{A:} Log-log plot of $\sigma^2_\text{ISI}$ for $\epsilon\in[e^{-10},e^{5}]$.
ISI variance contribution $\sigma^2_{\text{ISI}}$ with noise from all 28 edges included (black line), only from 8 \K~edges (dashed grey), only from 20 \Na~edges (solid grey), and SS using noise from six  edges (red line, see text for details). 
The linear prediction from Theorem \ref{theorem:noise} for the whole system is plotted for comparison (dashed blue line.) \textbf{B:} Coefficient of variation (C.V.), or mean ISI divided by $\sqrt{\sigma^2_\text{ISI}}$, vs.~$\log(\epsilon)$.
Same color scheme as \textbf{A}.
Compare Fig.~\ref{fig:WTLE_stats}, which shows data from two cerebellar Purkinje cells, a wild-type cell with C.V.~$\approx0.039$ and a cell from a \emph{leaner} mouse with C.V.~$\approx0.30$.
}\label{fig:all_SS}
\end{figure}

Fig.~\ref{fig:all_SS} shows the overall performance of the prediction of 
$\sigma^2_\text{ISI}$
based on Theorem \ref{theorem:noise},  when noise from all 28 directed edges are included (black line).
The theorem is stated as an asymptotic result in the limit of weak noise. 
The predicted ISI variance using the theorem (dashed blue curve) matches the ISI variance obtained from the full numerical simulation for modest noise levels, up to $\ln(\epsilon)\le-3.9$,  corresponding to $\gtrsim$ 90,000 \K~channels and $\gtrsim$ 300,000 \Na~channels.
These population sizes are at the high end of the range of typical numbers of channels neurons (cf.~Tab.~\ref{tab:Na_K_channel_num}).  

For smaller ion channel populations (larger noise levels), the linear approximation breaks down, but the stochastic shielding approximation remains in good agreement with full numerical simulations.  Fig.~\ref{fig:all_SS} shows $\sigma^2_\text{ISI}$ from simulations using the SS method including only noise from the six most important edges (edges 7-8 in \K~and 17-20 in \Na),  plotted in solid red. 
For $\ln(\epsilon)\gtrsim-3.5,$ both the full simulation and the SS simulation show a rapid increase in $\sigma^2_\text{ISI}$ with increasing noise level.  
This dramatic increase in timing variability results when increasing noise causes the neuron to ``miss" spikes, that is, to generate a mixture of regular spiking and small subthreshold oscillations \cite{RowatGreenwood2011NECO}.
Including noise from all 20 \Na~channel edges (gray line) or all eight \K~channel edges (gray dashed line) shows a similar jump, albeit delayed to higher values of $\epsilon$ for the \Na~channel.  
Note also the \Na~channel alone has a quantitatively smaller contribution to ISI variability for the stochastic HH model than the \K~channel for all noise levels in the linear region.  

For larger noise levels ($\ln(\epsilon)\gtrsim-2$), all simulations become sufficiently noisy that $\ln(\sigma^2_\text{ISI})$ collapse to a similar level, approximately 3.  
As the interspike interval is a nonnegative random quantity with a constrained mean (bounded by the reciprocal of the firing rate), once the spike train has maximal variability, further increasing the strength of the channel noise does not drive the ISI variance appreciably higher.
However, although the ISI variance appears approximately to saturate with increasing noise, the coefficient of variation (C.V., $\sqrt{\sigma^2_\text{ISI}}/I$) continues to increase (Fig.~\ref{fig:all_SS}B), because the mean ISI ($I$) decreases with increasing noise (the firing rate increases with increasing noise, data not shown).

\section{Discussion}
\label{sec:discussion}
We prove in \S\ref{sec:noise_decomp} that the numerically calculated edge importance can be quantified from the molecular-level fluctuations of the stochastic Hodgkin-Huxley (HH) kinetics. Specifically, we combine the stochastic shielding approximation with the re-scaled Langevin models (eqn.~\eqref{eq:langevin-rescaled}) of the HH model to derive analytic results for decomposing the variance of the cycle time (the iso-phase intervals) for mean--return-time isochrons of the stochastic HH models.
We prove in theory, and confirm via numerical simulations, that in the limit of small noise, the variance of the iso-phase intervals decomposes linearly into a sum of contributions from each edge.
We show numerically that the same decomposition affords an efficient and accurate estimation procedure for the interspike intervals, which are experimentally observable.
Importantly, our results apply to current clamp rather than to voltage clamp conditions. 
Under current clamp, a stochastic conductance-based model is an example of a piecewise-deterministic Markov process (PDMP).
We show in \S\ref{sec:num_perf_theorem} that our theory is exact in the limit of small channel noise (equivalently, large ion channel population size); through numerical simulations we demonstrate its applicability even in a range of small to medium noise levels, consistent with experimentally inferred single-cell ion channel population sizes. 
In addition, we present the numerical performance of the SS method under different scenarios and argue that the stochastic-shielding approximation together with the $14\times 28$D Langevin representation give an excellent choice of simulation method for ion channel populations spanning the entire physiologically observed range.

Our $14\times 28$D Langevin model (eqn.~\eqref{eq:langevin-rescaled}) can be shown to be \emph{pathwise equivalent} to a class of Langevin models on a 14D state space \cite{PuThomas2020NECO}.  
The first such model was introduced by Fox and Lu \cite{FoxLu1994PRE} and subsequently investigated by \cite{GoldwynSheaBrown2011PLoSComputBiol}.
Pathwise equivalence of two models implies that they have the same distribution over sample paths, hence identical moments including moments related to first-passage and return times.  
One could undertake the same investigation into the variability of spike timing as in this paper using Fox and Lu's formulation, however the $14\times 28$D representation lends itself to an elegant application of the stochastic shielding approximation \cite{SchmandtGalan2012PRL,SchmidtThomas2014JMN,SchmidtGalanThomas2018PLoSCB} that would be cumbersome to apply to other formulations.  
Moreover, as shown in \cite{PuThomas2020NECO}, the $14\times 28$D formulation is at least as fast or faster than its pathwise equivalent alternatives, while having (necessarily) the same accuracy (cf.~Fig~\ref{fig:all_SS}).  
Thus we concur with the assessment of \cite{OrioSoudry2012PLoS1} that the best combination of speed and accuracy for Langevin-type simulation of stochastic conductance based models is given by the diffusion approximation simulations \cite{OrioSoudry2012PLoS1}, while we treat each edge as an independent noise source, combined with stochastic shielding.  

Initially stochastic shielding was introduced for both voltage clamp and current clamp scenarios \cite{SchmandtGalan2012PRL}, but rigorous investigation of the method \cite{SchmidtThomas2014JMN,SchmidtGalanThomas2018PLoSCB} were restricted to voltage clamp.  
Stochastic conductance-based models under current clamp comprise hybrid or piecewise deterministic systems (cf.~Fig.\ref{plot:HHNaKgates}~C), while systems under fixed-voltage-clamp correspond to time-invariant discrete-state Markov chains, for which the theory is well established \cite{SchmidtThomas2014JMN}.
\subsection{Number of Channels in Different Cell Types}
\begin{table}[htbp]
    \centering
\begin{tabular}{ |l|l|l|l| }
\hline
\multicolumn{4}{ |c| }{Estimated numbers of \Na~and \K~channels in different cell types}\\
\hline
Ion & Type of cell& Number of channels & Reference \\ \hline
\multirow{7}{*}{\Na} & chromaffin cells & 1,800-12,500 & \cite{FenwickJPhy1982,TNP2007}$^a$\\
 & cardiac Purkinje cells & $\gtrsim$325,000 & \cite{MakielskiBJ1987}$^b$\\
 & node of Ranvier & 21,000-42,000 & \cite{Sigworth1977Nature}$^c$\\
& squid axon $(1 mm)^d$& $\gtrsim$18,800 & \cite{Faisal2005CB}$^d$  \\ 
& pyramidal cell& $\gtrsim$17,000   & \cite{Faisal2005CB}$^d$ \\
&Purkinje cell$^g$ &  47,000-158,000& \cite{Forrest2015BMC,Shingai1986BR}$^{d,f,g}$\\
&pre-B\"otC neurons$^{h}$ & 56-5,600&\cite{Butera1999JNeuro,Faisal2005CB}$^{d,f,h}$\\\hline
\multirow{4}{*}{\K} & squid axon $(1 mm)^d$& $\gtrsim$5,600 & \cite{Faisal2005CB}$^d$  \\ 
& pyramidal cell& $\gtrsim$2,000   & \cite{Faisal2005CB}$^d$ \\
&Purkinje cell$^g$ & 3,000-55,000& \cite{Forrest2015BMC,Shingai1986BR}$^{d,e,g}$ \\
&pre-B\"otC neurons$^{h}$ & 112-2,240&\cite{Butera1999JNeuro,Faisal2005CB}$^{d,e,h}$\\
\hline
\end{tabular}
    \captionsetup{singlelinecheck=off}
\caption[Number of \Na~and \K~channels]{Details of the data sources:
  \begin{enumerate}[label=(\alph*)]
    \item \Na~density: 1.5-10 channels/$\mu  m^2$\cite{FenwickJPhy1982}, the average diameter of rounded chromaffin cells is $d\approx20\mu m$, Area=$\pi d^2$ \cite{TNP2007}. 
    \item \Na~density: 260 channels/$\mu m^2$ \cite{MakielskiBJ1987}, and diameter of roughly $20\mu m$ \cite{MakielskiBJ1987}. 
    \item Number of \Na~channels in Tab.~1 from \cite{Sigworth1977Nature}.
    \item \Na~density: 60 channels/$\mu m^2$ in squid axon, and 68 channels/$\mu m^2$ in pyramidal axon (Tab.~S1 in \cite{Faisal2005CB}). 
    \K~density: 18 channels/$\mu m^2$ in squid axon, and 8 channels/$\mu m^2$ in pyramidal axon (Tab.~S1 in \cite{Faisal2005CB}). 
    Membrane area: squid axon: 0.1 $\mu m$ diameter and 1$mm$ length (Fig.~S2 in \cite{Faisal2005CB}); pyramical cell: $0.08 \mu m$ diameter with 1 $mm$ length (Fig.~S1 in \cite{Faisal2005CB}). Single voltage-gated ion channel conductance is typically in the range of 5-50 pS, and 15-25 pS for \Na~(p.~1148 \cite{Faisal2005CB}).
    \item Single \K~channel conductance (\cite{Shingai1986BR}): inward rectifier in horizontal cells (20-30 pS in 62-125 mM \K, 9-14$^\circ$C); skeletal muscle (10 pS in 155 mM \K, 24-26$^\circ$C);
    egg cells ($\approx$6 pS for 155 mM \K, 14-15$^\circ$C); heart cells (27 pS for 145 mM \K~at 17-23$^\circ$C; 45 pS for 150 mM \K~at 31-36$^\circ$C).
    \item Single \Na~channel conductance is $\approx$14 pS in squid axon, other measurements under various conditions show results in the range of 2-18 pS (Tab.~1 in \cite{Bezanilla1987BJ}).
    \item Maximal conductance for different \K~channels (Tab.~1 in \cite{Forrest2015BMC}): SK \K~(10 mS/$cm^2$), highly TEA \K~(41.6 mS/$cm^2$) sensitive BK \K~(72.8 mS/$cm^2$); membrane area (1521 $\mu m^2$). Maximal conductance for resurgent \Na~(156 mS/$cm^2$).
    Note that the range of \K~channels provided here is for each type of \K~channel, not the total number of \K~channels.
    \item Maximal conductance ($\bar{g}_\text{ion}$) in pacemaker cells of the pre-B\"otzinger complex (pre-B\"otC)  \cite{Butera1999JNeuro}:
    $\bar{g}_\text{NaP}=2.8$ nS for persistent \Na~current, $\bar{g}_\text{Na}=28$ nS for fast \Na~current, and $\bar{g}_\text{K}\in[5.6,11.2]$ nS for different types of \K~channels (p.~384-385).
  \end{enumerate}
  }
    \label{tab:Na_K_channel_num}
\end{table}
Channel noise arises from the random opening and closing of finite populations of ion channels embedded in the cell membranes of individual nerve cells, or localized regions of axons or dendrites.  
Electrophysiological and neuroanatomical measurements do not typically provide direct measures of the sizes of ion channel populations.  
Rather, the size of ion channel populations must be inferred indirectly from other measurements.  
Several papers report the density of sodium or potassium  channels per area of cell membrane \cite{Faisal2005CB,FenwickJPhy1982,MakielskiBJ1987}.  
Multiplying such a density by an estimate of the total membrane area of a cell gives one estimate for the size of a population of ion channels. 
Sigworth \cite{Sigworth1977Nature} pioneered statistical measures of ion channel populations based on the mean and variance of current fluctuations observed in excitable membranes, for instance in the isolated node of Ranvier in axons of the frog.  
Single-channel recordings \cite{Neher1976Nature} allowed direct measurement of the ``unitary", or single--channel-conductance, $g^o_\text{Na}$ or $g^o_\text{K}$.
Most conductance-based, ordinary differential equations  models of neural dynamics incorporate  maximal conductance parameters ($\overline{g}_\text{Na}$ or $\overline{g}_\text{K}$) which nominally represents the conductance that would be present if all channels of a given type were open.  
The ratio of $\overline{g}$ to $g^o$ thus gives an indirect estimate of the number of ion channels in a specific cell type.
Tab.~\ref{tab:Na_K_channel_num} summarizes a range of estimates for ion channel populations from several sources in the literature.  
Individual cells range from 50 to 325,000 channels for each type of ion.  
In \S \ref{sec:num_perf_theorem} of this thesis, we will consider effective channel populations spanning this entire range (cf.~Figs \ref{fig:ISI_K_vth}, \ref{fig:ISI_Na_K_all}).

\subsection{Different Methods for Defining ISIs}
There are several different methods for detecting spikes and quantifying interspike intervals (ISIs). 
In one widely used approach \cite{Gerstein1960Science,GutkinErmentrout1998NC,Lindner2004PRE_interspike,Mukhametov1970WOL,Netoff2012PRCN,Shadlen1998JNe,Stein1965BioJ}, we can define the threshold as the time of upcrossing a fixed voltage, which is also called a Schmitt trigger (after O.H. Schmitt \cite{Schmitt1938IOP}). 
We primarily use this method in this thesis.

As an alternative, the time at which the rate of change of voltage, $dV/dt$, reaches its maximum value (within a given spike) has also been used as the condition for detecting spikes \cite{Azouz2000NAS}.
However, in contrast with the voltage-based Schmitt trigger, using the maximum of $dV/dt$ to localize the spike does not give a well-defined Poincar\'e section.  
To see this, consider that for a system of the form \eqref{eq:langevin-rescaled} we would have to set 
\begin{align}
    \label{eq:d2Vdt2}
    \frac{d^2V}{dt^2}&=\frac{d}{dt}f(V,\mbN)\\
    &=f(V,\mbN)\frac{\partial f}{\partial V}(V,\mbN) - \frac{dM_8}{dt} g_\text{Na}(V-V_\text{Na}) - \frac{dN_5}{dt} g_\text{K}(V-V_\text{K})
    \nonumber
\end{align}
equal to zero to find the corresponding section.  
The difficulty is evident: for the Langevin system the open fraction $M_8$ (resp.~$N_5$) of sodium (resp.~potassium) channels is a diffusion process, and is not differentiable, so ``$dM_8/dt$'' and ``$dN_5/dt$" are not well defined.  
Moreover, even if we could interpret these expressions, the set of voltages $V$ and gating variables $\mbN$ for which \eqref{eq:d2Vdt2} equals zero depends on the instantaneous value of the noise forcing, so the corresponding section would not be fixed within the phase space.  
For a discrete state stochastic channel model, the point of maximum rate of change of voltage could be determined post-hoc from a trajectory, but again depends on the random waiting times between events, and so is not a fixed set of points in phase space.
For these reasons we do not further analyze ISIs based on this method of defining spikes, although we nevertheless include numerical ISI variance based on this method, for comparison (see Fig.~\ref{fig:dvdt1} below).    

As a third possibility, used for example in  \cite{GutkinErmentrout1998NC}, one sets the voltage nullcline ($dV/dt=0$), at the top of the spike, as the Poincar\'e section for spike detection. 
That is, one uses a surface such as $\mathcal{S}_\text{peak}=\{(v,\mbn)\given f(v,\mbn)=0\}\cap\{v>-40\}$.
This condition does correspond to a well-defined Poincar\'e section, albeit one with a different normal direction than the voltage-based sections.  

In contrast to the ISI variance, which depends to some degree on the choice of spike-timing method used, the \emph{mean} ISI is invariant.  
Both in numerical simulations and from experimental recordings, the mean interspike interval using any of the three methods above is very stable. 
But the apparent ISI variance changes, depending on the method chosen.

We observe in both real and simulated voltage traces that the ISI variance, $\sigma^2_\text{ISI}$, depends not only on the method for identifying spikes, but also on the voltage used for the Schmitt trigger.  
To our knowledge this sensitivity of ISI variance to trigger voltage has not been previously reported. 
Generally speaking, from analyzing both simulation and recorded data from in vitro studies, the ISI variance is not a constant, but increases slightly as the voltage threshold defining a ``spike" is increased (cf.~\S\ref{sec:obers_ISI}). 
Thus the ISI variance is not an intrinsically precisely invariant quantity for model or real nerve cells. 

\begin{figure}[htbp] 
   \centering
\includegraphics[width=1\textwidth]{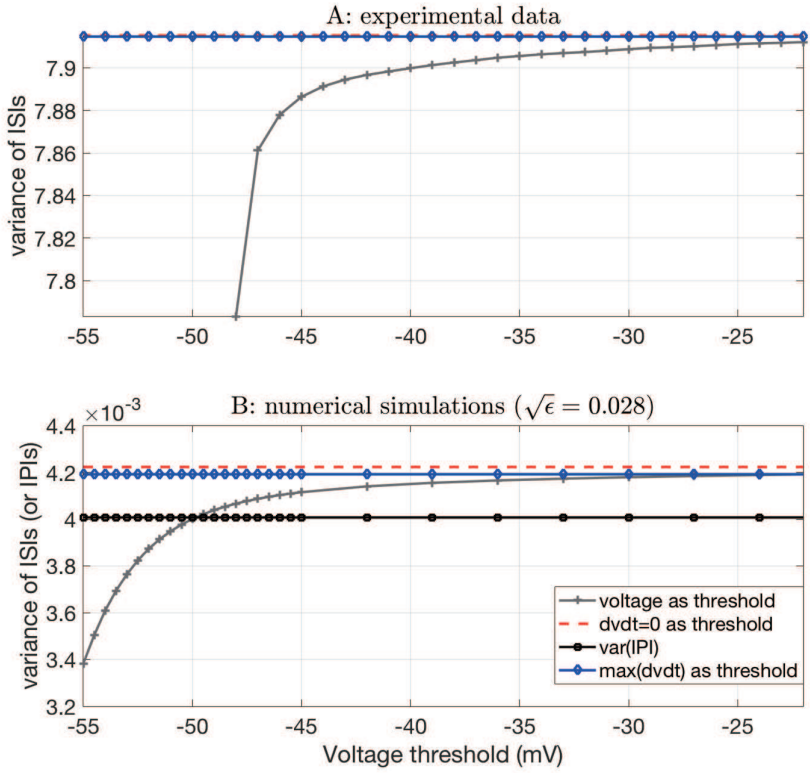}
\caption[Variance of ISI using different threshold conditions]{Variance of interspike intervals using different threshold conditions. 
\textbf{A:} $\sigma^2_\text{ISI}$ of spikes from a single trace of a wild type Purkinje cells comprising 785 ISIs.  $\sigma^2_\text{ISI}=7.9151$ when setting $dv/dt=0$ as the threshold (dashed red), and $\sigma^2_\text{ISI}=7.9146$ when maximum $dv/dt$ is set to be the threshold condition (blue). Different voltage thresholds show increasing $\sigma^2_\text{IS}$ with voltage (gray).  
For experimental methods see \S\ref{sec:experimental-methods}.
\textbf{B:} ISI variance from a Langevin HH (cf.~eq.~\ref{eq:langevin-rescaled}) simulation with small noise ($\sqrt{\epsilon}=0.028$) comprising c.~1000 ISIs.  Labels as in \textbf{A}. The variance of the inter-\emph{phase} intervals is constant regardless of the particular isochron chosen (black).
}\label{fig:dvdt1}
\end{figure}

Fig.~\ref{fig:dvdt1} shows $\sigma^2_\text{ISI}$ obtained empirically from electrophysiological recordings of Purkinje cells \textit{in vitro} (upper plot) and from simulations of the stochastic Hodgkin-Huxley system (lower plot) with a small noise amplitude ($\sqrt{\epsilon}=0.028$) using the three methods for spike time extraction described above, for a single voltage trace comprising 785 interspike intervals. 
The ISI variance as a function of trigger voltage increases steadily from below 7.8 ms$^2$ to above 7.9 ms$^2$ as the trigger voltage increases from -50 mV to -20 mV.
In contrast, the ISI variance obtained from the peak voltage ($dV/dt\approx 0$, obtained using linear interpolation of the first-order voltage difference)
or the maximum slope condition ($d^2V/dt^2\approx 0$ and $dV/dt>0$, obtained using linear interpolation of the second-order voltage difference) give nearly indistinguishable values (red and blue superimposed traces in Fig.~\ref{fig:dvdt1}A) that lie slightly above the largest value of $\sigma^2_\text{ISI}$ at the upper range of the trigger voltage. 

A similar phenomenon occurs for  Langevin simulations of the HH model with small noise (Fig.~\ref{fig:dvdt1}B).  
In this case, the ISI variance based on maximum slope falls slightly below the variance based on the spike peaks; both are similar to the variance obtained with a Schmitt trigger close to $-20$ mV.  
This similarity at higher trigger voltages probably occurs because the inflection point of each spike occurs at nearly the same voltage (at least, for small noise).

As shown in \S\ref{sec:noise_decomp} and \S\ref{sec:num_perf_theorem}, the inter-phase interval (IPI, also refered as iso-phase interval), based on the crossing time of iso-phase sections, provides a uniform $\sigma^2_\text{IPI}$ for all choices of reference iso-phase sections (cf.~Fig.~\ref{fig:ISI_K_vth}).  
Fig.~\ref{fig:dvdt1}B shows the IPI variance (in black) for different mean--return-time isochronal sections, each passing through the limit cycle trajectory at the specified voltage.

For experimental voltage recordings, we cannot specify the interphase variables without a measurement or estimate of the entire state vector.
Fortunately, the sensitivity of ISI variance to voltage threshold, while statistically significant, is relatively small (a few percent), as voltage is the practical measure available for marking spike times.
Moreover, as shown \S\ref{sec:num_perf_theorem}, Theorem \ref{theorem:noise} and Corollary \ref{corollary} can be is well suited to approximating the variance of ISIs ($\sigma^2_\text{ISI}$) despite its threshold-dependence.

\begin{figure}[htbp] 
   \centering
\includegraphics[width=1\textwidth]{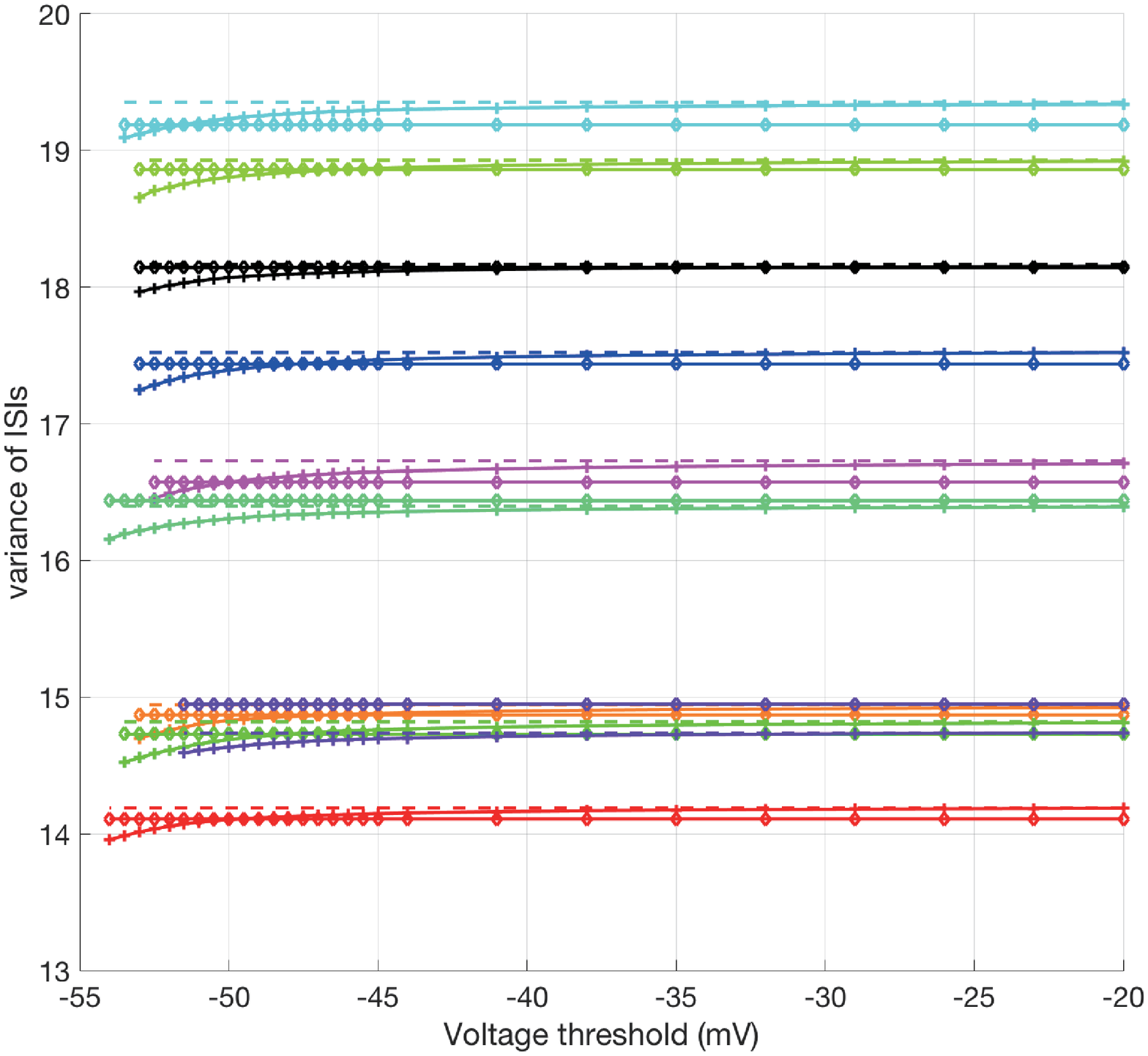}
\caption[Variance of ISI using different thresholds for large noise]{Model simulation of $\sigma^2_\text{ISI}$ for the Langvin model (eqn.~\ref{eq:langevin-rescaled}) using different thresholds when $\epsilon=1$. Ten repeated simulations are plotted, with each containing roughly 1000 ISIs. Solid plus: $\sigma^2_\text{ISI}$ using different voltages as threshold.
Dashed: $dV/dt=0$ as the spike condition.
Diamonds: maximal $dV/dt$ condition. 
Each color represents a different sample with independent noise.  
Variance is in units of ms$^2$.
}\label{fig:dvdt2}
\end{figure}

For moderate to large noise  Langevin model traces ($\epsilon\approx 1$), we also see a systematic shift in $\sigma^2_\text{ISI}$ with increasing Schmitt-trigger voltage. 
However, the size of the shift is an order of magnitude smaller than the variability of the variance across trials.  
Fig.\ref{fig:dvdt2} plots $\sigma^2_\text{ISI}$ versus trigger voltage, as well as the ISI variance based on the peak voltage and the maximal slope conditions, for ten different samples of the Langevin HH model with $\epsilon=1$, each comprising $\gtrsim 1000$ interspike intervals.  
In each case $\sigma^2_\text{ISI}$ is a smoothly increasing function of the trigger voltage, but the range of the increase in variance is approximately  0.25 ms$^2$, while the sample variance of the ISI variance itself is approximately 3.5 ms$^2$ across the ten trials, an order of magnitude larger.  
For comparison, the sample variance of $\sigma^2_\text{ISI}$ across c.~4000 trials, cf.~Fig.~\ref{fig:ISI_K_vth}, is approximately $3\times 10^{-7}$ ms$^2$.
The source of the variance for the larger noise value may involve the introduction of extra  or missing spikes from the regular spiking pattern, cf.~Fig.~\ref{fig:all_SS}.
Thus, although $\sigma^2_\text{ISI}$ based on the standard Schmitt trigger approach is sensitive to the trigger value, the IPI variance estimate given by Theorem \ref{theorem:noise} lies within the range of this sensitivity, which for realistic noise levels is small compared to the intrinsic variability of the variance across trials.  


\subsection{Relation to Other Methods}
Ermentrout and colleagues \cite{ErmentroutBeverlinTroyerNetoff2011JCNS} developed an asymptotic treatment of the interspike interval variance as part of their analysis of the variance of phase response curves; the variance of the phase response in the absence of a perturbing input is simply the ISI variance.
Our analysis was inspired in part by the approach of \cite{ErmentroutBeverlinTroyerNetoff2011JCNS} in that we study the accumulation of variance of the timing variable (the asymptotic phase function for the unperturbed system) over a single period.   
Our approach differs from that of \cite{ErmentroutBeverlinTroyerNetoff2011JCNS} in several key respects.  
While \cite{ErmentroutBeverlinTroyerNetoff2011JCNS} used an additive Gaussian white noise current to obtain stochastic trajectories, in our model the fluctuations arise from channel noise based on a detailed $14\times 28$ D Langevin description of the Hodgkin-Huxley system.  
In addition \cite{ErmentroutBeverlinTroyerNetoff2011JCNS} truncated the small-$\epsilon$ expansion of the phase dynamics at first order, i.e.~they neglected terms at $O(\epsilon^2)$ and higher orders.  
However, their expression for the ISI variance begins with a term that is $O(\epsilon^2)$, suggesting a possibly inconsistency in their result.  
In contrast, we retain terms through $O(\epsilon^2)$, and demonstrate excellent agreement between full numerical and semi-analytic results in the small-$\epsilon$ regime.
Finally, we explicitly construct the ISI variance as a first-passage time problem, allowing us to leverage Dynkin's theorem, which goes beyond the treatment given in \cite{ErmentroutBeverlinTroyerNetoff2011JCNS}.

There is a rich literature on the variability of interspike intervals in low-dimensional neural models, such as integrate and fire models with or without an adaptation current \cite{Lindner2004PRE_interspike,SchwalgerTiloetal2010PCB,Shiau2015JCN,Fisch2012JofNeuro}.
In this literature, there is significant interest in serial correlation structure of successive interspike intervals \cite{Chacron2004PRE,Lindner2005PRE,Schwalger2012EPL}.
However, to our knowledge, the literature since \cite{ErmentroutBeverlinTroyerNetoff2011JCNS} has not further addressed analytical treatment of interspike interval variance for higher dimensional models such as the Hodgkin-Huxley system.  

Jan-Hendrik Schleimer's thesis \cite{Schleimer2013Thesis} addresses the moments of the interspike interval through a different approach.  
In his thesis, Schleimer assumes an a prior reduction from physical coordinates (voltage and gating variables) to a one-dimensional phase description $\phi\in[0,1)$.  
He formulates a stochastic Langevin equation for $\phi$ using Linaro's model \cite{Linaro2011PublicLibraryScience} which is similar (but not pathwise equivalent) to \cite{FoxLu1994PRE,OrioSoudry2012PLoS1,PuThomas2020NECO} of the form (eqn.~(5.1) on page 47)
$$\frac{d\phi}{dt}=\omega+\mbZ(\phi)^\intercal\sqrt{D(\phi)}\xi(t),$$ 
where $\xi$ is a delta-correlated Gaussian white noise disturbance vector, $\mbZ$ is the infinitesimal phase response curve of the deterministic limit cycle system, and $D(\phi)$ is a noise coefficient matrix. 
If $\mbZ$ is $n$-dimensional then $\xi(t)$ is assumed to be $n$-dimensional as well, and $D$ is $n\times n$.    
Using the Stratonovich interpretation, he derives a local phase-specific increment of the timing variability of the form (eqn.~(5.3) on page 47)
$$\sigma^2(\phi)=\mbZ(\phi)^\intercal D(\phi)\mbZ(\phi).$$
Implicitly,$\int_0^1\sigma^2(\phi)\,d\phi$ gives the ISI variability. 
To compare our approach with Schleimer's, we give an explicit expression for $\sigma^2(\phi)$ as a sum of contributions from each directed edge, at each phase of the limit cycle:
\begin{equation}
    \sigma^2(s)=\left(\sum_k \alpha_k(V(S))X_{i(k)}(s)\zeta_k^\intercal \mbZ(s) \mbZ^\intercal(s)\zeta_k\right)\frac1{T_0},
\end{equation}
where the sum runs over directed edges in the ion channel transition graph(s), and the factor of $T_0^{-1}$ reflects changing the range of phase from $[0,T_0]$ to $[0,1]$.


\subsection{Limitations}
Like other approaches in the literature, our calculations are based on a linear approximation to the effects of the noise.  However,
Ito's formula \eqref{eq:ItoLemma} includes terms both of order $\sqrt{\epsilon}$ and $\epsilon$.  The latter weights the Hessian matrix of the asymptotic phase function, $\partial^2_{ij}\phi$.  In our main result \eqref{eq:ISI_channel_noise_decomp} we neglected the contribution of this higher order term.  
    Similar truncations of either Taylor's expansion or Ito's formula are seen throughout the literature, for example eqn.~(3.2.8) in \cite{KuramotoBook}, eqn.~(120) in \cite{Schwemmer2012Springer}, and eqn.~(2) in \cite{ErmentroutBeverlinTroyerNetoff2011JCNS}. 
These authors favor an immediate phase reduction when $\epsilon$ is small,  setting $\mbX(t)\approx \mbX_0(\theta(t))$ and 
\begin{equation}
    \frac{d\phi}{dt}=1+\sqrt{\epsilon} \left(\nabla\phi(\mbX)\right)^\intercal\cdot \mathcal{G}(\mbX)\cdot d\mbW(t)
\end{equation}(\cite{Schwemmer2012Springer,KuramotoBook}).  As in Kuramoto's original phase reduction approach \cite{KuramotoBook}, we also evaluate the infinitesimal phase response curve $\mbZ$ on the limit cycle throughout this thesis.
This ommission of the Hessian term could possible cause additional discrepancies.  Recent advances in the theory of nonlinear oscillators have provided means to obtain the asymptotic phase Hessian \cite{Aminzare2019IEEE,wilson2018greater,wilson2019augmented} but we have not attempted to implement these calculations for our 14D HH model.  

\section{Acknowledgments}
The authors thank Dr.~David Friel (Case Western Reserve University, School of Medicine) for introducing us to the problem of channel noise and spontaneous spike time variability in the \textit{leaner} mouse, for providing access to voltage recordings obtained in his laboratory, and for many illuminating discussions on mechanisms at work in Purkinje cells. 
The authors thank Dr.~Daniela Calvetti and Dr.~Erkki Somersalo (Case Western Reserve University, Department of Mathematics, Applied Mathematics and Statistics) 
for helpful advice and discussion.
This work was made possible in part by grants from the National Science Foundation (DMS-1413770 and DEB-1654989) and the Simons Foundation. PT thanks the Oberlin College Department of Mathematics for research support.  
This research has been supported in part by the Mathematical Biosciences Institute and the National Science Foundation under grant DMS-1440386. 
Large-scale simulations made use of the High Performance Computing Resource in the Core Facility for Advanced Research Computing at Case Western Reserve University.

\begin{appendices}

\section{Model Parameters, Common Symbols and Notations}\label{app:notations}

\begin{table}[htbp]\centering
   {\begin{tabular}{|l|l|r|} 
   \hline
   Symbol  & Meaning & Value \\
   \hline
   $C$ & Membrane capacitance & 1 $\mu F/cm^2$ \\
    \hline
     $\bar{g}_\text{Na}$ & Maximal sodium conductance & 120 $\mu S/cm^2$ \\
     \hline
     $\bar{g}_\text{K}$ & Maximal potassium conductance & 36 $\mu S/cm^2$ \\
     \hline
     $g_\text{leak}$ & Leak conductance & 0.3 $\mu S/cm^2$ \\
   \hline
   $V_\text{Na}$ & 
   Sodium reversal potential for \Na & 50 $mV$\\
   \hline
    $V_\text{K}$ & 
   Potassium reversal potential for \K & -77 $mV$\\
   \hline
    $V_\text{leak}$ & 
   Leak reversal potential & -54.4 $mV$\\
   \hline
   $I_\text{app}$&
   Applied current to the membrane  &10 $nA/cm^2$\\
   \hline
   $\mathcal{A}$ & Membrane Area & $100\,\mu\text{m}^2$\\
   \hline
   $M_\text{tot}$& Total number of \Na~channels & 6,000\\
\hline
$N_\text{tot}$& Total number of\K~channels & 18,00\\
\hline
   \end{tabular}}
   \caption{ {Parameters used for simulations in this paper.}}
   \label{tab:parameters}
\end{table}

 Subunit kinetics for Hodgkin and Huxley parameters are given by
\begin{align}
  \alpha_m(V)&=  \frac{0.1*(25-V)}{ \exp(2.5-0.1V)-1} \label{eq:rate4}  \\
\beta_m(V)&= 4*\exp(-V/18) \label{eq:rate5}  \\
\alpha_h(V)&= 0.07*\exp(-V/20) \label{eq:rate6}  \\
\beta_h(V)&= \frac{1}{ \exp(3-0.1V)+1} \label{eq:rate7}  \\
\alpha_n(V)&= \frac{0.01* (10-V)}{\exp(1-0.1V)-1} \label{eq:rate8}  \\
\beta_n(V)&= 0.125\exp(-V/80) \label{eq:rate9}    
\end{align}


 \[ A_K(V) =\begin{bmatrix}
   D_1(1)& \beta_n(V)             & 0                & 0                  & 0\\
   4\alpha_n(V)& D_1(2)&   2\beta_n(V)              & 0&                   0\\
    0&        3\alpha_n(V)&        D_1(3)& 3\beta_n(V)&          0\\
    0&        0&               2\alpha_n(V)&          D_1(4)& 4\beta_n(V)\\
    0&        0&               0&                 \alpha_n(V)&          D_1(5)
\end{bmatrix},
\]  

 \[ A_\text{Na} =\begin{bmatrix}
D_2(1) & \beta_m&0 &0 &\beta_h&0&0&0\\
3\alpha_m&D_2(2)&2\beta_m&0&0&\beta_h&0&0\\
0&2\alpha_m&D_2(3) &3\beta_m&0&0&\beta_h&0 \\
0&0&\alpha_m&D_2(4)&0&0&0&\beta_h \\
\alpha_h&0&0&0&D_2(5)&\beta_m&0&0\\
0&\alpha_h&0&0&3\alpha_m&D_2(6)&2\beta_m&0\\
0&0&\alpha_h&0&0&2\alpha_m&D_2(7)&3\beta_m\\
0&0&0&\alpha_h&0&0&\alpha_m&D_2(8)\\
\end{bmatrix},
\]  
where the diagonal elements $$D_k(i)=-\sum_{j\neq i}A_\text{ion}(j,i),\quad k\in \{1,2\} \quad \text{ion}\in \{\text{Na},\text{K}\}.$$

\renewcommand{\arraystretch}{1.25}
\begin{table}[htbp]\centering
   \begin{tabular}{cl} 
   \hline\hline
   Symbol  & Meaning \\
   \hline\hline
   $\mbX$ \& $\mbx$ & state vector (random variable \& realization value)
   \\ 
   $\mbM$ \& $\mbm$ & eight-component state vector for the \Na gates (random variable \& realization value)\\
   $\mbN$ \& $\mbn$ & five-component state vector for the \K gates (random variable \& realization value)\\
    \hline \hline
   $C$  & membrane capacitance (1 $\mu F/cm^2$)\\
   $I_\text{app}$  &  applied current (10 $nA/cm^2$)\\
   $\bar{g}_\text{ion}$ &  maximal conductance for $\text{ion}\in\{\text{Na}^+,\text{K}^+\}$   \\
   $V_\text{ion}$&  reversal potential of $\text{ion}\in\{\text{Na}^+,\text{K}^+\}$  \\
    \hline \hline
    $e^\text{Na}_i$ \& $e^\text{K}_i$ & $i$th standard unit vector in $\R^8$ \& $\R^5$ \\
   $\zeta^\text{ion}_k=e^\text{ion}_{j(k)}-e^\text{ion}_{i(k)}$ & stoichiometry vector for the $k$th edge, for $\text{ion}\in\{\text{Na}^+,\text{K}^+\}$  \\
    $\alpha_k(v)$ & voltage-dependent per capita transition rate along $k$th edge, $1\le k \le 28$ \\ $i(k)$ \& $j(k)$ & source \& destination nodes for $k$th edge\\
    $M_{i(k)}$ & fractional occupancy of source node for $k$th edge\\
    $\mbF(\mbX)$ \& $f(x)$ & deterministic part of the evolution equation (mean-field)\\
    $\mathcal{G}$, $S_\text{Na}$, $S_\text{K}$  & noise coefficient matrix for the the $14\times28$D Langevin model, \Na gates and \K gates, resp. \\
     \hline \hline
       $\Delta^k$ & $k$-dimensional simplex in $\R^{k+1}$ given by $y_1+\ldots+y_{k+1}=1, y_i\ge 0$\\
       $\mathcal{D}$ & domain of the (stochastic) differential equation \\
       $\mathcal{V}^0$ & ``nullcline'' surface associated with the voltage variable, where $f(v,\mbm,\mbn)=0$ \\
       $\mathcal{S}$ &  arbitrary section transverse to the deterministic limit cycle \\
       $\mathcal{S}^u$ &  isovoltage Poincar\'{e} section (where voltage is a constant $u$)\\
       $\mathcal{S}^{u}_0$,  $\mathcal{S}^{u}_+$, \&  $\mathcal{S}^{u}_-$ & ``null", ``inward current"  \&  ``outward current" surface for voltage $u$ \\
      &  and $f(v,\mbm,\mbn)=0,\ f(v,\mbm,\mbn)>0 \&\ f(v,\mbm,\mbn)<0$, resp.\\
    \hline \hline
    $\tau(\mbx,\mathcal{S})$ &  first passage time (FPT) from a point $\mbx\in\mathcal{D}$ to section $\mathcal{S}$\\
    $T(\mbx,\mathcal{S})$ & mean first passage time (MFPT) from point $\mbx\in\mathcal{D}$ to set $\mathcal{S}$ \\
    $S(\mbx,\mathcal{S})$ & the second moment of the FPT from a point $\mbx\in\mathcal{D}$ to section $\mathcal{S}$\\
      \hline \hline
      $\tau_k^u$ \& $\tau_k^d$ & $k$th voltage surface upcrossing \& downcrossing time\\
      $I_k$ & $k$th interspike interval (ISI), for some reference voltage $v_0$ \\
      $I$, $H$ \& $\sigma^2_\text{ISI}$ & mean, 2nd moment, and variance of ISI\\
     $\mu_k$ & $k$th iso-phase crossing time\\
  $\Delta_k$ & $k$th iso-phase interval (IPI), for some reference phase $\phi_0$\\
    $\tbar_\epsilon$, $S_\epsilon$, $\sigma^2_\text{IPI}$ & mean, 2nd moment, variance of iso-phase interval (for noise level $\epsilon$)\\
   $\sigma^2_{\phi,k}$ \&  $\sigma^2_{\text{ISI},k}$ & contribution of $k$th edge to the IPI variance and the ISI variance, resp. \\
      \vspace{-3mm}\\
      \hline \hline
      $\gamma(t)$ & deterministic limit cycle trajectory\\
      $\tbar_0$ & period of deterministic limit cycle\\
      $\phi(\mbx)$  & asymptotic phase function for deterministic limit cycle\\
      $\mbZ(t)=\nabla\phi(\gamma(t))$ & infinitesimal phase response curve (iPRC) for deterministic limit cycle  \\
      \hline \hline
    $\tbar_\epsilon$ & mean period for noise level set to $\epsilon$\\
       $\tbar_1=\left.\frac{\partial \tbar_\epsilon}{\partial \epsilon}\right|_{\epsilon=0}$ & sensitivity of the mean period to increasing noise level, in the small-noise limit\\
        \vspace{-3mm}\\
      \hline \hline
    $T_\epsilon(\mbx)$ & mean--return-time (MRT) phase function for noise level set to $\epsilon$ \\
    $T_1(\mbx)=\left.\frac{\partial T_\epsilon(\mbx)}{\partial \epsilon}\right|_{\epsilon=0}$ & sensitivity of the phase function to noise in the small-noise limit\\
    \vspace{-3mm}\\
      $T_0(\mbx)$  & MRT phase function for $\epsilon=0$. Note $T_0(\mbx)=\text{const}-\tbar_0\frac{\phi(\mbx)}{2\pi}$ for an arbitrary constant\\
    \hline\hline
   \end{tabular}
   \caption{Table of Common Symbols and Notations.}
   \label{notations}
\end{table}


\section{Diffusion Matrix of the 14D Model}\label{app:SNaSK}
Define the state vector for \Na~and \K~channels as
\begin{align}
&\mbM=[m_{00},m_{10},m_{20},m_{30},m_{01},m_{11},m_{21},m_{31}]^\intercal,\nonumber 
\end{align}
and $\mbN=[n_0,n_1,n_2,n_3,n_4]^\intercal$, respectively. 

The matrices $S_\text{K}$ and $S_\text{Na}$ are given by
\begin{align*}
S_\text{K}=&\frac{1}{\sqrt{N_\text{ref}}}\left[
\begin{array}{cccc}
-\sqrt{4\alpha_n n_0}& \sqrt{\beta_n n_1}&0&0\\
   \sqrt{4\alpha_n n_0}& -\sqrt{\beta_n n_1}&-\sqrt{3\alpha_n n_1}&\sqrt{2\beta_n n_2} \\
  0 &0 &\sqrt{3\alpha_n n_1}&-\sqrt{2\beta_n n_2}\\
    0 &0&0&0 \\
    0&0&0&0 \\
   \end{array} 
\right.\cdots\\
&\quad\cdots \left. 
\begin{array}{cccc}
0&0&0&0\\
0&0&0&0 \\
\sqrt{2\alpha_n n_2}&-\sqrt{3\beta_n n_3}&-\sqrt{\alpha_n n_3}&\sqrt{4\beta_n n_4} \\
    0&0&\sqrt{\alpha_n n_3}&-\sqrt{4\beta_n n_4} \\
  \end{array}
    \right],
\end{align*}
and
\begin{align*}
S^{(1:5)}_\text{Na}=&\frac{1}{\sqrt{M_\text{ref}}}\left[
\begin{array}{ccccc}
-\sqrt{\alpha_h m_{00}}& \sqrt{\beta_h m_{01}}&-\sqrt{3\alpha_m m_{00}}&\sqrt{\beta_m m_{10}}&
0\\
   0& 0&\sqrt{3\alpha_m m_{00}}&-\sqrt{\beta_m m_{10}}&
   -\sqrt{\alpha_h m_{10}} \\
  0 &0 &0&0&0\\
    0 &0&0&0&0\\
    \sqrt{\alpha_h m_{00}}&-\sqrt{\beta_h m_{01}}&0&0&0 \\
     -\sqrt{\beta_h m_{11}}&0&0&0&0 \\
      0&0&0&\sqrt{\alpha_h m_{20}}&-\sqrt{\beta_h m_{21}} \\
       0&0&0&0&0 \\
    \end{array} 
\right]\\
S_\text{Na}^{(6:10)}=&\frac{1}{\sqrt{M_\text{ref}}}\left[ 
\begin{array}{ccccc}
0&0&0&0&0\\
 \sqrt{\beta_h m_{11}}
   &-\sqrt{2\alpha_m m_{10}}&\sqrt{2\beta_m m_{20}}&0&0 \\
  \sqrt{2\alpha_m m_{10}}
  &-\sqrt{2\beta_m m_{20}}&-\sqrt{\alpha_h m_{20}}&\sqrt{\beta_h m_{21}}\\
    0&0&0 &0&0\\
   0&0&0&0&0 \\
     -\sqrt{\beta_h m_{11}}&0&0&0&0 \\
      0&0&0&\sqrt{\alpha_h m_{20}}&-\sqrt{\beta_h m_{21}} \\
       0&0&0&0&0 \\
    \end{array}
\right]\\
\end{align*}
\begin{align*}
S^{(11:15)}_\text{Na}=&\frac{1}{\sqrt{M_\text{ref}}}\left[ 
\begin{array}{ccccc}
0&0&0&0&0\\
0&0&0&0&0\\
   -\sqrt{\alpha_m m_{20}}&
   \sqrt{3\beta_m m_{30}}&0&0
   &0\\
   \sqrt{\alpha_m m_{20}}&
   -\sqrt{3\beta_m m_{30}}
   &-\sqrt{\alpha_h m_{30}}
   &\sqrt{\beta_h m_{31}}&0\\
 0&0&0&0
   0&\\
0&0&0&0&\sqrt{3\alpha_m m_{01}}\\
 0&0&0&0
   &0\\
   0&0&\sqrt{\alpha_h m_{30}}&-\sqrt{\beta_h m_{31}}
   &0\\
    \end{array}
\right]\\
S^{(16:20)}_\text{Na}=&\frac{1}{\sqrt{M_\text{ref}}}\left[
\begin{array}{ccccc}
0&0&0&0&0\\
0&0&0&0&0\\
0&0&0&0&0\\
0&0&0&0&0\\
-\sqrt{3\alpha_m m_{01}}
   &\sqrt{\beta_m m_{11}}&0&0&0\\
   -\sqrt{\beta_m m_{11}}&-\sqrt{2\alpha_m m_{11}}&\sqrt{2\beta_m m_{21}}&0&0\\
   0&\sqrt{2\alpha_m m_{11}}&-\sqrt{2\beta_m m_{21}}&-\sqrt{\alpha_m m_{21}}&\sqrt{3\beta_m m_{31}}\\
   0&0&0&\sqrt{\alpha_m m_{21}}&-\sqrt{3\beta_m m_{31}}\\
    \end{array}
    \right],
\end{align*}
where $S^{(i:j)}_\text{Na}$ is the i$^{th}$-j$^{th}$ column of $S_\text{Na}$.

Note that each of the 8 columns of $S_\text{K}$ corresponds to the flux vector along a single directed edge in the \K~channel transition graph.
Similarly, each of the 20 columns of $S_\text{Na}$ corresponds to the flux vector along a directed edge in the \Na~graph (cf.~Fig.~\ref{plot:HHNaKgates}). 
Factors $\mref=6000$ and $\nref=1800$ represent the reference number of \K~and \Na~channels from Goldwyn and Shea-Brown's model \cite{GoldwynSheaBrown2011PLoSComputBiol}.

\section{Proof of Lemma \ref{Lemma:vmin_vmax}}\label{append_Lemma_vmin_vmax}
For the reader's convenience we restate

\begin{customlemma}{\ref{Lemma:vmin_vmax}}
For a conductance-based model of the form \eqref{eq:langevin-rescaled}, and for any fixed applied current $I_\text{app}$, there exist  upper and lower bounds $v_\text{max}$ and $v_\text{min}$ such that trajectories with initial voltage condition $v\in[v_\text{min},v_\text{max}]$ remain within this interval for all times $t>0$, with probability 1, regardless of the initial channel state, provided the gating variables satisfy $0\le M_{ij}\le 1$ and $0\le N_i\le 1$.
\end{customlemma}

\begin{proof}
  Let $V_1=\underset{\text{ion}}{\text{min}}\{V_\text{ion}\}\land V_\text{leak}$, and 
$V_2=\underset{\text{ion}}{\text{max}}\{V_\text{ion}\}\lor V_\text{leak}$, where  $\text{ion}\in\{\text{Na}^+, \text{K}^+\}$. 
Note that by assumption, for both the \Na~ and \K~channel, $0\leq \mbM_8 \leq1$, $0\leq \mbN_5 \leq1$. Moreover, $g_i>0,\ g_\text{leak}>0$, therefore  when $V\le V_1$
\begin{align}\label{eq:vmin11}
\frac{dV}{dt}&=\frac{1}{C}\left\{I_\text{app}(t)-\bar{g}_\text{Na}\mbM_8\left(V-V_\text{Na}\right)-\bar{g}_\text{K}\mbN_5\left(V-V_\text{K}\right)-g_\text{leak}(V-V_\text{leak})\right\}\\
\label{eq:vmin21}
&\ge  \frac{1}{C}\left\{I_\text{app}(t)-\bar{g}_\text{Na}\mbM_8\left(V-V_1\right)-\bar{g}_\text{K}\mbN_5\left(V-V_1\right)-g_\text{leak}(V-V_1)\right\}\\
\label{eq:vmin31}
&\ge \frac{1}{C}\left\{I_\text{app}(t)-0\times\mbM_8\left(V-V_1\right)-0\times\mbN_5\left(V-V_1\right)-g_\text{leak}(V-V_1)\right\}\\
\label{eq:vmin41}
&=\frac{1}{C}\left\{I_\text{app}(t)-g_\text{leak}(V-V_1)\right\}.
\end{align} 
Inequality \eqref{eq:vmin21} holds with probability 1 because $V_1 =\underset{i\in \mathcal{I}}{\text{min}}\{V_i\}\land V_\text{leak}$, and  inequality \eqref{eq:vmin31} follows because $V-V_1\le 0$, $g_i>0$ and $\mbM_8\ge 0$, $\mbN_5\ge 0$. 
Let $V_\text{min}:=\text{min}\left\{\frac{I_-}{g_\text{leak}} +V_1,V_1\right\}$.  When $V<V_\text{min}$, $\frac{dV}{dt}>0$. Therefore,  
$V$ will not decrease beyond $V_\text{min}$. 

Similarly, when $V\ge V_2$
\begin{align}\label{eq:vmax11}
\frac{dV}{dt}&=\frac{1}{C}\left\{I_\text{app}(t)-\bar{g}_\text{Na}\mbM_8\left(V-V_\text{Na}\right)-\bar{g}_\text{K}\mbN_5\left(V-V_\text{K}\right)-g_\text{leak}(V-V_\text{leak})\right\}\\
\label{eq:vmax21}
&\le  \frac{1}{C}\left\{I_\text{app}(t)-\bar{g}_\text{Na}\mbM_8\left(V-V_2\right)-\bar{g}_\text{K}\mbN_5\left(V-V_2\right)-g_\text{leak}(V-V_2)\right\}\\
\label{eq:vmax31}
&\le \frac{1}{C}\left\{I_\text{app}(t)-0\times\mbM_8\left(V-V_2\right)-0\times\mbN_5\left(V-V_2\right)-g_\text{leak}(V-V_2)\right\}\\
\label{eq:vmax41}
&=\frac{1}{C}\left\{I_\text{app}(t)-g_\text{leak}(V-V_2)\right\}.
\end{align} 

Inequality \eqref{eq:vmax21} and inequality \eqref{eq:vmax31} holds  because $V_2=\underset{i\in \mathcal{I}}{\text{max}}\{V_i\}\lor V_\text{leak},$  $V-V_2\ge 0$, $g_i>0$ and $\mbM_8\ge 0$, $\mbN_5\ge 0$. 
Let $V_\text{max}=\text{max}\left\{\frac{I_\text{app}}{g_\text{leak}}+V_2, V_2 \right\}$. When $V>V_\text{max},$ $\frac{dV}{dt}<0$.   Therefore, $V$ will not go beyond $V_\text{max}$.  

We conclude that if $V$ takes an initial condition in the interval $[V_\text{min},V_\text{max}],$ and if $0\le M_{ij},N_i\le 1$ for all time, then $V(t)$ remains within this interval for all $t\ge 0$.
 Thus we complete the proof of Lemma \ref{Lemma:vmin_vmax}.
\end{proof}

\section{Experimental Methods}
\label{sec:experimental-methods}
Whole-cell current-clamp recordings of Purkinje cells from \textit{in vitro} cerebellar slice preparations taken from wild type and \emph{leaner} mice were performed in the laboratory of Dr.~David Friel (Case Western Reserve University School of Medicine), as described in \cite{OvsepianFriel2008EJN}.
Experimental procedures conformed to guidelines approved by the
Institutional Animal Care and Use Committee at Case Western
Reserve University. 
Voltage signals were sampled at a frequency of 20kHz, filtered at 5–10 kHz, digitized at a resolution of 32/mV, and analyzed using custom software written in IgorPro and Matlab.

\end{appendices}

\bibliography{Shusen}


\bibliographystyle{spmpsci}      


\end{document}